%% file: coalgpartref-lmcs.tex
\documentclass[final]{lmcs} 
\pdfoutput=1

\usepackage{lastpage}
\lmcsdoi{16}{1}{8}
\lmcsheading{}{\pageref{LastPage}}{}{}%
{Jun.~15,~2018}{Jan.~31,~2020}{}

\pdfoutput=1
\usepackage[utf8]{inputenc}
\usepackage[footnote]{fixme}
\FXRegisterAuthor{sm}{asm}{SM}
\FXRegisterAuthor{ls}{als}{LS}
\FXRegisterAuthor{tw}{atw}{TW}
\FXRegisterAuthor{ud}{aud}{UD}

\usepackage{amsthm}
\usepackage{thmtools}
\usepackage[notcite,notref]{showkeys}

\usepackage{color}
\usepackage{amsmath}
\usepackage{tikz-cd}
\usepackage{ifthen}
\usepackage{etoolbox}
\usepackage{xspace}
\usepackage{mathtools}
\usepackage[noend]{algpseudocode}
\usepackage{enumitem}
\newcommand{\unnicefrac}[2]{\ensuremath{#1\mkern-1.5mu/\mkern-1.5mu{#2}}}
\usepackage{rotating}
\usepackage{array}   
\usepackage{subcaption}
\usepackage{fdsymbol}
\usepackage{dsfont}
\usepackage{environ}
\usepackage{booktabs}
\usepackage{makecell}
\usepackage{todonotes}
\usepackage{pifont}
\newcolumntype{L}{>{\(}l<{\)}} 


\definecolor{myorange}{HTML}{DA7200}
\definecolor{myblue}{HTML}{0050DC}
\definecolor{lipicsYellow}{rgb}{0.5,0.5,0.5}

\usepackage{seqsplit}
\usepackage{xstring}
\newcommand\noshowkeys{\def\hideNextShowKeysLabel{test}}
\makeatletter
\renewcommand*\showkeyslabelformat[1]{%
\@ifundefined{hideNextShowKeysLabel}{%
\noexpandarg%
\StrSubstitute{#1}{ }{\textvisiblespace}[\TEMP]%
\parbox[t]{\marginparwidth}{\raggedright\normalfont\small\ttfamily\{{\color{red!50!black}\expandafter\seqsplit\expandafter{\TEMP}}\}}%
}{}
}
\makeatother

\usetikzlibrary{calc}
\usetikzlibrary{fit}
\usetikzlibrary{backgrounds}
\usetikzlibrary{scopes}
\usetikzlibrary{decorations.markings}
\usetikzlibrary{decorations.pathreplacing}
\usetikzlibrary{external}
\usetikzlibrary{arrows.meta}
\tikzset{external/up to date check=diff}
\pgfkeys{/pgf/images/include external/.code={\includegraphics[draft=false]{#1}}}

\makeatletter
\renewcommand{\todo}[2][]{\tikzexternaldisable\@todo[#1]{#2}\tikzexternalenable}
\makeatother

\tikzsetexternalprefix{diagrams/}
\makeatletter
\tikzifexternalizing{%
\def\@enddocumenthook{}%
}{}
\makeatother

\def\temp{&} \catcode`&=\active \let&=\temp

\tikzcdset{
arrow style=tikz,
diagrams={>={Straight Barb[scale=0.8]}}
}

\newcommand{\mytikzcdcontext}[2]{
  \begin{tikzpicture}[baseline=(mainnode.base)]
    \node (mainnode) [inner sep=0, outer sep=0] {\begin{tikzcd}[#2]
        #1
    \end{tikzcd}};
  \end{tikzpicture}}

\NewEnviron{mytikzcd}[1][]{%
\def\myargs{#1}%
\edef\mydiagram{\noexpand\mytikzcdcontext{\expandonce\BODY}{\expandonce\myargs}}%
\mydiagram%
}

\makeatletter
\@ifundefined{BC}{%

}{}
\makeatother

\setlist[enumerate,1]{label=(\arabic*),font=\normalfont,align=left,leftmargin=0pt,labelindent=0pt,listparindent=\parindent,labelwidth=0pt,itemindent=!,topsep=3pt,parsep=0pt,itemsep=3pt,start=1}
\setlist[enumerate,2]{label=(\alph*),font=\normalfont,labelindent=*,leftmargin=*,start=1}
\setlist[itemize]{labelindent=*,leftmargin=*,topsep=5pt,itemsep=3pt}
\setlist[description]{labelindent=*,leftmargin=*,itemindent=-1 em}

\newcommand{\sortcoprod}[1][i]{{\textstyle\coprod}}

\newcommand{\gap}{\mathbin{\_}}

\newcommand{\CO}{\mathcal{O}}
\newcommand{\op}[1]{\ensuremath{\mathsf{#1}}}
\newcommand{\CoPaR}{\ensuremath{\mathsf{CoPaR}}\xspace}
\newcommand{\inj}{\op{in}}
\newcommand{\inl}{\op{inl}}
\newcommand{\inr}{\op{inr}}
\newcommand{\old}{{\op{old}}}
\newcommand{\id}{\op{id}}

\newcommand{\geZero}{>^{\!\!\smash{?}}}

\renewcommand{\Im}{\mathop{\op{Im}}}
\renewcommand{\ker}{\mathop{\op{ker}}}
\newcommand{\out}{\pi}

\newcommand{\Pot}{\mathcal{P}}
\newcommand{\Bag}{\mathcal{B}}

\newcommand{\partialto}{\rightharpoonup}
\newcommand{\fpair}[1]{\langle #1\rangle}
\newcommand{\Set}{\ensuremath{\mathsf{Set}}\xspace}

\newcommand{\Potf}{\ensuremath{\mathcal{P}_{\mathrm{f}}}\xspace}
\newcommand{\Dist}{\ensuremath{\mathcal{D}}\xspace}
\newcommand{\Bagf}{\ensuremath{\mathcal{B}_{\mathrm{f}}}\xspace}
\newcommand{\C}{\ensuremath{\mathcal{C}}\xspace}
\newcommand{\D}{\ensuremath{\mathcal{D}}\xspace}
\newcommand{\E}{\ensuremath{\mathcal{E}}\xspace}
\newcommand{\Z}{\ensuremath{\mathds{Z}}\xspace}
\newcommand{\N}{\ensuremath{\mathds{N}}\xspace}
\newcommand{\B}{\ensuremath{\mathds{B}}\xspace}
\newcommand{\R}{\ensuremath{\mathds{R}}\xspace}
\newcommand{\Pad}{\ensuremath{\mathsf{Pad}}\xspace}
\newcommand{\Comp}{\ensuremath{\mathsf{Comp}}\xspace}
\newcommand{\Coalg}{\ensuremath{\mathsf{Coalg}}\xspace}
\newcommand{\Inputs}{\ensuremath{A}}
\newcommand{\Leftblock}{\mathcal{L}}
\newcommand{\Middleblock}{\mathcal{M}}
\renewcommand{\gets}{\ensuremath{\xspace\mathop{:=}}}
\newcommand{\groupsum}{\raisebox{2pt}{$\scriptstyle\sum$}}
\newcommand{\coprodfunctor}{\raisebox{1pt}{$\scriptstyle\coprod$}}
\newcommand{\scriptcoprodfunctor}{\raisebox{1pt}{$\scriptscriptstyle\coprod$}}

\newcommand{\textqt}[1]{``#1''}
\newcommand{\scissors}{\ding{34}} 

\tikzstyle{inlinecd}=[column sep = 5mm, row sep = 5mm]
\tikzstyle{weakly}=[dash pattern=on 1pt off 1pt]
\tikzstyle{shiftarr}=[
        rounded corners,%
        to path={--([#1]\tikztostart.center)
                     -- ([#1]\tikztotarget.center) \tikztonodes
                     -- (\tikztotarget)},
]

\newsavebox{\mypullbackcorner}%
\sbox{\mypullbackcorner}{%
\tikzexternaldisable%
\begin{tikzpicture}
    \draw[-] (0,0) -- (.5em,.5em) -- (0,1em);
\end{tikzpicture}%
\tikzexternalenable}
\newcommand{\pullbackangle}[2][]{\arrow[phantom,to path={
                     -- ($ (\tikztostart)!1cm!#2:([xshift=8cm]\tikztostart) $)
                        node[anchor=west,pos=0.0,rotate=#2,
                        inner xsep = 0]
                        {\begin{tikzpicture}[minimum
                        height=1mm,baseline=0,#1]
    \draw[-] (0,0) -- (.5em,.5em) -- (0,1em);
                        \end{tikzpicture}}}]{}}
\newcommand{\descto}[3][]{\arrow[phantom]{#2}[#1]{\text{\footnotesize{}\begin{tabular}{c}#3\end{tabular}}}}

\makeatletter
\pgfdeclareshape{unaryfunctor}{
  \inheritsavedanchors[from=rectangle] 
  \inheritanchorborder[from=rectangle]
  \inheritanchor[from=rectangle]{center}
  \inheritanchor[from=rectangle]{north}
  \inheritanchor[from=rectangle]{south}
  \inheritanchor[from=rectangle]{west}
  \inheritanchor[from=rectangle]{east}

  \savedanchor\centerpoint{\pgfpointorigin}

  \anchor{in}{%
    \southwest \pgf@xa=\pgf@x \pgf@ya=\pgf@y
    \northeast \pgf@xb=\pgf@x \pgf@yb=\pgf@y
    \advance\pgf@x by0pt
    \pgf@y=\dimexpr 0.5\pgf@ya + 0.5\pgf@yb
  }
  \anchor{out}{%
    \southwest \pgf@xa=\pgf@x \pgf@ya=\pgf@y
    \northeast \pgf@xb=\pgf@x \pgf@yb=\pgf@y
    \pgf@x=\dimexpr \pgf@xa
    \pgf@y=\dimexpr 0.5\pgf@ya + 0.5\pgf@yb
  }

  \backgroundpath{
    \southwest \pgf@xa=\pgf@x \pgf@ya=\pgf@y%
    \northeast \pgf@xb=\pgf@x \pgf@yb=\pgf@y%

    \pgfpathmoveto{\pgfpoint{\pgf@xa}{\pgf@ya}}%
    \pgfpathlineto{\pgfpoint{\pgf@xa}{\pgf@yb}}%
    \pgfpathlineto{\pgfpoint{\pgf@xb}{\pgf@yb}}%
    \pgfpathlineto{\pgfpoint{\pgf@xb}{\pgf@ya}}%
    \pgfpathclose%
  }
}
\makeatother

\makeatletter
\pgfdeclareshape{binaryfunctor}{
  \inheritsavedanchors[from=rectangle] 
  \inheritanchorborder[from=rectangle]
  \inheritanchor[from=rectangle]{center}
  \inheritanchor[from=rectangle]{north}
  \inheritanchor[from=rectangle]{south}
  \inheritanchor[from=rectangle]{east}

  \savedanchor\centerpoint{\pgfpointorigin}

  \anchor{in1}{%
    \southwest \pgf@xa=\pgf@x \pgf@ya=\pgf@y
    \northeast \pgf@xb=\pgf@x \pgf@yb=\pgf@y
    \pgf@xc=\dimexpr 0.5773502 \pgf@xb - 0.5773502 \pgf@xa + 0.5 \pgf@yb - 0.5 \pgf@ya
    \pgf@yc=\dimexpr 0.5 \pgf@ya +  0.5 \pgf@yb
    \pgf@y=\dimexpr 0.75\pgf@ya + 0.25\pgf@yb
    \pgf@x=\dimexpr \pgf@xb
    \pgf@y=\dimexpr \pgf@yc + 0.5\pgf@xc
  }
  \anchor{in2}{%
    \southwest \pgf@xa=\pgf@x \pgf@ya=\pgf@y
    \northeast \pgf@xb=\pgf@x \pgf@yb=\pgf@y
    \pgf@xc=\dimexpr 0.5773502 \pgf@xb - 0.5773502 \pgf@xa + 0.5 \pgf@yb - 0.5 \pgf@ya
    \pgf@yc=\dimexpr 0.5 \pgf@ya +  0.5 \pgf@yb
    \pgf@x=\dimexpr \pgf@xb
    \pgf@y=\dimexpr \pgf@yc - 0.5\pgf@xc
  }
  \anchor{out}{%
    \southwest \pgf@xa=\pgf@x \pgf@ya=\pgf@y
    \northeast \pgf@xb=\pgf@x \pgf@yb=\pgf@y
    \pgf@xc=\dimexpr 0.5773502 \pgf@xb - 0.5773502 \pgf@xa + 0.5 \pgf@yb - 0.5 \pgf@ya
    \pgf@yc=\dimexpr 0.5 \pgf@ya +  0.5 \pgf@yb
    \pgf@x=\dimexpr \pgf@xb - 1.73205\pgf@xc
    \pgf@y=\dimexpr \pgf@yc
  }

  \backgroundpath{
    \southwest \pgf@xa=\pgf@x \pgf@ya=\pgf@y%
    \northeast \pgf@xb=\pgf@x \pgf@yb=\pgf@y%
    \pgf@xc=\dimexpr 0.5773502 \pgf@xb - 0.5773502 \pgf@xa + 0.5 \pgf@yb - 0.5 \pgf@ya
    \pgf@yc=\dimexpr 0.5 \pgf@ya +  0.5 \pgf@yb

    \pgfpathmoveto{\pgfpointadd{\pgfpoint{\pgf@xb}{\pgf@yc}}{\pgfpoint{-1.73205\pgf@xc}{0}}}%
    \pgfpathlineto{\pgfpointadd{\pgfpoint{\pgf@xb}{\pgf@yc}}{\pgfpoint{0}{\pgf@xc}}}%
    \pgfpathlineto{\pgfpointadd{\pgfpoint{\pgf@xb}{\pgf@yc}}{\pgfpoint{0}{-\pgf@xc}}}%
    \pgfpathclose%
  }
}
\makeatother
\tikzset{
      functornode/.style={
        draw=black!80,
        line width=0.5pt,
        anchor=out,
        xshift=8mm,
        fill=none,
        execute at begin node=$,%
        execute at end node=$,%
        rounded corners= 2pt,
      },
      unary/.style={
        functornode,
        shape=unaryfunctor,
        minimum height = 2em,
        minimum width = 2.2em,
      },
      binary/.style={
        functornode,
        shape=binaryfunctor,
        inner ysep=1pt,
      },
}

\tikzstyle{partitionBlock}=[
        shape=rectangle,
        rounded corners=2.5mm,
        minimum height=5mm,
        minimum width=5mm,
        draw=lipicsYellow,
        line width=1pt,
        inner sep = 2mm,
]

\DeclareUnicodeCharacter{D7}{\ensuremath{\times}}

\theoremstyle{plain}
\newtheorem{theorem}[thm]{Theorem}
\newtheorem{corollary}[thm]{Corollary}
\newtheorem{fact}[thm]{Fact}
\newtheorem{lemma}[thm]{Lemma}
\newtheorem{proposition}[thm]{Proposition}

\theoremstyle{definition}
\newtheorem{algorithm}[thm]{Algorithm}
\newtheorem{assumption}[thm]{Assumption}
\newtheorem{construction}[thm]{Construction}
\newtheorem{definition}[thm]{Definition}
\newtheorem{examples}[thm]{Examples}
\newtheorem{example}[thm]{Example}
\newtheorem{notation}[thm]{Notation}
\newtheorem{remark}[thm]{Remark}

\theoremstyle{remark}

\numberwithin{equation}{section}

\def\tikzarrow#1{\tikzset{external/export next=false}\tikz[commutative diagrams/every diagram,baseline=-3.1pt]
  \draw[commutative diagrams/.cd,every arrow, every label,#1] (0,0) -- +(1.2em,0);}
\newsavebox{\boxtwoheadrightarrow}
\sbox{\boxtwoheadrightarrow}{\tikzarrow{->>}}
\def\twoheadrightarrow{\mathbin{\usebox{\boxtwoheadrightarrow}}}
\def\epito{\twoheadrightarrow}
\newsavebox{\boxmonoto}
\sbox{\boxmonoto}{\tikzarrow{>->}}
\def\rightarrowtail{\mathbin{\usebox{\boxmonoto}}}
\def\monoto{\rightarrowtail}
\def\subto{\hookrightarrow}

\def\upa{\mathord{\uparrow}}

\newcommand{\takeout}[1]{\empty}

\newcommand{\ith}{\ensuremath{i^\text{th}}\xspace}

\newcommand{\copar}{\textsf{CoPaR}\xspace}

\title[Efficient and Modular Coalgebraic Partition Refinement]{Efficient and Modular\\ Coalgebraic Partition Refinement}

\begin{document}

\titlecomment{Work performed as part of the DFG-funded project
  COAX (MI~717/5-1 and SCHR 1118/12-1) }

\author[T.~Wißmann]{Thorsten Wißmann}
\author[U.~Dorsch]{Ulrich Dorsch}
\author[S.~Milius]{Stefan Milius}
\author[L.~Schröder]{Lutz Schröder}
\address{Friedrich-Alexander-Universität Erlangen-Nürnberg, Germany}	
\email{\{thorsten.wissmann,ulrich.dorsch,stefan.milius,lutz.schroeder\}@fau.de}  

\begin{abstract}
  We present a generic partition refinement algorithm that quotients
  coalgebraic systems by behavioural equivalence, an important task in
  system analysis and verification.  Coalgebraic generality allows us
  to cover not only classical relational systems but also, e.g.\
  various forms of weighted systems and furthermore to flexibly
  combine existing system types. Under assumptions on the type functor
  that allow representing its finite coalgebras in terms of nodes and
  edges, our algorithm runs in time $\mathcal{O}(m\cdot \log n)$ where
  $n$ and $m$ are the numbers of nodes and edges, respectively. The
  generic complexity result and the possibility of combining system
  types yields a toolbox for efficient partition refinement
  algorithms. Instances of our generic algorithm match the
  run-time of the best known algorithms for unlabelled transition
  systems, Markov chains, deterministic automata (with fixed
  alphabets), Segala systems, and for color refinement.
\end{abstract}

\maketitle


\section{Introduction}
The minimization of a state based system typically consists of two steps:
\begin{enumerate}
\item Removal of unreachable states.
\item Identification of states exhibiting the same behaviour, w.r.t.~a
  suitable notion of `sameness'; here we are interested in
  \emph{minimization under bisimilarity}.
\end{enumerate} The computation of reachable states is usually
accomplished by a straightforward search through the transition graph
of a system. Minimization under
bisimilarity however is more complex because of its corecursive
nature: whether two states are bisimilar depends on which of their
successors are bisimilar. In the present work, we present a generic
algorithm to perform bisimilarity minimization efficiently for a broad
class of systems.

The task of minimization appears as a subtask in state space reduction
(e.g.~\cite{BlomO05}) or non-interference
checking~\cite{MeydenZ07}. The notion of bisimulation was first
defined for relational systems~\cite{Benthem77,Milner80,Park81}; it
was later extended to other system types including probabilistic
systems~\cite{LarsenS91,DesharnaisEA02}, weighted
automata~\cite{Buchholz08}, and (weighted) tree
automata~\cite{HoegbergEA09,HoegbergEA07}. More generally,
\emph{universal coalgebra} (see e.g.~Rutten~\cite{Rutten00}) provides
a framework capturing all these types of systems uniformly, and their
notions of bisimulation appear as special instances of Aczel and
Mendler's notion of bisimulation for
coalgebras~\cite{aczelmendler:89}.

The importance of minimization under bisimilarity appears to increase with
the complexity of the underlying system type. E.g.\ while in LTL model
checking, minimization drastically reduces the state space but,
depending on the application, does not necessarily lead to a speedup
in the overall balance~\cite{FislerV02}, in probabilistic model
checking, minimization under strong bisimilarity does lead to
substantial efficiency gains~\cite{KatoenEA07}. This is the reason why
model checkers implement bisimilarity minimization, e.g.~the mCRL2
toolset~\cite{BunteEA19} provides explicit routines for comparing and
minimizing systems w.r.t.~strong bisimilarity and also other types of
equivalences.  

The algorithmics of minimization, often referred to as \emph{partition
  refinement} or \emph{lumping}, has received a fair amount of
attention. Since bisimilarity is a greatest fixpoint, it is more or
less immediate that it can be calculated in polynomial time by
approximating this fixpoint from above following Kleene's fixpoint
theorem. For transition systems, Kanellakis and
Smolka~\cite{KanellakisSmolka83,KanellakisS90} introduced an algorithm that in fact runs
in time $\mathcal{O}(nm)$ where $n$ is the number of nodes and $m$ is
the number of transitions. An even more efficient algorithm running in
time~$\mathcal{O}(m\log n)$ was later described by Paige and
Tarjan~\cite{PaigeTarjan87}; this bound holds even if the number of
action labels is not fixed~\cite{Valmari09}. Current algorithms
typically apply further optimizations to the Paige-Tarjan algorithm,
achieving better average-case behaviour but the same worst-case
behaviour~\cite{DovierEA04}. Probabilistic minimization has undergone
a similarly dynamic
development~\cite{BaierEM00,CattaniS02,ZhangEA08,GrooteEA18}, and the
best algorithms for minimization of Markov chains now have the same
$\CO(m\log n)$ run-time as the relational Paige-Tarjan
algorithm~\cite{HuynhTian92,DerisaviEA03,ValmariF10}.  Using ideas
from abstract interpretation, Ranzato and Tapparo~\cite{RanzatoT08}
have developed a relational partition refinement algorithm that is
generic over \emph{notions of process equivalence}. As instances, they
recover the classical Paige-Tarjan algorithm for strong bisimilarity
and an algorithm for stuttering equivalence, and obtain new algorithms
for simulation equivalence and for a new process equivalence; the
generic run-time analysis, however, is coarser for this algorithm, and
in particular does not recover the $\CO(m\log n)$ bound for the
classical Paige-Tarjan algorithm. Recently, Groote et
al.~\cite{GrooteEA17} have presented an improved algorithm for
relational partition refinement that covers stuttering, branching and
strong bisimilarity.


In this paper we follow an orthogonal approach and provide a generic
partition refinement algorithm that can be instantiated for many
different \emph{types} of systems (e.g.~nondeterministic,
probabilistic, weighted). The key to genericity is to use the methods of
universal coalgebra. That is, we encapsulate
transition types of systems as endofunctors on sets (or a more general
category), and model systems as coalgebras for a given type
functor. 

\subsection*{Overview of the paper}
In \autoref{sec:prelim}, the categorical generalizations of the standard set
operations on partitions and equivalence relations are introduced. A short
introduction to coalgebras as a framework for state-based systems is given.

In order to explain the generic pattern that existing partition
refinement algorithms in the literature follow, we exhibit in
\autoref{sec:fuzzyAlgo} an informal partition refinement algorithm in
natural language that operates on a high level of generality.

In \autoref{sec:cat}, the generic pattern is made precise by a
categorical construction, in which we work with coalgebras for a
monomorphism-preserving endofunctor on a category with image
factorizations. Here we present a quite general category-theoretic
partition refinement algorithm, and we prove its correctness. The
algorithm is parametrized over a \op{select} routine that determines
which observations are used to split blocks of states. We present two
\op{select} routines; one yields a known coalgebraic final-chain
algorithm (e.g.~\cite{DBLP:conf/fossacs/AdamekBHKMS12,KonigKupper14}),
the other routine is \textqt{select the smaller half}, a trick that
goes back to Hopcroft~\cite{Hopcroft71} and lies at the heart of most
modern partition refinement algorithms including Paige and
Tarjan's~\cite{PaigeTarjan87}, being responsible for the logarithmic
(rather than linear) dependence of the run-time on the number of
states.

While the categorical construction recomputes the involved partitions from
scratch in each iteration, we present an optimized version of our algorithm
(\autoref{sec:opti}) that computes the partitions incrementally. For the
correctness of the optimization, we need to restrict to sets and assume that the
type endofunctor satisfies a condition we call \emph{zippability}. This property
holds, e.g.\ for all polynomial endofunctors on sets and for the type functors
of labelled and weighted transition systems, but is not closed under composition
of functors.

In order to be able to provide a concrete presentation of our
algorithm and perform a complexity analysis we make the algorithm
parametric in an abstract \emph{refinement interface} of the type
functor, which encapsulates simple functor specific computations
needed to implement the \textqt{select the smaller half}
routine. In \autoref{sec:interface} we introduce
refinement interfaces, and we provide several examples for various zippable
type functors of interest and show that they can be implemented with a
linear run-time.

Then in \autoref{sec:efficient} we provide pseudocode for the
algorithm using the incremental computation of the partitions from
\autoref{sec:opti} and with the \textqt{select the smaller half}
routine hard-wired. We show that if the
refinement interface operations can be implemented to run in linear time, then the
algorithm runs in time $\CO((m+n) \cdot \log n)$, where $n$ is the number of
states and $m$ the number of `edges' in a syntactic encoding of the input
coalgebra. We thus recover the complexity of the most efficient known algorithms
for transition systems (Paige and Tarjan~\cite{PaigeTarjan87}), for weighted
systems (Valmari and Franceschinis~\cite{ValmariF10}), for the task of color
refinement (Berkholz, Bonsma, and Grohe~\cite{BerkholzBG17}), as well as
Hopcroft's classical automata minimization algorithm~\cite{Hopcroft71} for a
fixed alphabet $A$ and $m= n\cdot |A|$.

\autoref{sec:multisorted} is devoted to modularity and explains how to
handle combinations of system types, in particular functor
composition. We will see that this can be achieved with just a bit of
extra preprocessing, so that our main algorithm need not be adjusted
at all. In fact, given a functor $T$ built as a term from finitary
functors $\Set^k\to \Set$, we first recall from
\cite{SchroderPattinson11} how this induces a functor
$\bar T\colon \Set^n\to\Set^n$ on multisorted sets, and we present a
transformation from finite $T$-coalgebras to finite
$\bar T$-coalgebras (with possibly more states) that reflects
bisimilarity minimization. We then present a new construction that
provides for every functor on the category $\C^n$, where $\C$ is any
extensive category (e.g.~\Set), a functor on $\C$ and a transformation
from coalgebras of the former to coalgebras of the latter that
preserves the size of carriers (i.e.~the number of states) and
preserves and reflects bisimilarity minimization. This yields a
reduction from bisimilarity minimization of $T$-coalgebras to
minimization of $\coprodfunctor \bar T\Delta$-coalgebras, where
$\Delta\colon \Set \to \Set^n$ is the diagonal functor and
$\coprodfunctor\colon \Set^n \to \Set$ takes coproducts. The latter
problem is solved by the algorithm from \autoref{sec:efficient},
because we prove that if $T$ is built from functors fulfilling our
assumptions, then $\coprodfunctor \bar T \Delta$ fulfils the
assumptions too -- even if $T$ itself does not.

As instances of this result, we obtain an efficient \emph{modular}
algorithm for systems whose type is built from basic system types
fulfilling our assumptions, e.g.~probability, non-determinism,
weighted branching (with weights in an arbitrary abelian group), by
composition, finite products and finite coproducts
(\autoref{ex:comp}).

One of these instances is an $\CO((m+n) \log (m+n))$ algorithm for
Segala systems, to our knowledge a new result (more precisely, we
improve an earlier bound established by Baier, Engelen, and
Majster-Cederbaum~\cite{BaierEM00}, roughly speaking by letting only
non-zero probabilistic edges enter into the time bound). Note that Groote et
al.'s recent algorithm~\cite{GrooteEA18} for Segala systems, which was
discovered independently and almost at the same time, has a similar
complexity as ours.
We also obtain efficient minimization algorithms for general Segala
systems and alternating systems~\cite{Hansson:1994:TPF:561335}. In
further work~\cite{DMSW19} we extend our algorithm to cover weighted
branching (with weights in an arbitrary monoids) and tree automata,
for which we obtain an algorithm improving the previous best
complexity for minimization w.r.t.~backwards bisimulation. In addition
\emph{op.~cit.} presents a generic implementation in the form of the
partition refinement tool \copar.

This paper is an extended and completely reworked version of a
previous conference paper~\cite{dmsw17}. Besides providing detailed
proofs of all our results, we have included the new
\autoref{sec:multisorted} showing that modularity is achieved without
any adjustment of our algorithm.

\paragraph{\bf Acknowledgement.} We would like to thank the anonymous
referees for their comments, which helped to improve the presentation
of our paper. 

\section{Preliminaries}
\label{sec:prelim}
\tikzsetfigurename{partref-prelim-}

It is advisable for readers to be familiar with basic category
theory~\cite{joyofcats}.  However, our results can be understood by
reading the notation in the usual set-theoretic way, which corresponds
to their meaning in $\Set$, the category of sets and functions. For the
convenience of the reader we recall some concepts that are central for
the categorical version of our algorithm.

\subsection{Equivalence Relations and Partitions, Categorically}\label{sec:equiv}

Our most general setting is a category $\C$ in which we have a well-behaved
notion of equivalence relation corresponding to quotient
objects. We will assume that $\C$ has finite products and
pullbacks. 
\begin{notation} \label{pullbackNotation} The terminal object of $\C$
  is denoted by $1$, with unique morphisms $!\colon A\to 1$. In \Set,
  $1 = \{0\}$ as usual. We denote the product of objects $A$, $B$ by
  $A \xleftarrow{\out_1} A×B \xrightarrow{\out_2} B$. Given
  $f\colon D\to A$ and $g\colon D\to B$, the morphism induced by the
  universal property of the product $A \times B$ is denoted by
  $\fpair{f,g}\colon D\to A\times B$.

  For morphisms $f\colon A\to D$, $g\colon B\to D$, we denote by
  \[
    \begin{mytikzcd}
      P \arrow{r}{\pi_1}
      \arrow{d}[swap]{\pi_2}
      \pullbackangle{-45}
      & A
      \arrow{d}{f}
      \\
      B
      \arrow{r}{g}
      & D
    \end{mytikzcd}
  \]
  that $P$ (together with the projections $\pi_1, \pi_2$) is the pullback of $f$
  along $g$. In \Set, we have
  \[
    P = \{ (a,b) \in A\times B \mid f(a) = g(b) \}.
  \]
\end{notation}

\noindent The \emph{kernel} $\ker f$ of a morphism $f\colon A\to B$ is
the pullback of $f$ along itself.
%
We write $\rightarrowtail$ for monomorphisms (in $\Set$, the
monomorphisms are the ~injections), and $\epito$ for regular
epimorphisms; by definition,~$q\colon A\to B$ is a regular epimorphism
($q\colon A \epito B$) if there exists a parallel pair of morphisms
$f,g\colon R\rightrightarrows A$ such that $q$ is the coequalizer of
$f$ and $g$. In \Set, the coequalizer of
$f,g\colon R\rightrightarrows A$ is the quotient of $A$ modulo the
smallest equivalence relation on $A$ relating $f(r)$ and $g(r)$ for
all $r\in R$. Here, one can think of $R$ as a set of witnesses~$r$ of
the pairs $(f(r), g(r))$ generating that equivalence (note that there
may be several witnesses for the same pair). Hence, we will often
denote the coequalizer of $f,g\colon R\rightrightarrows A$ by
$\kappa_R\colon A\twoheadrightarrow A/R$, and such a map represents a
\emph{quotient}. When $f$ and $g$ are clear from the context, we will
just write $R\rightrightarrows A$.

Kernels and coequalizers allow us to talk about equivalence relations
and partitions in a category, and when we speak of an
\emph{equivalence relation on the object $X$} of $\C$ we mean the
kernel of some morphism with domain $X$. Indeed, recall that for
every set $A$, every equivalence relation $\sim$ is the kernel of the
canonical quotient map $\kappa_\sim\colon A \epito A/\mathord{\sim}$, and there is a
bijection between equivalence relations and partitions on
$A$. In order to obtain a similar bijection for more general
categories than $\Set$ we make the following global assumption.

\begin{assumption}
  \label{ass:C}
  We assume throughout that $\C$ is a finitely complete category that
  has coequalizers and in which regular epimorphisms are closed under composition.
\end{assumption}
\begin{examples}
  Examples of categories satisfying~\autoref{ass:C} abound. In
  particular, every \emph{regular} category with coequalizers
  satisfies our assumptions. The category $\Set$ of sets and functions is
  regular. Every topos is regular, and so is every
  finitary variety, i.e.~a category of algebras for a finitary
  signature satisfying given equational axioms (e.g.~monoids, groups,
  vector spaces etc.). If $\C$ is regular, so is the functor category
  $\C^\E$ for any category $\E$. For our main applications, we will be
  interested in the special case $\C^n$ where $n$ is a natural number,
  i.e.~the case where $\E$ is the discrete category with the set of
  objects $\{1,\ldots, n\}$.

  The category of posets and the category of topological spaces both
  fail to be regular but still satisfy our assumptions.
\end{examples}

\noindent In \Set, a function $f\colon A\to B$ factorizes through the
partition $\unnicefrac{A}{\ker f}$ induced by its kernel, via the
function $[-]_f\colon A\twoheadrightarrow \unnicefrac{A}{\ker f}$ taking
equivalence classes
\begin{equation}
  [x]_f := \{ x'\in A\mid f(x) = f(x')\} = \{ x' \in D\mid (x,x') \in \ker f\}.
  \label{eq:funEquivClass}
\end{equation}
Well-definedness of functions on $\unnicefrac{A}{\ker f}$ is determined
precisely by the universal property of $[-]_f$ as a coequalizer of $\ker
f\rightrightarrows A$. In particular, $f$ induces an injection
$\unnicefrac{A}{\ker f}\rightarrowtail B$; together with $[-]_f$, this is the
factorization of $f$ into a regular epimorphism and a monomorphism.

More generally, our category $\C$ from \autoref{ass:C} has a
(RegEpi,Mono)-factorization system \cite[Prop.~14.22]{joyofcats}, that
is, every morphism $f\colon A\to B$ has a factorization
\[
  \begin{mytikzcd}
    A
    \arrow[shiftarr={yshift=6mm}]{rr}{f}
    \arrow[->>]{r}{e}
    & \Im(f)
    \arrow[>->]{r}{m}
    & B,
  \end{mytikzcd}
\]
where $m$ is a monomorphism and $e$ is a regular epimorphism,
specifically the coequalizer of the kernel
$\pi_1,\pi_2\colon \ker f\rightrightarrows A$. The subobject
$m\colon \Im(f) \monoto B$ is called the \emph{image} and the
factorization $f = m \cdot e$ the \emph{image factorization} of
$f$. In every category, we have the \emph{diagonal fill-in property}
for monomorphisms~$m$ and regular epimorphisms~$e$: Whenever
$f\cdot e=m\cdot g$ then there exists a unique diagonal~$d$ such that
$d\cdot e=g$ and $m\cdot d=f$, implying that
(RegEpi,Mono)-factorizations are unique.

Using image factorizations it is easy to show that in our category $\C$, there is a
bijection between kernels $K\rightrightarrows A$ and quotients of $A$
-- the two directions of this bijection are given by taking the kernel
of a coequalizer and by taking the coequalizer of a kernel.
In particular, every regular epimorphism is the
coequalizer of its kernel.

Furthermore, the above bijection is in fact an order isomorphism
between the natural partial orderings on kernels and quotients,
respectively. In detail, relations from $A$ to $B$ in~$\C$,
i.e.~jointly monic spans $A \leftarrow R \to B$, and in particular
kernels, represent subobjects of $A\times B$, which are ordered by
inclusion in the usual way: We say that a relation
$\fpair{p_1,p_2}\colon R \monoto A \times B$ (or a kernel) is
\emph{finer than} a relation
$\fpair{p_1',p_2'}\colon R' \monoto A\times B$ if there exists
$m\colon R \monoto R'$ (necessarily unique and monic) such that
$p_i' \cdot m = p_i$, for $i = 1,2$.  We write $\cap$
(\emph{intersection}) and $\cup$ (\emph{union}) for meets and joins of
kernels in the inclusion ordering on \emph{relations} (not equivalence
relations or kernels) on $A$. In this notation,
\begin{equation}\label{eq:kercap}
  \ker \fpair{f,g} =\ker f\cap\ker g;
\end{equation}
in particular, kernels are stable under intersection of relations.
Similarly, a quotient represented by
$q_1\colon A\twoheadrightarrow B_1$ is \emph{finer than} a quotient
represented by $q_2\colon A\twoheadrightarrow B_2$ if there exists
$b\colon B_1\epito B_2$ (necessarily unique and regular epic) with
$q_2=b\cdot q_1$.

We need a few simple observations on kernels that are familiar when
instantiated to~$\Set$:
\begin{remark}\label{rem:kernel}
  \begin{enumerate}
  \item\label{i:comp} For every $f\colon X \to Y$ and $g\colon Y \to Z$, 
    $\ker(f)$ is finer than $\ker(g\cdot f)$.
  \item\label{i:mono} For every $f\colon X \to Y$ and every mono $m\colon Y
    \monoto Z$, $\ker(m\cdot f) = \ker f$.
  \item\label{i:monofac} For every $f\colon X \to Y$ and regular epi $q\colon X
    \epito Z$, $\ker(f) = \ker (q)$ iff there exists a mono $m\colon
    Z \monoto Y$ with $f = m\cdot q$. One implication follows from the
    previous point, and conversely, take the image factorization $f =
    n \cdot e$. Then $\ker(e) = \ker(f)$ by point~\ref{i:mono}, and
    therefore $e$ and $q$ represent the same quotient, which means
    there exists some isomorphism $i\colon Z \to \Im(f)$ with $i \cdot q =
    e$. It follows that $m = n \cdot i$ is the desired mono.
  \item\label{i:finer} For every $f\colon X \to Z$ and regular epi $e\colon X \epito Y$,
    $\ker(e)$ is finer than $\ker(f)$ iff
    there exists some morphism $g\colon Y \to Z$ such that $g \cdot e = f$.

    To see this let $p,q\colon \ker(e) \rightrightarrows X$ be the kernel
    pair of $e$. If $\ker(e)$ is finer than
    $\ker(f)$, then we have $f\cdot p = f\cdot q$. Hence, since $e$
    is the coequalizer of its kernel pair we obtain $g$ as desired
    from the universal property of $e$. For the other direction apply \ref{i:comp}.


  \item\label{i:squares} 
    Whenever $\ker (a\colon D\to A) = \ker (b\colon D\to B)$ then
    $\ker (a\cdot g) = \ker (b\cdot g)$, for every $g\colon W\to D$.

    Indeed, the kernel $\ker(a \cdot g)$ can be obtained from $\ker a$
    by pasting pullback squares as shown below (and similarly for $b\colon
    D \to B$):
    \[
      \begin{mytikzcd}
        \ker\big( a \cdot g\big)
        \arrow{r}
        \arrow{d}
        \pullbackangle{-45}
        & \bullet
        \pullbackangle{-45}
        \arrow{r}
        \arrow{d}
        & W
        \arrow{d}{g}
        \\
        \bullet
        \arrow{d}
        \arrow{r}
        \pullbackangle{-45}
        & \ker a
        \arrow{d}
        \arrow{r}
        \pullbackangle{-45}
        & D
        \arrow{d}{a}
        \\
        W
        \arrow{r}{g}
        & D
        \arrow{r}{a}
        & A
      \end{mytikzcd}
    \]
    So if $\ker a = \ker b$, then $\ker(a\cdot g) = \ker(b\cdot g)$.
  \end{enumerate}
\end{remark}

\noindent Even though only existence of coequalizers is assumed, $\C$
has more colimits:

\begin{lemma} \label{lem:regepi-pushout}
  $\C$ has pushouts of regular epimorphisms (i.e.~pushouts of spans containing
  at least one regular epimorphism).\twnote{sollten wir die nicht lieber
  pushouts \emph{along} reg epis nennen? Inverse images sind ja auch pullbacks
  along monos und nicht pullbacks of monos}
\end{lemma}
\begin{proof}
  Let $X\xleftarrow{e}Y\xrightarrow{h}W$ be a span, with $e$ a
  regular epi. Let $(\pi_1,\pi_2)$ be the kernel pair of $e$, and
  let $q\colon W\epito Z$ be the coequalizer of $h\cdot \pi_1$ and $h\cdot \pi_2$.
  Since $e$ is the
  coequalizer of $\pi_1,\pi_2$ there exists $r\colon X\to Z$ such
  that $r\cdot e=q\cdot h$. We claim that the square in
  \begin{equation*}
    \begin{mytikzcd}
      \ker e
      \arrow[shift left=1]{r}{\pi_1}
      \arrow[shift right=1]{r}[swap]{\pi_2}
      & Y \arrow[->>]{r}[above]{e} \arrow{d}[left]{h} & 
      X\arrow{d}[right]{r} \\
      & W\arrow[->>]{r}[below]{q} & Z
    \end{mytikzcd}
  \end{equation*}
  is a pushout. Uniqueness of mediating morphisms is clear since $q$
  is epic; it remains to show existence. So let
  $W\xrightarrow{g} U \xleftarrow{f} X$ be such that
  $g\cdot h = f\cdot e$. Then
  $g\cdot h\cdot \pi_1=f\cdot e\cdot\pi_1=f\cdot e\cdot \pi_2=g\cdot
  h\cdot\pi_2$,
  so by the universal property of $q$ we obtain $k\colon Z\to U$ such
  that $k\cdot q=g$. It remains to check that $k\cdot r=f$. Indeed, we
  have $k\cdot r\cdot e=k\cdot q\cdot h=g\cdot h=f\cdot e$, which
  implies the claim because $e$ is epic.
\end{proof}

\subsection{Coalgebra}\label{sec:coalg}
We now briefly recall basic notions from universal coalgebra, seen as
a unified framework for state-based reactive systems. For introductory
texts, see~\cite{Rutten00,JacobsR97,Adamek05,Jacobs17}.  Given an
endofunctor $H\colon \C \to \C$, a \emph{coalgebra} is a pair
$(X,\xi)$
where $X$ is an object of $\C$ called the \emph{carrier} and thought
of as an object of \emph{states}, and $\xi\colon X \to HX$ is a morphism
called the \emph{structure} of the coalgebra. Our leading examples are
the following.
\begin{example}\label{ex:coalg}
  \begin{enumerate}
  \item Labelled transition systems with labels from a set $A$ are
    coalgebras for the functor $HX = \Pot(A \times X)$ (and unlabelled
    transition systems are simply coalgebras for~$\Pot$) on $\Set$. Explicitly,
    a coalgebra $\xi\colon X\to HX$ assigns to each state $x$ a set
    $\xi(x)\in\Pot(A\times X)$, and this represents the transition
    structure at $x$: $x$ has an $a$-transition to $y$ iff
    $(a,y)\in \xi(x)$.
    In concrete examples, we restrict to the finite powerset functor $\Potf X =
    \{S\in \Pot X \mid S\text{ finite}\}$, and coalgebras for $HX =
    \Potf(A\times X)$ are finitely branching LTSs.
  \item\label{ex:coalg:2} Weighted transition systems with weights from a
    commutative monoid $(M,+,0)$ are modelled as coalgebras as
    follows. We consider the \emph{monoid-valued} functor $M^{(-)}$
    defined on sets~$X$ by
    \[
      M^{(X)} = \{ f\colon X \to M \mid f(x) \neq 0\text{ for only finitely many
      }x\},
    \]
    and on maps $h\colon X \to Y$ by
    \[
      M^{(h)}(f)(y) = \textstyle\sum_{h(x) = y} f(x).
    \]
    $M$-weighted transition systems are in bijective correspondence
    with coalgebras for $M^{(-)}$, and for $M$-weighted
    labelled transition systems one takes $(M^{(-)})^A$, where $A$ is
    the label alphabet (see~\cite{gs01}).

  \item\label{item:bags} The finite powerset functor $\Potf$ is the
    monoid-valued functor $\B^{(-)}$ for the Boolean monoid
    $\B = (2, \vee, 0)$.  The \emph{bag functor} $\Bagf$, which
    assigns to a set~$X$ the set of bags (i.e.~finite multisets)
    on~$X$, is the monoid-valued functor for the additive monoid of
    natural numbers $(\N, + , 0)$.

  \item Probabilistic transition systems are modelled coalgebraically
    using the distribution functor~$\Dist$. This is the subfunctor
    $\Dist X\subseteq \R_{\ge 0}^{(X)}$, where $\R_{\ge 0}$ is the
    monoid of addition on the non-negative reals, given by
    $\Dist X=\{f\in\R_{\ge 0}^{(X)}\mid \sum_{x \in X} f(x) = 1\}$.

  \item\label{i:segala} Simple (resp.~general) Segala
    systems~\cite{Segala95} strictly alternate between
    non-deterministic and probabilistic transitions. 
    Simple Segala systems can
    be modelled as coalgebras for the set functor $\Potf(A×\Dist(-))$,
    which means that for every label $a \in A$, a state
    non-deterministically proceeds to one of a finite number of
    possible probability distributions over states. General Segala
    systems are coalgebras for $\Potf\Dist(A \times -)$, which means
    that a state $s$ non-deterministically proceeds to one of a finite
    number of distributions over the set of $A$-labelled transitions
    from $s$.
    
  \item\label{i:poly} Let $\Sigma$ be a signature (a.k.a ranked alphabet), i.e.~a
    set of (operation) symbols, each with a prescibed natural number, its
    \emph{arity}. The corresponding \emph{polynomial
      functor} $H_\Sigma\colon \Set \to \Set$ maps a set $X$ to the set
    \[
      H_\Sigma X = \coprod_{n \in \N} \Sigma_n \times X^n,
    \]
    where $\Sigma_n$ is the set of symbols of arity $n$, and
    $H_\Sigma$ acts similarly on maps. Note that the elements of $H_\Sigma X$
    may be identified with shallow terms over $X$, i.e.~formal
    expressions $\sigma(x_1, \ldots, x_n)$, where $\sigma \in \Sigma$ is an
    $n$-ary symbol and $x_1, \ldots, x_n \in X$.  

    A coalgebra $\xi\colon X \to H_\Sigma X$ is a deterministic system,
    where the coalgebra structure assigns to each state $x \in X$ a
    tuple $\xi(x) = (\sigma, x_1, \ldots, x_n)$ in which $\sigma \in
    \Sigma$ is an $n$-ary \emph{output} symbol and $x_1, \ldots, x_n$
    are successor states, one for every input in $\{1, \ldots, n\}$. 
    Coalgebras for $H_\Sigma$ can also be thought of as top-down
    deterministic tree automata.
  \item For a fixed finite set $A$, the coalgebras of the functor
    $HX=2\times X^A$ are deterministic automata for the input alphabet
    $A$ (neglecting the initial state). Concretely, a coalgebra
    structure $\fpair{o,\delta}\colon X\to 2\times X^A$ consists of
    the characteristic function $o\colon X \to 2$ of the set of final
    states of the automaton and the next state function $\delta$. Note
    that the functor $H$ is (naturally isomorphic to) a polynomial
    functor, where the signature $\Sigma$ consists of two $A$-ary
    operation symbols.
  \end{enumerate}
\end{example}

\noindent
\begin{minipage}{.8\textwidth}
  A \emph{coalgebra morphism} from a coalgebra $(X,\xi)$ to a coalgebra
  $(Y,\zeta)$ is a morphism $h\colon X \to Y$ such that
  $\zeta \cdot h = Hh \cdot \xi$; intuitively, coalgebra morphisms preserve
  observable behaviour.  Coalgebras and their morphisms form a
  category $\Coalg(H)$. The forgetful functor $\Coalg(H)\to \C$
  creates all colimits, so $\Coalg(H)$ has all colimits that $\C$ has;
  in particular, our running assumptions imply the following:
\end{minipage}%
\begin{minipage}{.2\textwidth}%
\begin{flushright}%
\begin{mytikzcd}
    X
    \arrow{r}{\xi}
    \arrow{d}[swap]{h}
    & HX
    \arrow{d}{Hh}
    \\
    Y \arrow{r}{\zeta}
    & HY
  \end{mytikzcd}%
\end{flushright}%
\end{minipage}

\begin{corollary}\label{cor:coalg-colims}
  $\Coalg(H)$ has all coequalizers and pushouts of regular
  epimorphisms.
\end{corollary}

\noindent A \emph{subcoalgebra} of a coalgebra $(X,\xi)$ is represented
by a coalgebra morphism $m\colon (Y,\zeta) \monoto (X,\xi)$ such that $m$ is
a monomorphism in $\C$. Likewise, a \emph{quotient} of a coalgebra
$(X,\xi)$ is represented by a coalgebra morphism
$q\colon (X,\xi) \epito (Y,\zeta)$ carried by a regular epimorphism $q$ of
$\C$. If $H$ preserves monomorphisms, then the image factorization
structure on $\C$ lifts to coalgebras in the sense that every
coalgebra morphism $f$ has a factorization $f = m \cdot e$ into
coalgebra morphisms $m$ and $e$ such that $m$ is a monomorphism and
$e$ a regular epimorphism in $\C$ (see e.g.~\cite[Lemma 2.5]{MPW19}).

Recall that a coalgebra is called \emph{simple} if it
does not have any non-trivial quotients~\cite{Gumm03}. We will use the following
equivalent characterization:
\begin{proposition}\label{prop:simple}
  If $H$ preserves monomorphisms,
  then a coalgebra $(X,\xi)$ is simple iff every coalgebra
 morphism with domain $(X,\xi)$ is carried by a monomorphism.
\end{proposition}
\begin{proof}
  For necessity, consider a coalgebra morphism
  $h\colon (X,\xi) \to (Y,\zeta)$ and take its image factorization to obtain
  $q\colon (X,\xi) \twoheadrightarrow (\Im(h),i)$ and
  $m\colon (\Im(h),i) \monoto (Y,\zeta)$ in $\Coalg(H)$. Since $(X,\xi)$ is
  simple, $q$ is an isomorphism, and so $h = m\cdot q$ is a
  monomorphism. For sufficiency, consider $q\colon (X,\xi)\epito(Y,\zeta)$
  with~$q$ a regular epimorphism in~$\C$. By assumption, $q$ is also
  monic, whence an isomorphism.
\end{proof}
\noindent Intuitively, in a simple coalgebra all states exhibiting the
same observable behaviour are already identified. This paper is
concerned with the design of algorithms for computing \emph{the}
simple quotient of a given coalgebra:
\begin{lemma}\label{lemmaSimple}
  A simple quotient of a coalgebra is unique (up to
  isomorphism). Concretely, let $(X,\xi)$ be a coalgebra, and let
  $e_i\colon (X,\xi)\epito(Y_i,\zeta_i)$, $i=1,2$, be quotients with
  $(D_i,d_i)$ simple. Then $(D_1,d_1)$ and $(D_2,d_2)$ are isomorphic;
  more precisely, $e_1$ and $e_2$ represent the same quotient.
\end{lemma}
\begin{proof}
  By \autoref{cor:coalg-colims}, there is a pushout
  $Y_1\xrightarrow{f_1}Z\xleftarrow{f_2}Y_2$ of
  $Y_1\xleftarrow{e_1}X\xrightarrow{e_2}Y_2$ in $\Coalg(H)$. Since
  regular epimorphisms are generally stable under pushouts,
  $f_1$ and $f_2$ are regular
  epimorphisms, hence isomorphisms because $Y_1$ and $Y_2$ are simple;
  this proves the claim.
\end{proof}
\noindent Existence of the simple quotient can be shown under additional
assumptions on $\C$ (cf.~\autoref{R:final} below):
\begin{thmC}[\cite{Gumm03}]
  Assume that $\C$ is cocomplete and cowellpowered. Then every
  coalgebra $(X,\xi)$ has a simple quotient given by the cointersection
  (i.e.~the wide pushout) of all quotient coalgebras
  \[
    q\colon (X,\xi) \epito (X', \xi').
  \]
\end{thmC}

\noindent For $\C = \Set$, two elements $x\in X$ and $y\in Y$ of
coalgebras $(X,\xi)$ and $(Y, \zeta)$ are \emph{behaviourally
  equivalent} if they can be merged by coalgebra
morphisms, that is, if there exist a coalgebra $(Z,\omega)$ and
coalgebra morphisms $f\colon (X,\xi)\to (Z,\omega)$,
$g\colon (Y,\zeta) \to (Z,\omega)$ such that $f(x)=g(y)$. 
Intuitively, the simple quotient of a coalgebra in \Set is its quotient modulo
behavioural equivalence. In our main examples, this means that we minimize
w.r.t.~standard bisimilarity-type equivalences:

\begin{example} \label{exCoalgEquivalence}
  Behavioural equivalence instantiates to various notions of
  bisimilarity:
  \begin{enumerate}
  \item Park-Milner bisimilarity on labelled transition systems~\cite{aczelmendler:89};
  \item weighted bisimilarity on weighted transition systems~\cite[Proposition~2]{Klin09};
  \item stochastic bisimilarity on probabilistic transition systems~\cite{Klin09};
  \item Segala bisimilarity on simple and general Segala systems~\cite[Theorem~4.2]{BARTELS200357}.
  \end{enumerate}
\end{example}

\begin{remark}\label{R:final}
  A \emph{final coalgebra} is a terminal object in the category of
  coalgebras, i.e.~a coalgebra $(T,\tau)$ such that every coalgebra
  $(X,\xi)$ has a unique coalgebra morphism into $(T,\tau)$.  There
  are reasonable conditions under which a final coalgebra is
  guaranteed to exist, e.g.~when $\C$ is a locally presentable
  category (in particular, when $\C=\Set$) and $H$ is
  accessible~\cite{AdamekR94}. If $H$ preserves monomorphisms and has a final
  coalgebra $(T,\tau)$, then the simple quotient of a coalgebra
  $(X,\xi)$ is the image of $(X,\xi)$ under the unique morphism into
  $(T,\tau)$; in particular, in this case every coalgebra has a simple
  quotient.
\end{remark}

Finally, we note a useful result that implies that computing the
simple quotient of a coalgebra for $H$ can be reduced to computing the simple
quotient of its induced coalgebra for a superfunctor of $H$:
\smnote{This is a beautiful and general reduction result,
  that I'd put up front and not hide in some technical development
  on many-sorted sets! Besides it is needed for the correctness of
  automata minimization with $\Potf$ in Section~8.}
\begin{proposition}\label{P:reduction}
  Suppose that $m\colon H \monoto G$ is a natural transformation with
  monomorphic components. Then every $H$-coalgebra $\xi\colon X \to HX$
  and its induced $G$-coalgebra
  \[
    X \xrightarrow{\xi} HX \xrightarrow{m_X} GX
  \]
  have the same quotients and, hence, the same simple ones (if they exists). 
\end{proposition}
\begin{proof}
  We prove only the first claim. Let $q\colon X\epito Y$ be a regular
  epimorphism. It suffices to show that $q$ carries an $H$-coalgebra
  morphism with domain $(X,\xi)$ iff it carries a
  $G$-coalgebra morphism with domain
  $(X,m_X \cdot \xi)$. Note that $m\colon H \monoto G$ induces an
  embedding $\Coalg(H)\to\Coalg(G)$. Hence, `only
  if' is clear, and we prove `if'. Suppose that $q$ is a coalgebra
  morphism $(X, m_X\cdot \xi) \to (Y,\zeta)$. Then the outside of
  the following diagram commutes:
  \[
    \begin{mytikzcd}
      X
      \arrow{r}{\xi}
      \arrow[->>]{d}[swap]{q}
      & HX
      \arrow{r}{m_X}
      \arrow{d}[swap]{Hq}
      & GX
      \arrow{d}{G q}
      \\
      Y
      \arrow[dashed]{r}[near end]{\exists! \zeta'}
      \arrow{rr}[swap]{\zeta}
      & |[yshift=5mm]| HY
      \arrow[>->]{r}{m_Y}
      & GY
    \end{mytikzcd}
  \]
  By the naturality of $m$ the right-hand part commutes. Hence, since
  $m_Y$ is monic we obtain $\zeta'$ as in the diagram by the diagonal
  fill-in property (Section~\ref{sec:equiv}), making $q$ an
  $H$-coalgebra morphism $(X,\xi)\epito(Y,\zeta')$.
\end{proof}

\section{Partition Refinement from an Abstract Point of View}
\label{sec:fuzzyAlgo}
\tikzsetfigurename{partref-fuzzy-}


In the next section we will provide an abstract partition refinement
algorithm and formally prove its correctness.  Our main contribution
is \emph{genericity}: we are able to state and prove our results at a
level of abstraction that uniformly captures various kinds of state
based systems. This allows us to instantiate our efficient generic
algorithm to many different (combinations of) transition structures.
Before describing the abstract algorithm (\autoref{catPT}) formally,
we now give an informal description of a partition refinement
algorithm, which will make it clear which parts of partition
refinement algorithms are generic and which parts are specific to a
particular transition type. Although \autoref{catPT} works in
categorical generality, we use set-theoretic parlance in the present
informal discussion.

Given a system with a set~$X$ of states, one of the core ideas of the
known partition refinement algorithms mentioned so far, in particular
the algorithms by Hopcroft~\cite{Hopcroft71} and
Paige-Tarjan~\cite{PaigeTarjan87}, is to maintain two equivalence
relations~$P$ and~$Q$ on~$X$, represented by the corresponding
partitions $X/P$ and $X/Q$, where $X/P$ will be \textqt{one transition
  step ahead of $X/Q$}, so the relation $P$ is a refinement of
$Q$. Therefore, the elements of $X/P$ are called \emph{subblocks} and
the elements of $X/Q$ are called \emph{compound blocks}.

Initially, we put $X/Q = \{X\}$ and let $X/P$ be the initial partition
with respect to the \textqt{output behaviour} of the states in
$X$. For example, in the case of deterministic automata, the output
behaviour of a state is its finality, i.e.~the initial partition
separates final from non-final states; in the case of
transition systems, the initial partition separates deadlock states
from states with successors; and for weighted systems, the initial
partition groups the states by the sum of weights of outgoing edges.

\begin{figure}
  \newboolean{partRefExampleGraphShowLabels}
  \newcommand{\maybelabel}[1]{%
    \ifthenelse{\boolean{partRefExampleGraphShowLabels}}%
    {\ensuremath{#1}}%
    {}}%
  \newcommand{\partRefExampleGraph}[3]{
    \setboolean{partRefExampleGraphShowLabels}{#1}
  \begin{tikzpicture}[
    x=14mm,
    y=14mm,
    baseline=(x2.base),
    every node/.append style={
      inner sep = 1pt,
    },
    transitions/.style={
      -{>[length=2mm,width=2mm]},
      every loop/.append style={
        -{>[length=2mm,width=2mm]}, 
      },
      shorten <= 1pt,
      shorten >= 1pt,
      every node/.append style={
        above,
        font=\footnotesize,
        fill=white,
        inner sep=0pt,
        outer sep=2pt,
        text depth=0pt,
      },
    },
    transitions thick/.style={
      transitions,
      line width=3pt,
      draw=white, 
      every node/.append style={
        execute at begin node=\setboolean{partRefExampleGraphShowLabels}{false},
        execute at end node=\setboolean{partRefExampleGraphShowLabels}{#1},
      },
    },
    partitionQ/.append style={
      partitionBlock,
      inner sep = 5pt,
      dashed,
    },
    partitionP/.append style={
      partitionBlock,
      inner sep = 3pt,
      rounded corners=2.0mm,
    },
    ]
    \node (x0) at (-150:1) {$x_0$};
    \node (x1) at (150:1) {$x_1$};
    \node (x2) at (0,0) {$x_2$};
    \node (x3) at (30:1) {$x_3$};
    \node (x4) at (-30:1) {$x_4$};
    \foreach \style in {transitions thick,transitions} {
      \path[\style]
      (x0) edge node[left] {\maybelabel{\frac{1}{2}}} (x1)
      (x1) edge node[pos=0.3,above right] {\maybelabel{1}} (x2)
      (x0) edge node[pos=0.3,below right] {\maybelabel{\frac{3}{2}}} (x2)
      (x1) edge[loop above] node {\maybelabel{1}} (x1)
      (x2) edge node[pos=0.6,above left] {\maybelabel{1}} (x3)
      (x2) edge node[pos=0.6,below left] {\maybelabel{-1}} (x4)
      ;
    }
    \begin{scope}[on background layer]
      \foreach \block in {#2} {
        \node[partitionQ, fit=\block] {};
      }
      \foreach \block in {#3} {
        \node[partitionP, fit=\block] {};
      }
    \end{scope}
  \end{tikzpicture}
  }
  \centering
  \def\arraystretch{1.4}
  \newcommand{\legendPartition}[1]{
    \begin{tikzpicture}[baseline=(node.base),every node/.append style={
        partitionBlock,
        inner sep=2pt,
        rounded corners=1.0mm,
      }]
      \node[dashed] (node) {\(X/Q_{#1}\)};
      \node[solid,anchor=west] at ([xshift=6mm]node.east){\(X/P_{#1}\)};
    \end{tikzpicture}
  }
  \begin{tabular}{@{}lccc@{}}
    \toprule
    & \makecell{Initial Partitions \\ \legendPartition{0}}
    & \makecell{After the first iteration \\ \legendPartition{1}}
    & \makecell{After the second iteration \\ \legendPartition{2}}
      \\
      \midrule
  \begin{subfigure}{6mm}
    \caption{}
    \label{fig:exPartRef:LTS}
  \end{subfigure}
  & \partRefExampleGraph{false}{(x0) (x1) (x2) (x3) (x4)}{(x0) (x1) (x2),(x3) (x4)}
  & \partRefExampleGraph{false}{(x0) (x1) (x2),(x3) (x4)}{(x0) (x1),(x2),(x3) (x4)}
  & \partRefExampleGraph{false}{(x0) (x1),(x2),(x3) (x4)}{(x0) (x1),(x2),(x3) (x4)}
  \\
  \begin{subfigure}{6mm}
    \caption{}
    \label{fig:exPartRef:Markov}
  \end{subfigure}
  \rule{0pt}{18mm} 
  & \partRefExampleGraph{true}{(x0) (x1) (x2) (x3) (x4)}{(x0) (x1),(x2) (x3) (x4)}
  & \partRefExampleGraph{true}{(x0) (x1),(x2) (x3) (x4)}{(x0),(x1),(x2) (x3) (x4)}
  & \partRefExampleGraph{true}{(x0),(x1),(x2) (x3) (x4)}{(x0),(x1),(x2) (x3) (x4)}
    \\
  \bottomrule
  \end{tabular}
  \caption{Example of the partition refinement in \textsc{(a)}
    transition systems ($H=\Potf$) and  \textsc{(b)} $\R$-weighted systems ($H=\R^{(-)}$). The partition
    $X/Q_i$ is indicated by dashed lines, and $X/P_i$ by solid lines.}
  \label{fig:exPartRef}
\end{figure}

\begin{algorithm}[Informal Partition Refinement] \label{informalPT}
  Given a system on $X$ and initial partitions $X/Q=\{X\}$ and $X/P$
  as above, iterate the following while~$P$ is different from~$Q$:
\begin{enumerate}[ref=(\arabic*)]
\item \label{stepInformalS}
  Pick 
  a subblock $S$ in $X/P$ that
  is properly contained in a compound block $C \in X/Q$, i.e.~$S\subsetneqq C$.
  Note that this choice represents a quotient $q\colon X\epito \{S, C\setminus
    S, X\setminus C\}$.

  \item \label{stepInformalQ} Refine $X/Q$ with respect to this
    quotient by splitting $C$ into two blocks $S$ and $C\setminus S$;
    that is, replace $X/Q$ with the intersection (greatest common
    refinement) of the partitions $X/Q$ and
    $\{S,C\setminus S, X\setminus C\}$ according to the terminology
    introduced in \autoref{sec:prelim}.

  \item \label{stepInformalP} Update $X/P$ to the coarsest refinement
    of~$X/P$ in which any two states are distinguished if the
    transition structure distinguishes them up to~$Q$, equivalently:
    If two states are identified, then the transition structure
    identifies them up to~$Q$.
  \end{enumerate}
\end{algorithm}
\begin{remark}
  \begin{enumerate}
  \item Since $X/Q$ is refined using only information from $X/P$ in
    each iteration, the algorithm maintains the invariant that the
    partition $X/P$ is finer than $X/Q$.

  \item The informal property of the refinement
    of~$X/P$ in step~(3) is referred to as \emph{stability} of $P$
    w.r.t.~$Q$. The formal definition of this notion takes into account
    the transition type of the given system. In fact,  Paige and
    Tarjan~\cite{PaigeTarjan87} have defined stability for transition
    systems, and we generalize their definition to the level of
    coalgebras in \autoref{D:stable}.
  \end{enumerate}
\end{remark}

\noindent Note that Steps~(1) and~(2) are independent of the given
transition type, as they perform only basic operations on
quotients. In contrast, the initialization procedure and Step~(3)
depend on the transition type of the specific system, encoded by the
type functor.

\begin{example} \label{exPartRef} \autoref{fig:exPartRef} illustrates
  runs of partition refinement algorithms and how the transition type
  specific steps are handled for two different types of systems: in
  Paige and Tarjan's algorithm for transition
  systems~\cite{PaigeTarjan87} and in Valmari and Franceschinis'
  Markov chain lumping algorithm, which in fact works for
  $\R$-weighted systems~\cite{ValmariF10}:
  \begin{enumerate}[label=(\textsc{\alph*})]
  \item In transition systems, the initial partition only
    distinguishes between deadlocks and live states; in our example
    $X/P_0 = \{\{x_0,x_1,x_2\},\{x_3,x_4\}\}$, and $X/Q_0 = \{X\}$
    identifies all states.
    \begin{itemize}
    \item In the first iteration of the loop of \autoref{informalPT}, there are
      two choices for $S$ in Step~\ref{stepInformalS}, $\{x_0,x_1,x_2\}$ and
      $\{x_3,x_4\}$, both leading to the same refinement step: In
      Step~\ref{stepInformalQ} the only block of $X/Q_0$ is split into two
      blocks, and so $X/Q_1 = X/P_0$. In Step~\ref{stepInformalP}, $X/P_1$ is
      the coarsest refinement of $X/P_0$ that is stable w.r.t.~$X/Q_1$. In
      transition systems, $X/P$ is stable w.r.t.~$X/Q$ if for all elements in
      the same block $x,x'\in B\in X/P$ and any other block $B'\in X/Q$, $x$ has
      an edge to $B'$ iff $x'$ has an edge to $B'$~\cite{PaigeTarjan87}. In
      detail, the state $x_2$ becomes a separate block because $x_2$ has a
      transition to $\{x_3,x_4\}\in X/Q_1$ whereas $x_0$ and $x_1$ do not. The
      states $x_0$ and $x_1$ remain in the same block in $X/P_1$, because both
      have a transition to $\{x_0,x_1,x_2\}\in X/Q_1$, and have no transition to
      $\{x_3,x_4\} \in X/Q_1$. So $X/P_1 = \{\{x_0,x_1\},\{x_2\},\{x_3,x_4\}\}$.
    \item In the second iteration, we have symmetric choices
      $\{x_0,x_1\}$ and $\{x_2\}$ for $S$, and both split the compound
      block $\{x_0,x_1,x_2\}$ so that $X/Q_2 = X/P_1$. The partition
      $X/Q_2$ is stable w.r.t.~itself because both $x_0$ and $x_1$
      have transitions to $\{x_0,x_1\}$ and $\{x_2\}$, and because
      $x_3$ and $x_4$ both are deadlock states (for the singleton
      block $\{x_2\}$ there is nothing to check). Hence,
      $X/P_2 = X/Q_2$ and the algorithm terminates.
    \end{itemize}

  \item In weighted systems, the initial partition groups the states
    by the sum of the weights of their outgoing
    transitions.\footnote{For actual Markov chains or probabilistic
      transition systems, where weights sum up to~$1$, the initial
      partition would thus be trivial, causing immediate termination;
      this corresponds to the fact that all states of a Markov chain
      are behaviourally equivalent unless we introduce additional
      observable features such as propositional atoms or deadlock. In
      standard treatments of Markov chain lumping, such additional
      features are abstracted in the choice of a non-trivial initial
      partition.} For $x_0$ and $x_1$, the sum is 2, and for $x_2$,
    $x_3$, $x_4$ the sum is 0. Hence
    $X/P_0 =\{\{x_0,x_1\},\{x_2,x_3,x_4\}\}$, and as always,
    $X/Q_0 = \{X\}$ identifies all states.
    \begin{itemize}
    \item In the first iteration of \autoref{informalPT}, we can choose
      $\{x_0,x_1\}$ or $\{x_2,x_3,x_4\}$ for $S$, both leading to $X/Q_1 =
      X/P_0$. In weighted systems, stability of $X/P$ w.r.t.~$X/Q$ means that
      all elements of the same block $x,x'\in B\in X/P$ have the same
      accumulated transition weight to any other block $B'\in X/Q$, for $f\colon
      X\to \R^{(X)}$ this means $\sum_{y\in B'} f(x)(y) = \sum_{y\in
        B'}f(x')(y)$ (called `compatibility' in \cite{ValmariF10}). When
      refining $X/P_1$ to be as coarse as possible and stable w.r.t.~$X/Q_1$, we
      have the following accumulated sums. The transition from $x_2$ to
      $\{x_2,x_3,x_4\} \in X/Q_1$ has weight~$0$, just like the (non-existent)
      transition from $x_3$ (resp.~$x_4$) to $\{x_2,x_3,x_4\} \in X/Q_1$, and
      hence $x_2, x_3, x_4$ are identified in $X/P_1$. The transition from $x_0$
      to $\{x_2,x_3,x_4\}\in X/Q_1$ has weight $\frac{3}{2}$, but the transition
      from $x_1$ to $\{x_2,x_3,x_4\}$ has weight $1$. Thus, $x_0$ and $x_1$ are
      split in $X/P_1$, and therefore we have $X/P_1 =
      \{\{x_0\},\{x_1\},\{x_2,x_3,x_4\}\}$.

    \item In the second iteration, we can choose $S$ to be $\{x_0\}$ or
      $\{x_1\}$. In either case, $X/Q_2 = X/P_1$, and we have that $X/Q_2$ is
      stable w.r.t.~$X/Q_2$ because the transition from $x_2$ to
      $\{x_2,x_3,x_4\}$ in $X/Q_2$ has weight 0. Thus, $X/P_2 = X/Q_2$ and the
      algorithm terminates.
    \end{itemize}
  \end{enumerate}
\end{example}
\noindent In the remainder of the paper we shall see that one can
unify the similarities between \autoref{fig:exPartRef}(\textsc{a}) and
\autoref{fig:exPartRef}(\textsc{b}) into a generic algorithm on a
coalgebraic level, where the specifics of the transition type are
hidden within the coalgebraic type functor.

\section{A Categorical Algorithm for Behavioural Equivalence}
\label{sec:cat}
\tikzsetfigurename{partref-catPT-}

We proceed to give a formal description of a categorical partition
refinement algorithm that computes the simple quotient of a given
coalgebra under fairly general assumptions.

\begin{assumption}\label{ass:sec3}
  In addition to \autoref{ass:C}, we fix an endofunctor $H\colon \C\to \C$ that
  preserves monomorphisms.
\end{assumption}

\begin{remark}
  For $\C = \Set$, the assumption that $H$ preserves monomorphisms is
  w.l.o.g. First note that every endofunctor on \Set preserves
  nonempty monomorphisms. Moreover, for every set functor $H$ there exists a
  mono-preserving set functor $H'$ that is identical to $H$ on the
  full subcategory of all nonempty
  sets~\cite[Theorem~3.4.5]{AdamekT90}. Hence, $H'$ has essentially
  the same coalgebras as $H$ since there is only one coalgebra
  structure on $\emptyset$.
\end{remark}

\noindent For a given coalgebra $\xi\colon X\to HX$ in $\Set$, any
partition refinement algorithm should maintain a quotient
$q\colon X\twoheadrightarrow \unnicefrac{X}{Q}$ that distinguishes
some (but possibly not all) states with different behaviour, and in
fact, initially~$q$ typically identifies everything. Using the
language of universal coalgebras, one can express the transition
type specific steps from \autoref{informalPT} generically. Informally,
our algorithm repeats the following steps until it stabilizes:
\begin{enumerate}
\item 
  Analyse $X\xrightarrow{\xi} HX \xrightarrow{Hq} H(\unnicefrac{X}{Q})$
  to identify equivalence classes w.r.t.~$q$ (i.e.~compound blocks) containing
  states that exhibit distinguishable behaviour when considering one more step of
  the transition structure $\xi$.
\item Use parts of this information to refine $q$.
\end{enumerate}
Here, the composite $Hq\cdot\xi$ (more precisely its kernel) defines
the finer partition $\unnicefrac X P$ in \autoref{informalPT} (our
algorithm will guarantee that this does refine the previous value
of~$P$). More precisely,~$P$ will be defined to be the kernel of
$Hq\cdot \xi$, and so $P$ is as coarse as possible to be stable
w.r.t.~$q\colon X\twoheadrightarrow X/Q$, in the following sense:
\begin{definition}\label{D:stable}
  Given a coalgebra $\xi \colon X\to HX$, a kernel
  $P \rightrightarrows X$ is said to be \emph{stable}
  w.r.t.~a morphism $q\colon X\to Y$ provided that there exists a morphism
  $\xi/q\colon X/P \to HY$ such that the following square
  commutes:
  \[
    \begin{tikzcd}
      X
      \arrow{r}{\xi}
      \arrow[->>]{d}[swap]{\kappa_{P}}
      &
      HX
      \arrow{d}{Hq}
      \\
      X/P
      \arrow{r}{\xi/q}
      &
      HY
    \end{tikzcd}
  \]
  Equivalently, the kernel $\pi_1,\pi_2\colon P\rightrightarrows X$ is
  stable w.r.t.~$q$ if $Hq\cdot \xi \cdot \pi_1 = Hq\cdot \xi \cdot
  \pi_2$.
\end{definition}
\begin{rem}
  Note that $P \rightrightarrows X$ is stable
  w.r.t.~$\kappa_P\colon X \epito X/P$ iff $\kappa_P$ is a coalgebra
  morphism and thus represents a quotient of $(X, \xi)$.
\end{rem}
\begin{example}
  This definition of stability for a coalgebra $\xi\colon X\to HX$ matches the
  concrete instances we have seen in \autoref{exPartRef}:
  \begin{enumerate}
  \item In a transition system $\xi\colon X\to \Potf X$, $P\rightrightarrows X$ is stable
    w.r.t.~$\kappa_Q\colon X\twoheadrightarrow X/Q$ if for all elements in the
    same block $x,x'\in B\in X/P$ and any other block $B'\in X/Q$, $x$ has an
    edge to $B'$ iff $x'$ has an edge to $B'$~\cite{PaigeTarjan87}.
  \item In a weighted system $\xi\colon X\to \R^{(X)}$, stability of $P
    \rightrightarrows X$ w.r.t.~$X/Q$ means that all elements of the same block
    $x,x'\in B\in X/P$ have the same accumulated transition weight to any other
    block $B'\in X/Q$, i.e.~$\sum_{y\in B'} \xi(x)(y) = \sum_{y\in B'}\xi(x')(y)$.
    This is called `compatibility' in \cite{ValmariF10}.
  \item For LTSs, Blom and Orzan~\cite{BlomO05} define the \emph{signature} of
    state $x\in X$ with respect to $q$ as $Hq\cdot \xi(x)$ (with
    $HX=\Pot(A\times X)$) and then define a partition $X/P$ to be stable if
    every two members $(x,x')\in P$ have the same signature w.r.t.~$q:=\kappa_P$.
  \end{enumerate}
\end{example}

The refinement step~(2) above corresponds to
the subblock selection in \autoref{informalPT}. This selection will be
encapsulated at the present level of generality in a routine
$\op{select}$, assumed as a parameter of our algorithm:
\begin{definition}
  A \op{select} routine is an operation that receives a chain of two
  regular epis
  \(
    \begin{mytikzcd}
      |[inner sep=0mm]|
      X\  \arrow[->>]{r}{y}
      & Y \arrow[->>]{r}{z}
      & Z
    \end{mytikzcd}
    \)
    and returns some morphism $\op{select}(y,z)\colon Y \to K$. We
    refer to~$Y$ and~$Z$ as the objects of \emph{subblocks} and
    \emph{compound blocks}, respectively.
\end{definition}
\noindent In our algorithm,~$y$ will represent the canonical quotient
$X\twoheadrightarrow X/P$ and~$z$ the canonical map
$X/P\twoheadrightarrow X/Q$ given by the invariant that $P$ is finer
than $Q$.  Intuitively, the morphism $\op{select}(y,z)\colon Y\to K$
selects some of the information contained in $Y$ and discards all the
remaining information by identifying the remaining elements. For
example, in the Paige-Tarjan algorithm with $Y=X/P$ and $Z=X/Q$,
$\op{select}$ models the selection of only a single block $S \in X/P$
leading to a split of the surrounding block $C\in X/Q, S\subseteq C$
into two blocks $S$ and $C\setminus S$. In this case,
$\op{select}(y,z)$ is thus essentially a characteristic function of
the following shape:
\begin{definition} \label{defChi3}
  For sets $S\subseteq C\subseteq X$, define the map
  \[
    \chi_S^C\colon X\to 3
    \qquad
    3 = \{0,1,2\}
    \qquad
    \chi_S^C(x) =
    \begin{cases}
      2 &\text{if }x \in S \\
      1 & \text{if }x \in C\setminus S \\
      0 & \text{if }x \in X\setminus C.
    \end{cases}
  \]
  This is a three-valued version of the characteristic function
  $\chi_{M}\colon X\to 2$ for a subset $M\subseteq X$. Put
  differently, $\chi_S^C$ is the codomain restriction of
  $\fpair{\chi_S,\chi_C}\colon X\to 2\times 2$ obtained by leaving out
  the impossible case of $x\in S\setminus C$. 
\end{definition}
\begin{example} \label{exampleSelects}
  We present some examples of \op{select} routines. Throughout, we fix
  a chain 
  \(
    \begin{mytikzcd}[inlinecd]
      |[inner sep=0mm]|
      X\  \arrow[->>]{r}{y}
      & Y \arrow[->>]{r}{z}
      & Z
    \end{mytikzcd}
    \)%
    of quotients.
\begin{enumerate}
\item \label{exampleSelectsChi} In Hopcroft's
  algorithm~\cite{Hopcroft71}, and in all the known efficient
  partition refinement algorithms mentioned so far, the goal is to
  find a proper subblock that is at most half the size of the compound
  block it is contained in. The optimized version of our algorithm
  over $\C = \Set$ (\autoref{sec:efficient}) will use the same
  strategy, embodied in the following $\op{select}$ routine. First, we
  identify the elements of $Y$ with the corresponding equivalence
  classes in $\unnicefrac{X}{\ker y}$ and those of $Z$ with
  equivalence classes in $X/\ker(z \cdot y)$. Now let $S$ be a
  subblock, i.e.~$S \in \unnicefrac{X}{\ker y} \cong Y$, such that its
  compound block,
  i.e.~$C = z(S) \in \unnicefrac{X}{\ker(z \cdot y)} \cong Z$
  satisfies $2\cdot |S| \le |C|$. Then we put (using our
  notation~\eqref{eq:funEquivClass} for equivalence classes induced by
  the map $z$):
  \begin{equation}\label{eq:select}
    \op{select}(y,z) = \chi_{\{S\}}^{[S]_z}\colon Y\to 3.
  \end{equation}
  Note that we then have
  \[
    \op{select}(y,z)\cdot y = \chi_S^C\colon X\to 3.
  \]
  If there is no such $S \in Y$, then $X$ is infinite or $z$ is an isomorphism,
  and in either case we simply put $\op{select}(y,z) = \id_Y$.
  
\item \label{exampleSelectsId} One obvious choice for
  $\op{select}(y,z)\colon Y\to K$ is the identity on $Y$, so that
  \emph{all} of the information present in $Y$ is used for further
  refinement. We will explain in \autoref{finalchain} how under this
  choice, our algorithm instantiates to König and Küpper's final chain
  algorithm~\cite{KonigKupper14}.

\item Two other, trivial, choices are
  $\op{select}(y,z)\colon Y\overset{!}\to 1$ and
  $\op{select}(y,z)=z$. Since both of these choices discard all the
  information in $Y$, this will leave the partitions computed by the
  algorithm unchanged, see the proof of \autoref{thm:correct}.
\end{enumerate}
\end{example}
\noindent Given a \op{select} routine, the most general form of our
partition refinement algorithm works as follows. Given a coalgebra
$\xi\colon X\to HX$, we successively refine equivalence relations
(i.e.~kernel pairs)~$Q \rightrightarrows X$ and
$P \rightrightarrows X$, maintaining the invariant that $P$ is finer
than $Q$ (cf.~\autoref{PfinerthanQ}), which is witnessed by a
(necessarily unique) morphism $f\colon X/P \epito X/Q$ such that
$f \cdot \kappa_P = \kappa_Q$, where $\kappa_P$ and $\kappa_Q$ denote
the canonical quotients (i.e.~coequalizers) of the above two kernel
pairs (\autoref{sec:equiv}).

Before each iteration of the main loop, we take into account new
information on the behaviour of states, represented by a morphism
$q\colon X\to K$, and accumulate this information in a morphism
$\bar q\colon X\to \bar K$ where $\bar K$ is the Cartesian product of
the instances of~$K$ encountered up to the present iteration. In order
to facilitate the analysis later, we index the variables
$P,Q,f,q,\bar q$ over loop iterations~$i$ in the description. For
brevity, we will just write the objects $Q_i$ and $P_i$ in lieu of the
respective kernel pairs $P_i \rightrightarrows X$ and
$Q_i \rightrightarrows X$.

\begin{algorithm} \label{catPT}
  Given a coalgebra $\xi\colon X\to HX$ and a $\op{select}$ routine,
  initially put
  \begin{equation*}
    Q_0  = X \times X,\qquad
    \bar q_0  = \mathbin{!}\colon X \to 1 = K_0,\qquad
    P_0 =  \ker(
    X \xrightarrow{\xi}
    HX \xrightarrow{H!} H1).
  \end{equation*}
  Then iterate the following steps, with~$i$ counting iterations starting
  at~$0$, while $P_i$ is properly finer than $Q_i$:
  \begin{enumerate}
  \item\label{step:S}
    $q_{i+1} := {\!\!\begin{mytikzcd}[inlinecd,sep=2.5em]
      X \arrow[->>]{r}{\kappa_{P_i}}
      & \unnicefrac{X}{P_i} \arrow{r}{\op{select}(\kappa_{P_i},f_i)}
      &[10mm] K_{i+1}
    \end{mytikzcd}\!\!}$,
    \quad$\bar q_{i+1} := \fpair{\bar q_i,q_{i+1}}\colon
    \begin{mytikzcd}[inlinecd]
      X \arrow{r}
      & \prod_{j \le i} K_i × K_{i+1}
    \end{mytikzcd}$
    \\
    where $f_i\colon X/P_i\twoheadrightarrow X/Q_i$ witnesses that
    $P_i$ is finer than $Q_i$.

  \item\label{step:Q} $Q_{i+1} := \ker\bar q_{i+1} = Q_i\cap \ker q_{i+1}$

  \item\label{step:P}
    $P_{i+1} := \ker\big(\!\begin{mytikzcd}[inlinecd] X \arrow{r}{\xi}
      & HX \arrow{rr}{H\bar q_{i+1}} && H\prod_{j\le i+1} K_j
    \end{mytikzcd}\!\big)$ 

  \end{enumerate}
  Upon termination of the above loop return
  $\unnicefrac{X}{P_i}=\unnicefrac{X}{Q_i}$.
\end{algorithm}
\begin{remark}
  Note that \autoref{catPT} is precisely the informal \autoref{informalPT}
  where
  \begin{itemize}
  \item the partitions and equivalence relations are replaced by
    coequalizers and kernel pairs;
  \item $P_i$ being properly finer than $Q_i$ means that the canonical
    morphism $f_i\colon X/P_i \epito X/Q_i$ is not an isomorphism; 
  \item in Step~\ref{step:Q}, we have
    $\ker \bar q_{i+1} = \ker \fpair{\bar q_i,q_{i+1}}= \ker \bar q_i
    \cap \ker q_{i+1} = Q_i\cap \ker q_{i+1}$, see~\eqref{eq:kercap};
  \item transition type specific steps involve how the type functor
    $H$ acts on morphisms;
  \item Step~\ref{step:P} makes $P_{i+1}$ stable w.r.t.~$\bar q_{i+1}$;
  \item the choice of a subblock $S$ and a compound block $C$ is
    performed by the $\op{select}$ routine.
  \end{itemize}
\end{remark}
\noindent In general, \autoref{catPT} need not terminate, but we
present sufficient conditions for termination in
\autoref{thm:correct}. We now proceed to show that when
\autoref{catPT} terminates, then it returns the (carrier object of
the) simple quotient of $(X,\xi)$, i.e.~we prove correctness. We
continue to use the notation established in
Algorithm~\ref{catPT}. Since $\bar q$ accumulates more information in
every step, it is clear that $P$ and $Q$ are being successively
refined:

\begin{lemma}\label{PfinerthanQ} 
  For every $i$, we have monomorphisms
  $P_{i+1} \monoto P_i \monoto Q_{i+1} \monoto Q_{i}$ witnessing
  inclusions of relations.
\end{lemma}

\noindent (The above inclusions state that the kernel $P_{i+1}
\rightrightarrows X$ is finer than the kernel $P_i \rightrightarrows X$ etc., see \autoref{sec:equiv}.)
%
%
\begin{proof}
  We have that $Q_{i+1}$ is finer than $Q_i$ by definition
  and use \autoref{rem:kernel}\ref{i:comp} for the remaining inclusions:
\begin{enumerate}
\item\label{PfinerthanQ:1} $P_{i+1}\monoto P_i$: Let $p\colon \prod_{j\le i+1} K_j \to \prod_{j\le i} K_j$ be the product projection.
  We have
  \(
  H\bar q_i \cdot \xi = H p \cdot H\bar q_{i+1} \cdot
  \xi
  \),
  so $P_{i+1} = \ker (H\bar q_{i+1}\cdot \xi)$ is finer than $P_i = \ker(H\bar q_i\cdot
  \xi)$.

\item $P_i\monoto Q_{i+1}$: First observe that for every $i$, $P_i$ is
  finer than $\ker q_{i+1}$ because $q_{i+1}$ factors through
  $\kappa_{P_i}\colon X\epito\unnicefrac{X}{P_i}$ by step \ref{step:S}
  in \autoref{catPT}. We now obtain the desired result by induction on
  $i$. In the base case we have that $P_0$ is finer than
  $Q_1 = \ker \fpair{\bar q_0, q_1} = \ker \fpair{!, q_1} = \ker q_1$.
  For the induction step ($i> 0$), since
  $Q_{i+1}=Q_i\cap \ker q_{i+1}$, it suffices to show that~$P_i$ is
  finer than $Q_i $ and $\ker q_{i+1}$. The latter is just our lead-in
  observation, and for the former use the inductive hypothesis
  ($P_{i-1}\monoto Q_i$) and the fact that,
  by~\ref{PfinerthanQ:1},~$P_i$ is finer than~$P_{i-1}$.
\qedhere
\end{enumerate}
\end{proof}
\noindent
One of the key ingredients in the correctness proof is that the
partition $X/P_i$ is one ``transition step'' ahead of $X/Q_i$,
i.e.~the kernel $P_i$ is stable w.r.t.~the quotient $\kappa_{Q_i}$:%
\smnote{The extra property that $\xi/Q_i$ is monic is used
  later! So please do not reformulate as `$P_i$ is stable
  w.r.t.~$Q_i$' because more is proved here and needed later.}%
\\
\begin{minipage}[c]{.55\textwidth}
\begin{proposition}\label{propQuot}
  There exist monomorphisms
  $\unnicefrac{\xi}{Q_i}\colon \unnicefrac{X}{P_i} \monoto H(\unnicefrac{X}{Q_i})$
  for $i\ge 0$ (necessarily unique) such that \eqref{eq:xiQuotient}
  commutes.
\end{proposition}
\end{minipage}\hfill
\raisebox{2mm}{
\begin{minipage}[c]{.4\textwidth}
\begin{equation}
    \begin{mytikzcd}[
            row sep=4mm,
            column sep=12mm,
        ]
        X
            \arrow[->>]{d}[swap]{\kappa_{P_i}}
            \arrow{r}{\xi\ }
        &  HX
            \arrow[->]{d}[yshift=1pt]{H\kappa_{Q_i}}
        \\
        \unnicefrac{X}{P_i}
            \arrow[dash pattern=on 2pt off 1pt,>->]{r}{\ \unnicefrac{\xi}{Q_i}}
        & H (\unnicefrac{X}{Q_i})
    \end{mytikzcd}
    \label{eq:xiQuotient}
  \end{equation}
\end{minipage}}

\begin{proof}
  Since $Q_i=\ker\bar q_i$, the image factorization of $\bar q_i$ has
  the form
  \[
    \bar q_i=
    \big(\!\!
    \begin{mytikzcd}
      X \arrow[->>]{r}{\kappa_{Q_i}}
      & X/Q_i
      \arrow[>->]{r}{m}
      & \prod_{j\le i} K_j
    \end{mytikzcd}\!\!
    \big).
  \]
  By definition of $P_i$ and
  because $H$ preserves monos, we thus have
  $P_i = \ker(H\bar q_i\cdot \xi)=\ker(H\kappa_{Q_i}\cdot\xi)$, and
  hence we obtain $\unnicefrac{\xi}{Q_i}$ as in~\eqref{eq:xiQuotient} by
  the universal property of the coequalizer~$\kappa_{P_i}$:
  \[
    \begin{tikzcd}
      P_i
      \arrow[xshift=-1mm]{d}\arrow[xshift=1mm]{d}
      \\
      X
      \arrow[->>]{d}[swap]{\kappa_{P_i}}
      \arrow{r}{\xi\ }
      &
      HX
      \arrow[->]{d}[yshift=1pt]{H\kappa_{Q_i}}
      \arrow{rd}{H\bar q_i}
      \\
      \unnicefrac{X}{P_i}
      \arrow[dash pattern=on 2pt off 1pt,>->]{r}{\ \unnicefrac{\xi}{Q_i}}
      &
      H (\unnicefrac{X}{Q_i})
      \arrow[>->]{r}{Hm}
      &
      H\big(\prod_{j\le i} K_j\big)
    \end{tikzcd}
  \]
  In fact, $\kappa_{P_i}$ is the regular-epi part of the factorization of $H\bar
  q_i\cdot \xi$, and so $Hm\cdot \xi/Q_i$ is the mono part, and thus $\xi/Q_i$
  is also a monomorphism.
\end{proof}

\begin{corollary}
  If $P_i = Q_i$ for some $i$, then $\unnicefrac{X}{Q_i}$ carries a unique coalgebra
  structure making $\kappa_{P_i}$ a coalgebra morphism.
\end{corollary}
\noindent For $\C = \Set$, this means that all states of $X$ that are
merged by the algorithm are actually behaviourally equivalent. We
still need to prove the converse, namely that all behaviourally equivalent
states are indeed identified in $X/Q_i$:

\takeout{
\begin{lemma}
    \label{soundness}
    Let $h\colon (X,\xi)\to (D,d)$ represent a quotient of $(X,\xi)$. Then
    $\ker h$ is finer than both~$P_i$ and $Q_i$, for all $i\ge 0$.
\end{lemma}
\begin{proof}
  We claim that
  \begin{equation}\label{eq:hq-to-hp}
    \text{if $\ker h$ is finer than $Q_i$, then $\ker h$ is finer than $P_i$}.
  \end{equation}
  This is seen as follows: If $\ker h$ is finer than $Q_i$, then
  $\kappa_{Q_i}\colon X\to\unnicefrac{X}{Q_i}$ factorizes through $h\colon X\epito D$, so
  that $H\kappa_{Q_i}\cdot\xi$ factorizes through $Hh\cdot\xi$ and hence
  through $h$, since $Hh\cdot\xi=d\cdot h$:
  \[
    \begin{mytikzcd}
      X
      \arrow{r}{\xi}
      \arrow{d}{h}
      &
      HX
      \arrow{d}{Hh}
      \arrow{r}{H\kappa_{Q_i}}
      &
      H(\unnicefrac{X}{Q_i})
      \\
      D
      \arrow{r}{d}
      &
      HD
      \arrow[swap]{ru}{Hq}
    \end{mytikzcd}
  \]
  Since $P_i=\ker(H\kappa_{Q_i}\cdot\xi)$, this implies that $\ker h$
  is finer than $P_i$.  The claim of the lemma is then proved by
  induction: for $i=0$, the claim for $Q_0=X\times X$ is trivial, and
  the one for $P_0$ follows by~\eqref{eq:hq-to-hp}. The inductive step
  follows from \autoref{PfinerthanQ}: $\ker(h) \monoto P_i \monoto
  Q_{i+1}$, whence we are done by~\eqref{eq:hq-to-hp}.
\end{proof}}

\begin{theorem}[Correctness] \label{correctness}
  If $P_i = Q_i$ for some $i$, then $\unnicefrac{\xi}{Q_i}\colon \unnicefrac{X}{Q_i}\to
  H(\unnicefrac{X}{Q_i})$ is a simple coalgebra.
\end{theorem}
\begin{proof}
  Let $h\colon (X,\xi) \epito (D,d)$ represent a quotient.

  \begin{enumerate}
  \item We first prove that for all $i\ge 0$, 
  \begin{equation}\label{eq:hq-to-hp}
    \text{if $\ker h$ is finer than $Q_i$, then $\ker h$ is finer than $P_i$}.
  \end{equation}
  This is seen as follows: If $\ker h$ is finer than $Q_i$, then
  $\kappa_{Q_i}\colon X\epito\unnicefrac{X}{Q_i}$ factorizes through
  $h\colon X\epito D$, i.e.~we have some
  $q\colon D \epito \unnicefrac{X}{Q_i}$ such that
  $q \cdot h = \kappa_{Q_i}$ (see
  \autoref{rem:kernel}\ref{i:finer}). So $H\kappa_{Q_i}\cdot\xi$
  factorizes through $Hh\cdot\xi$ and hence through $h$, since
  $Hh\cdot\xi=d\cdot h$:
  \[
    \begin{mytikzcd}
      X
      \arrow{r}{\xi}
      \arrow[->>]{d}[swap]{h}
      &
      HX
      \arrow{d}{Hh}
      \arrow{r}{H\kappa_{Q_i}}
      &
      H(\unnicefrac{X}{Q_i})
      \\
      D
      \arrow{r}{d}
      &
      HD
      \arrow[swap]{ru}{Hq}
    \end{mytikzcd}
  \]
  Since $P_i=\ker(H\kappa_{Q_i}\cdot\xi)$, this implies that $\ker h$
  is finer than $P_i$, again by \autoref{rem:kernel}\ref{i:finer}.

\item Next we prove by induction on $i$ that $\ker h$ is finer than
  both~$P_i$ and $Q_i$, for all $i\ge 0$. For $i=0$, the claim for
  $Q_0=X\times X$ is trivial, and the one for $P_0$ follows
  by~\eqref{eq:hq-to-hp}. In the induction step, we have by the
  inductive hypothesis that $\ker(h)$ is finer than $P_i$, thus by
  \autoref{PfinerthanQ} also finer than $Q_{i+1}$ and consequently
  by~\eqref{eq:hq-to-hp} finer than $P_{i+1}$.

\item Now we are ready to prove the claim of the theorem. Let
  $q\colon (\unnicefrac{X}{Q_i},\unnicefrac{\xi}{Q_i}) \epito (D,d)$
  represent a quotient. Then
  $q\cdot \kappa_{Q_i}\colon (X,\xi) \to (D,d)$ represents a quotient
  of $(X,\xi)$, so by point~(2) above, $\ker( q\cdot\kappa_{Q_i})$ is
  finer than $Q_i$. By \autoref{rem:kernel}\ref{i:comp},
  $Q_i = \ker(\kappa_{Q_i})$ is also finer than
  $\ker( q\cdot\kappa_{Q_i})$, so
  $\ker(q\cdot \kappa_{Q_i}) = \ker(\kappa_{Q_i}) =Q_i$.  This implies
  that $\kappa_{Q_i}\colon X\epito X/Q_i$ is the regular epi part of
  the image factorization of $q\cdot \kappa_{Q_i}$, i.e.~we have
  $m \cdot \kappa_{Q_i} = q \cdot \kappa_{Q_i}$ for some monomorphism
  $m$. Since $\kappa_{Q_i}$ is an epimorphism, we obtain $m = q$,
  i.e.~$q$ is a monomorphism, and hence an isomorphism. \qedhere
\end{enumerate}
\end{proof}

\begin{remark} \label{coalgebraInitialPartition}
  Most classical partition refinement algorithms are parametrized by an initial
  partition $\kappa_{\mathcal{I}}\colon X\twoheadrightarrow
  \unnicefrac{X}{\mathcal{I}}$. We start with the trivial partition $!\colon X \to 1$
  because a non-trivial initial partition might split equivalent behaviours and
  then would invalidate \autoref{correctness}. To accommodate an initial partition
  $\unnicefrac{X}{\mathcal{I}}$ coalgebraically, replace $(X,\xi)$ with the
  coalgebra $\fpair{\xi,\kappa_\mathcal{I}}$ for the functor
  $H(-)×\unnicefrac{X}{\mathcal{I}}$ -- indeed, already $P_0$ will then be finer
  than $\mathcal{I}$.
\end{remark}
\noindent We look in more detail at two corner cases of the algorithm
where the \op{select} routine retains all available information,
respectively none. 

\begin{remark} \label{finalchain}
  If $\op{select}(
    \begin{mytikzcd}
      |[inner sep=0mm]|
      X\  \arrow[->>]{r}{y}
      & Y \arrow[->>]{r}{z}
      & Z
    \end{mytikzcd}
    ) = \id_Y$ (cf.~\autoref{exampleSelects}\ref{exampleSelectsId}),
    then \autoref{catPT} becomes König and Küppers' final chain
    algorithm~\cite{KonigKupper14}, as we will now explain.
  \begin{enumerate}
  \item Recall that $H$ induces the \emph{final chain}:
    \[
      1 \xleftarrow{!} H1 \xleftarrow{H!} H^21 \xleftarrow{H^2!} \cdots
      \xleftarrow{H^{i-1}!} H^i 1 \xleftarrow{H^i !} H^{i+1} 1
      \xleftarrow{H^{i+1} !} \cdots
    \]
    (The chain is transfinite but we consider only the first~$\omega$
    stages.)  Every coalgebra $\xi\colon X\to HX$ then induces a
    \emph{canonical cone} $\xi^{(i)}\colon X\to H^i 1$ on the final
    chain, defined inductively by
    \[
      \xi^{(0)}=\mathbin{!}\colon X \to H^0 1 = 1
      \quad\text{and}\quad
      \xi^{(i+1)} = (X \xrightarrow{\xi} HX \xrightarrow{H\xi^{(i)}}
      HH^i 1 = H^{i+1} 1).
    \]
    The objects $H^n 1$ may be thought of as domains of $n$-step
    behaviour for $H$-coalgebras. If $\C=\Set$ and $X$ is finite, then
    states~$x$ and~$y$ are behaviourally equivalent iff
    $\xi^{(i)}(x) = \xi^{(i)}(y)$ for all $i < \omega$
    \cite{Worrell05}.  In fact, Worrell showed this for
    unrestricted~$X$ and for \emph{finitary} set functors~$H$,
    i.e.~set functors preserving filtered colimits; equivalently,~$H$
    is finitary if for every $x \in HX$ there exists a finite subset
    $m\colon Y \subto X$ and $y \in HY$ such that $x = Hm(y)$. Note
    that for a \emph{finite} coalgebra for an arbitrary set functor $H$,
    behavioural equivalence remains the same when we pass to the
    \emph{finitary part} of $H$, i.e.~the functor given by
    \[
      H_f X = \bigcup\{Hm[Y] \mid \text{$m\colon Y \subto X$ and $Y$
        finite}\}.
    \]
    To see this note that if two states in a finite
    coalgebra can be identified by a coalgebra morphism into some $H$-coalgebra,
    then they can be identified by a coalgebra morphism into a finite
    $H$-coalgebra. This is just by image factorization of coalgebras
    (see \autoref{sec:coalg}).
    
  \item The inclusions $P_i \monoto Q_{i+1}$ in \autoref{PfinerthanQ}
    reflect that only some and not necessarily all of the information
    present in the relation $P_i$ (resp.~the quotient
    $\unnicefrac{X}{P_i}$) is used for further refinement. If indeed
    everything is used, then $Q_{i+1}= P_i$, and our algorithm simply computes the
    kernels of the morphisms $\xi^{(i)}\colon X \to H^i 1$ forming the
    canonical cone:
  \end{enumerate}
\end{remark}
\begin{proposition}\label{P:chain}
  If $\op{select}(
    \begin{mytikzcd}[inlinecd]
      |[inner sep=0mm]|
      X\  \arrow[->>]{r}{y}
      & Y \arrow[->>]{r}{z}
      & Z
    \end{mytikzcd}\!\!) = \id_Y$, then for all $i\in
    \N$, $Q_i=\ker \xi^{(i)}$.
\end{proposition}
\begin{proof}
  With $\op{select}(y,z) = \id_Y$, we have
  $q_{i+1} = \kappa_{P_i}\colon X\to \unnicefrac{X}{P_i}$, and so
  $\ker q_{i+1} = P_i$ for all $i\in \N$. Thus, $\ker q_{i+1}$ is
  finer than $Q_i$ by \autoref{PfinerthanQ}. It follows that
  $Q_{i+1} = Q_i\cap \ker q_{i+1} = \ker q_{i+1} = P_i$.

  In order to prove that $Q_i=\ker \xi^{(i)}$, for all $i\in \N$, we
  construct monomorphisms $m_i\colon X/Q_i\rightarrowtail H^{i}1$ with
  $m_i\cdot \kappa_{Q_i} = \xi^{(i)}$ inductively (which implies
  $Q_i = \ker \xi^{(i)}$ by~\autoref{rem:kernel}.\ref{i:mono}). For
  $i=0$, we trivially have $m_0\colon X/Q_0 \xrightarrow{\cong} 1$. In the inductive
  step, we put $m_{i+1} := Hm_i\cdot \xi/Q_i$:
    \[
      \begin{mytikzcd}[baseline=(bot.base)]
        X \arrow{r}{\xi}
        \descto{dr}{\eqref{eq:xiQuotient}}
        \arrow[->>]{d}[swap]{\kappa_{Q_{i+1}}}
        \arrow[shiftarr={yshift=7mm}]{rr}{\xi^{(i+1)}}
        & HX
        \arrow{d}[swap]{H\kappa_{Q_i}}
        \descto[pos=0.2]{dr}{IH}
        \arrow{r}{H\xi^{(i)}}
        & H^{i+1} 1
        \\
        |[alias=bot]|
        X/Q_{i+1}
        \arrow[>->]{r}{\xi/Q_i}
        & H(X/Q_i)
        \arrow[>->]{ur}[swap]{Hm_i}
        & {}
      \end{mytikzcd}
      \tag*{\qedhere}
    \]
\end{proof}

\noindent Intuitively, the \op{select} routine in \autoref{P:chain} retains all
available information. The other extreme is the following:
\begin{definition}
  We say that \op{select} is \emph{discarding} at $X \overset{y}{\epito} Y
  \overset{z}{\epito} Z$ if $\op{select}(y,z)\colon Y\to K$ factorizes through $z$.
  Further, we call $\op{select}$ \emph{progressing} if $\op{select}(y,z)$ is
  discarding at $y, z$ only if $z$ is an isomorphism.
\end{definition}
\begin{example}
\begin{enumerate}
\item The $\op{select}$ picking the smaller half in
  \autoref{exampleSelects}\ref{exampleSelectsChi} is progressing. We
  prove the contraposition: if $z$ in
  $X \overset{y}{\epito} Y \overset{z}{\epito} Z$ is not an
  isomorphism, let $S \in \unnicefrac{X}{\ker y} \cong Y$ be the
  subblock used in the definition~\eqref{eq:select} of
  $\op{select}(y,z)$, and note that then there also exists a subblock
  $B \in \unnicefrac{X}{\ker y} \cong Y$ with $z(B) = z(S)$ and
  $|S| \leq |B|$. By the definition of $k=\op{select}(y,z)$, we have
  $k(B) = 1 \neq 2 = k(S)$, and so $k$ cannot factor through $z$.
  
\item The $\op{select}$ routine that always returns $\id_Y$ for
  $X \overset{y}{\epito} Y \overset{z}{\epito} Z$ in
  \autoref{exampleSelects}\ref{exampleSelectsId} is trivially
  progressing: if $\id_Y$ factorizes through $z$, then $z$ is a
  (split) mono, and hence an isomorphism.
\item The $\op{select}$ routine that returns the morphism
  $!\colon Y\to 1$ or $z\colon Y\to Z$ is always discarding, and thus
  fails to be progressing (unless all regular epis in~$\C$ are
  isomorphisms).
\end{enumerate}
\end{example}
\begin{theorem}\label{thm:correct}
  If \op{select} is progressing, then \autoref{catPT} terminates and
  computes the simple quotient of the input coalgebra $(X,\xi)$,
  provided that $(X,\xi)$ has only finitely many quotients.
\end{theorem}
\noindent
E.g.~for $\C=\Set$, every finite coalgebra has only finitely many quotients.
\begin{proof}
  \begin{enumerate}
  \item We first show that our algorithm fails to progress in
  the $(i+1)^{\text{st}}$ iteration, i.e.~$Q_{i+1}=Q_i$, iff \op{select}
  is discarding at $X/P_i,X/Q_i$, i.e.~iff
  $k_i:=\op{select}\big(\!
      X \overset{\kappa_{P_i}}\twoheadrightarrow
      \unnicefrac{X}{P_i} \overset{f_i}\twoheadrightarrow
      \unnicefrac{X}{Q_i}
    \!\big)$ factorizes through $f_i$.

    To see this, first note that \op{select} is discarding at
    $X/P_i,X/Q_i$ iff $q_{i+1}$ factorizes through~$\kappa_{Q_i}$:
  \[
    \begin{mytikzcd}
      X
      \arrow[->>]{r}{\kappa_{P_i}}
      \arrow[->>,shiftarr={yshift=6mm}]{rr}{q_{i+1}}
      \arrow[->>,swap]{rd}{\kappa_{Q_i}}
      &
      \unnicefrac{X}{P_i}
      \arrow{r}{k_{i+1}}
      \arrow[->>]{d}{f_i}
      &
      K_{i+1}
      \\
      &
      \unnicefrac{X}{Q_i}
      \ar[ru,dashed]
    \end{mytikzcd}
  \]
  We thus have the desired equivalence: $q_{i+1}$ factorizes through
  $\kappa_{Q_i}$ iff (by \autoref{rem:kernel}\ref{i:finer}) $Q_i$ is
  finer than $\ker q_{i+1}$ iff
  $Q_{i}=Q_i\cap \ker q_{i+1} = Q_{i+1}$.

\item We proceed to prove the claim. \autoref{PfinerthanQ} shows that
  we obtain a chain of successively finer quotients
  $\unnicefrac{X}{Q_i}$. Since $X$ has only finitely many quotients,
  there must be an $i$ such that $Q_i = Q_{i+1}$, and this implies,
  using point~(1), that $\op{select}$ is discarding at $X/P_i,
  X/Q_i$.
  Since the $\op{select}$ routine is progressing, we obtain
  $P_i = Q_i$ as desired.\qedhere
\end{enumerate}
\end{proof}


\section{Incremental Partition Refinement}
\label{sec:opti}
\tikzsetfigurename{partref-opti-}

\noindent In the most generic version of the partition refinement algorithm
(Algorithm~\ref{catPT}), the partitions are recomputed from scratch in every
step: In Step~\ref{step:P} of the algorithm, $P_{i+1}=\ker(H\fpair{\bar
  q_i,q_{i+1}}\cdot\xi)$ is computed from the information $\bar q_i$ accumulated
so far and the new information $q_{i+1}$, but in general one cannot exploit that
the kernel of $\bar q_i$ has already been computed. We now present a refinement
of the algorithm in which the partitions are computed incrementally, i.e.\
$P_{i+1}$ is computed from $P_i$ and $q_{i+1}$. This requires the type functor
$H$ to be \emph{zippable} (\autoref{D:zip}) and the \op{select} routine to
\emph{respect compound blocks} (\autoref{defRespectCompounds}).


Note that in Step~\ref{step:Q}, \autoref{catPT} computes a kernel
$Q_{i+1} = \ker \bar q_{i+1} = \ker \fpair{\bar q_i, q_{i+1}}$ as the
intersection of $\ker(\bar q_i)$ and $\ker(q_i)$ (cf.~\eqref{eq:kercap}).
Hence, the partition $X/\ker \bar q_{i+1}$ for such a kernel can be computed in
two steps:
\begin{enumerate}
\item Compute $\unnicefrac{X}{\ker \bar q_i}$.
\hfill
\item Refine every block in $\unnicefrac{X}{\ker \bar q_{i}}$ with respect to
$q_{i+1}\colon X\to K_{i+1}$.
\end{enumerate}
\autoref{catPT} can thus be implemented to keep track of the partition
$\unnicefrac{X}{Q_i}$ and then refine this partition by $q_{i+1}$ in
each iteration. 

However, the same trick cannot be applied immediately to the
computation of $\unnicefrac{X}{P_i}$, because of the functor $H$
inside the computation of the kernel:
\( P_{i+1} = \ker (H\fpair{\bar q_i, q_{i+1}}\cdot \xi) \).
In \autoref{propZippable}, we will provide sufficient conditions for $H$,
$a\colon D\to A$, $b\colon D\to B$ to satisfy
\[
    \ker H\fpair{a,b}
    = \ker \fpair{Ha,Hb}.
\]
As soon as this holds for $a=\bar q_i, b=q_{i+1}$, we can optimize the
algorithm by changing Step~\ref{step:P}~to%
\begin{equation}
    P_{i+1}' := \ker \fpair{H\bar q_i\cdot \xi, Hq_{i+1}\cdot \xi}
    \quad(= P_i \cap \ker 
    (Hq_{i+1}\cdot \xi)).
    \label{kernelOptimization}
\end{equation}
The conditions on $a$ and $b$ will be ensured by a condition on \op{select}, and
the condition on the functor $H$ is as follows:
\begin{definition}\label{D:zip}
    A functor $H\colon C\to \D$ is \emph{zippable}
    if the following morphisms are monomorphisms for every objects $A$
    and $B$:
    \[
        \op{unzip}_{H,A,B}\colon
        H(A+B) \xrightarrow{\fpair{H(A+!),H(!+B)}} H(A+1) × H (1+B)
    \]
\end{definition}

\noindent Intuitively, if $H$ is a functor on $\Set$, we may think of
elements $t$ of $H(A+B)$ as shallow terms with variables from
$A+B$. Then zippability means that each $t$ is uniquely determined by
the two terms obtained by replacing $A$- and $B$-variables,
respectively, by some placeholder $\gap$, viz.~the element of $1$, as
illustrated in the examples in \autoref{figZippable}.

\begin{figure}[h]
    \begin{subfigure}[b]{.22\textwidth}
        \(
        \begin{mytikzcd}[row sep = 0mm]
        a_1\,a_2\,b_1\,a_3\,b_2
        \arrow[shiftarr={xshift=18mm},mapsto]{d}[xshift=-4mm,pos=0.0,above]{\op{unzip}}
        \\
        \begin{array}{c}
        (a_1a_2\gap a_3 \gap, \\
        \phantom{(}\,\gap \gap\,b_1\!\gap b_2)
        \end{array}
        \end{mytikzcd}
        \)
        \caption{$(-)^*$ is zippable}
    \end{subfigure}%
\hfill%
\begin{subfigure}[b]{.22\textwidth}
        \(
        \begin{mytikzcd}[row sep = 0mm, ampersand replacement = \&]
        \{a_1,a_2,b_1\}
        \arrow[shiftarr={xshift=18mm},mapsto]{d}[xshift=-4mm,pos=0.0,above]{\op{unzip}}
        \\
        \begin{array}{r@{\,}l}
        (\{a_1,a_2,&\gap\}, \\
        \{\gap,&b_1\})
        \end{array}
        \end{mytikzcd}
        \)
        \caption{$\Potf$ is zippable}
    \end{subfigure}%
\hfill%
\begin{subfigure}[b]{.46\textwidth}
        \(
        \begin{mytikzcd}[row sep = 0mm,
                       column sep = -13mm,
                       ]
        |[inner xsep=0mm]|
        \begin{array}{@{}l@{}}
        \big\{\{a_1,b_1\},
        \{a_2,b_2\}\big\}
        \end{array}
        \arrow[start anchor={[xshift=-4mm]},
               rounded corners,
               pos=0.3,
               to path={ |- (\tikztotarget) \tikztonodes },
               mapsto]{dr}[left,]{\op{unzip}}
        &
        &
        |[inner xsep=0mm]|
        \begin{array}{@{}l@{}}
        \big\{\{a_1,b_2\},
        \{a_2,b_1\}\big\}
        \end{array}
        \arrow[start anchor={[xshift=4mm]},
               rounded corners,
               pos=0.3,
               to path={ |- (\tikztotarget) \tikztonodes },
               mapsto]{dl}[right,overlay]{\op{unzip}}
        \\
        &
        |[inner xsep = 1mm]|
        \begin{array}{@{}l@{}}
        (\big\{\{a_1,\gap\},\{a_2,\gap\}\big\}, \\
        \phantom{(}\big\{\{\gap, b_1\},\{\gap,b_2\}\big\}) \\
        \end{array}
        \end{mytikzcd}
        \)
        \caption{$\Potf\Potf$ is not zippable}
    \end{subfigure}
    \caption{Zippability of \Set-Functors for sets
    $A=\{a_1,a_2,a_3\}$, $B=\{b_1,b_2\}$.}
    \label{figZippable}
\end{figure}
\begin{lemma} \label{zippableExtended}
  Let $H$ be zippable and $f\colon A\to C$, $g\colon B\to D$. Then the following
  is a mono:
  \[
    H(A+B) \xrightarrow{\fpair{H(A+g),H(f+B)}} H(A+D) × H (C+B)
  \]
\end{lemma}
\begin{proof}
  By finality of $1$, the diagram
  \[
    \begin{mytikzcd}[column sep = 3cm]
      H(A+B)
      \arrow{d}[left]{\fpair{H(A+g),H(f+B)}}
      \arrow[>->]{dr}[sloped,above]{
        \op{unzip}_{H,A,B}
        = \fpair{H(A+!),H(!+B)}
      }
      \\
      H(A+D) × H (C+B)
      \arrow{r}{H(A+!) × H (!+B)}
      &
      H(A+1) × H (1+B)
    \end{mytikzcd}
  \]
  commutes. Since the diagonal arrow is monic, so is
  $\fpair{H(A+g),H(f+B)}$.
\end{proof}
\begin{assumption} \label{ass:zippableResults} For the remainder of
  \autoref{sec:opti}, we assume that $\C = \Set$. 
\end{assumption}
\noindent
However, most proofs are category-theoretic to clarify where working
in $\Set$ is really needed and where the arguments are more general.

\begin{example}
\begin{enumerate}
\item Constant functors $X\mapsto A$ are zippable: $\op{unzip}$
  is the diagonal $A \to A \times A$.

\item The identity functor is zippable since
  \(
  \fpair{A+!, !+B}\colon A+B \to (A+1) × (1+B)
  \)
  is monic in $\Set$.

  \item From Lemma~\ref{lem:closure} it follows that every
    polynomial endofunctor is zippable. Indeed, a polynomial functor
    is precisely one that is build from constant and the identity
    functors using (finite) products and coproducts (cf.~\autoref{ex:coalg}\ref{i:poly}).
\end{enumerate}
\end{example}
\begin{lemma}\label{lem:closure}
  Zippable endofunctors are closed under (possibly infinite) products, coproducts and
  subfunctors. 
\end{lemma}
\noindent (Recall that products and coproducts of functors are formed
pointwise, e.g.~$(F+G)(X)=FX+GX$.)
\begin{proof}
  For the closure under products and coproducts, we only provide the proof for the
  binary case; the proof for arbitrary products and coproducts is
  completely analogous. Let $F,
  G$ be endofunctors.

  \begin{enumerate}
  \item Suppose that both $F$ and $G$ are zippable. That $F \times
G$ is zippable follows from monos being closed under products:
\[
\begin{mytikzcd}[column sep=15mm,row sep=5mm]
    F(A+B) × G(A+B)
    \ar[r,>->,"\op{unzip}_{F,A,B} × \op{unzip}_{G,A,B}"{yshift=2mm}]
    \ar[to path={
            |- (\tikztotarget) \tikztonodes
        },
        rounded corners,
        ]{dr}[below,near end]{\op{unzip}_{F×G,A,B}}
    &
    F(A+1) × F(1+B)
    ×G(A+1) × G(1+B)
    \ar[phantom,d,"\cong" {sloped}]
    \\
    &
    \big(F(A+1)
    ×G(A+1)\big)
    × \big(F(1+B)
    × G(1+B)\big).
\end{mytikzcd}
\]

\item Suppose again that $F$ and $G$ are zippable. To see that $F+G$ is zippable consider the diagram  
\[
\begin{mytikzcd}[column sep=15mm,row sep=5mm]
    F(A+B) + G(A+B)
    \ar[r,>->,"\op{unzip}_{F,A,B} + \op{unzip}_{G,A,B}"{yshift=2mm}]
    \ar[to path={
            |- (\tikztotarget) \tikztonodes
        },
        rounded corners,
        ]{dr}[below,near end]{\op{unzip}_{F×G,A,B}}
    &
    \big(F(A+1) × F(1+B)\big)
    + \big(G(A+1) × G(1+B)\big)
    \ar[>->,d,"\fpair{(\pi_1+\pi_1), (\pi_2+\pi_2)}" {right}]
    \\
    &
    \big(F(A+1)
    +G(A+1)\big)
    × \big(F(1+B)
    + G(1+B)\big).
\end{mytikzcd}
\]
The horizontal morphism is monic since monos are closed under
coproducts in $\Set$. 
The vertical morphism is monic since for any sets
$A_i$ and $B_i$, $i = 1, 2$, the following morphism clearly is a
monomorphism:
\[
  (A_1×B_1) + (A_2×B_2)
  \xrightarrow{\fpair{(\pi_1+\pi_1), (\pi_2+\pi_2)}}
  (A_1+A_2) × (B_1+B_2).
\]

\item Suppose now that $F$ is a subfunctor of $G$ via $s\colon F
\rightarrowtail G$, where $G$ is zippable. Then the following diagram
shows that $F$ is zippable, too:
\[
    \begin{mytikzcd}[column sep = 2cm,baseline=(bot.base)]
        F(A+B)
            \arrow{r}{\op{unzip}_{F,A,B}}
            \arrow[>->]{d}[left]{s_{A×B}}
            \arrow[>->,
                  to path={
                    -- ([xshift=-4mm]\tikztostart.west)
                    |- ([yshift=-4mm]\tikztotarget.south)
                    -- (\tikztotarget)
                  },
                  rounded corners,
                  ]{dr}{}
        & F(A+1) × F(1+B)
            \arrow[>->]{d}{s_{A+1} × s_{1+B}}
        \\
        |[alias=bot]|
        G(A+B)
            \arrow[>->]{r}{\op{unzip}_{G,A,B}}
        & G(A+1) × G(1+B).
    \end{mytikzcd}
  \]
  Indeed, since the composition of the lower and left-hand morphisms
  is monomorphic, so is the upper morphism.\qedhere
\end{enumerate}
\end{proof}

%
%
\begin{lemma}\label{lem:additive}
  If $H$ has a componentwise monic natural transformation
  $H(X+Y) \rightarrowtail HX × HY$, then $H$ is zippable.
\end{lemma}

\begin{proof}
  Let $\alpha_{X,Y}\colon H(X+Y) \rightarrowtail HX × HY$ be monic and
  natural in $X$ and $Y$. Then the square
  \[
  \begin{mytikzcd}[column sep = 4cm]
      H(A+B) \ar[d,>->,"\alpha_{A,B}"'] \ar[r,"\op{unzip} = \fpair{H(A+!),H(!+B)}"]
      & H(A+1)×H(1+B)
        \ar[d,"\alpha_{A,1}×\alpha_{1,B}"]
      \\
      HA×HB
        \ar[r,>->,"\fpair{HA×H!,H!×HB}"]
      & (HA×H1)×(H1×HB)
  \end{mytikzcd}
  \]
  commutes by naturality of $\alpha$. The
  bottom morphism is monic because it has a left inverse,
  $\pi_1×\pi_2$. Therefore, $\op{unzip}$ is monic as well.
\end{proof}

\begin{example} \label{exMonoidZippable}
\begin{enumerate}
\item For every commutative monoid, the monoid-valued functor
  $M^{(-)}$ (see \autoref{ex:coalg}\ref{ex:coalg:2}) admits a natural
  isomorphism $M^{(X+Y)} \cong M^{(X)} \times M^{(Y)}$, and hence is
  zippable by Lemma~\ref{lem:additive}. \label{item:monoid-zippable}
\item As special cases of monoid-valued functors we obtain
  that the finite powerset functor $\Potf$ and the bag functor $\Bagf$
  are zippable.
\item By \autoref{lem:additive}, the full powerset functor $\Pot$ is zippable.
\item The distribution functor $\Dist$ (see Example~\ref{ex:coalg}) is
  a subfunctor of the monoid-valued functor $\R_{\ge 0}^{(-)}$ for the
  additive monoid $\R_{\ge 0}$ of real numbers, and hence is zippable by
  Item~\ref{item:monoid-zippable} and Lemma~\ref{lem:closure}.
\item The previous examples together with the closure properties in
  Lemma~\ref{lem:closure} show that a number of functors of interest
  are zippable, e.g.~$2×(-)^\Inputs$,
  $2×\Potf(-)^\Inputs$, $\Potf(\Inputs×(-))$, $2×\big((-)+1\big)^\Inputs$,
  and variants where $\Potf$ is replaced by $\Bagf$, $M^{(-)}$, or
  $\Dist$.
\end{enumerate}
\end{example}

\begin{remark}
  Out of the above results, only zippability of the identity and
  coproducts of zippable functors make use of properties of $\Set$
  (\autoref{ass:zippableResults}). Indeed, zippable functors on a
  category $\C$ are closed under coproducts as soon as monomorphisms
  are closed under coproducts in $\C$, which is satisfied in most
  categories of interest. Zippability of the identity holds whenever
  $\C$ is extensive, i.e.~it has well-behaved set-like coproducts~(see
  e.g.~\cite{clw93} or \autoref{defExtensive} later). Examples of
  extensive categories are the categories of sets, posets and graphs
  as well as any presheaf category. We will take a closer look at
  extensive categories when we discuss multisorted coalgebras
  (\autoref{sec:multisorted}).
\end{remark}
\begin{example}\label{ex:nbhd}
    The monotone neighbourhood functor, which maps a set $X$ to the set
    \[
        \mathcal{M}(X) =
        \{
            N \subseteq \Pot X
            \mid
            A \in N \wedge B\supseteq A
            \implies B\in N
        \},
    \]
    of monotone neighbourhood systems over~$X$, is not zippable, that
    is, there are distinct monotone neighbourhood systems that are
    identified by $\op{unzip}$. Indeed, denoting the upwards closure
    of a set system by $(-)\upa$, we have
    \begin{align*}
        \op{unzip}\left(\left\{
            \{a_1,b_1\},
            \{a_2,b_2\}
        \right\}\!\upa\right)
        &=
        \big(
            \big\{\{a_1,\gap\},
            \{a_2,\gap\}\big\}\upa
        ,
            \big\{\{\gap,b_1\},
            \{\gap,b_2\}\big\}\upa
        \big)
        \\&=
        \op{unzip}\left(\left\{
            \{a_1,b_2\},
            \{a_2,b_1\}
        \right\}\!\upa\right).
    \end{align*}
\end{example}
\begin{example}\label{non-zippable}
  The functor $\Potf\Potf$ fails to be zippable, as shown in
  \autoref{figZippable}. First, this shows that zippable functors are
  not closed under quotients, since $\Potf\Potf$ is a quotient of the
  polynomial, hence zippable, functor $HH$ where
  $HX = \coprod_{n <\omega} X^n$.  Secondly, this shows that zippable
  functors are not closed under composition.
\end{example}
\noindent The following example shows that the optimized algorithm,
i.e.~\autoref{catPT} run with \eqref{kernelOptimization} in lieu of
Step~\ref{step:P}, is not correct for the non-zippable functor
$\Potf\Potf$, even though the $\op{select}$ routine used here (see
\autoref{exampleSelects}\ref{exampleSelectsChi}) behaves sufficiently
well (specified later in \autoref{defRespectCompounds} and
cf.~\autoref{optimizationSummary}).
\begin{example}
  \label{exWrongRunForPP}
Consider the following coalgebra $\xi\colon X\to HX$ for $HX = 2×\Potf\Potf X$:
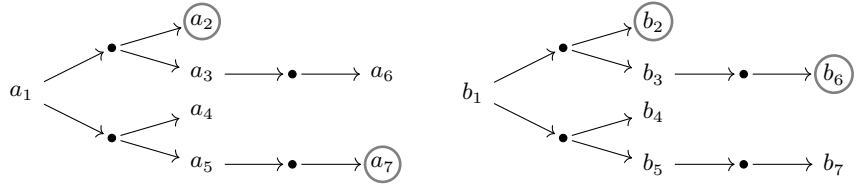
\begin{center}
\begin{tikzpicture}[
        n/.style = {
               execute at begin node=\(,%
               execute at end node=\),%
               inner sep = 1mm,
        },
    ]
    \begin{scope}[
        grow'                   = right,
        sibling distance        = 7mm,
        level distance          = 12mm,
        edge from parent/.style = {
            draw,
            ->,
            shorten <= 1pt,
            shorten >= 1pt,
        },
        level 1/.append style = {
            sibling distance = 12mm,
        },
        level 2/.append style = {
            sibling distance = 7mm,
        },
        every node/.style       = {n,font=\footnotesize},
        sloped,
        dummy/.style    = {circle,
                           draw=black,
                           fill=black,
                           inner sep=1pt,
                           outer sep=1pt,
                           minimum width=0,
                           minimum height=0,
        },
        final/.style    = {
            circle,
            draw=black!50,
            line width=1pt,
            inner sep =1pt,
        },
    ]
        \node (a1) {a_1}
        child { node[dummy] {}
            child { node[final] (a2) {a_2} }
            child { node (a3) {a_3}
                child { node[dummy] {}
                    child { node (a6) {a_6} }
                }
            }
        }
        child { node[dummy] {}
            child { node (a4) {a_4} }
            child { node (a5) {a_5}
                child { node[dummy] {}
                    child { node[final] (a7) {a_7} }
                }
            }
        }
        ;
        \node (b1) at ([xshift=6cm]a1){b_1}
        child { node[dummy] {}
            child { node[final] (b2) {b_2} }
            child { node (b3) {b_3}
                child { node[dummy] {}
                    child { node[final] (b6) {b_6} }
                }
            }
        }
        child { node[dummy] {}
            child { node (b4) {b_4} }
            child { node (b5) {b_5}
                child { node[dummy] {}
                    child { node (b7) {b_7} }
                }
            }
        }
        ;
    \end{scope}
\end{tikzpicture}
\end{center}
Final states, i.e.~states $x$ with $\pi_1(\xi(x)) = 1$, are indicated by a circle. Let us replace
step \ref{step:P} of \autoref{catPT} by equation \eqref{kernelOptimization},
i.e.~we compute
\begin{equation*}
  P_{i+1} = P_{i+1}' \overset{\eqref{kernelOptimization}}{=}
  \ker \fpair{H\bar q_i\cdot \xi,Hq_{i+1}\cdot\xi}
  = \ker P_i' \cap \ker (Hq_{i+1}\cdot\xi).
\end{equation*}
We will show that the states $a_1$ and $b_1$ are identified by all $P_i'$ and $Q_i$, i.e.~they are not distinguished by the algorithm, although they are
clearly behaviourally inequivalent.

We simplify the partitions by defining abbreviations for the final and
non-final states without successors as well as the rest,
\[
  F:= \{a_2,a_7, b_2,b_6\},
  \quad
  N:=\{a_4,a_6, b_4,b_7\}
  \quad
  \text{and}
  \quad
  R:=\{a_1,a_3,a_5, b_1,b_3,b_5\}.
\]
Then we run the optimized algorithm with the $\op{select}$ routine in \autoref{exampleSelects}\ref{exampleSelectsChi}, computing $Q_i$ and $P_{i}'$ (see
\eqref{kernelOptimization}), and we obtain the following sequence of
partitions.  \renewcommand{\arraystretch}{1.3}%
\newcommand{\myblocklist}[1]{ \big\{ #1 \big\} }%
\newcommand{\mytopbotrule}{ \noalign{ \global\dimen1\arrayrulewidth
    \global\arrayrulewidth1.3pt }\hline \noalign{
    \global\arrayrulewidth\dimen1 } }%
\begin{center}%
\begin{tabular}{@{\hspace{0mm}}LL@{\hspace{3mm}}L@{\hspace{3mm}}L@{\hspace{0mm}}}
  \mytopbotrule
  i & q_i & \unnicefrac{X}{Q_i} & \unnicefrac{X}{P_i'}
\\ \hline
\\[-4.5mm]
    0
    & !\colon X\to 1
    & \myblocklist{X}
    &
        \myblocklist{F, N,
        R }
    \\[1mm]
    1
    & \kappa_{P_0'}\colon X\twoheadrightarrow \unnicefrac{X}{P_0'}
    & \myblocklist{F,N, R}
    & \myblocklist{F,N, \{a_1,b_1\}, \{a_3,b_5\}, \{a_5,b_3\}}
\\[1mm]
    2
    & \chi_{\{a_1,b_1\}}^R\colon X\to 3
    & \myblocklist{F,N, \{a_1,b_1\}, \{a_3,b_5, a_5,b_3\}}
    & \myblocklist{F,N, \{a_1,b_1\}, \{a_3,b_5\}, \{a_5,b_3\}}
\\[1mm]
    3
    & \chi_{\{a_3,b_5\}}^{\{a_3,b_5,a_5,b_3\}}\colon X\to 3
    & \myblocklist{F,N, \{a_1,b_1\}, \{a_3,b_5\}, \{a_5,b_3\}}
    & \myblocklist{F,N, \{a_1,b_1\}, \{a_3,b_5\}, \{a_5,b_3\}}
\\[1mm] \mytopbotrule
\end{tabular}
\end{center}
For the subblock $S=\{a_3,b_5\}$ selected in the third 
iteration we see that $\{a_1, b_1\}$ is not split in $\unnicefrac{X}{P_3'}$
because:
\begin{align*}
  H\chi_{\{a_3,b_5\}}^{\{a_3,b_5,a_5,b_3\}} \cdot \xi(a_1)
  &= H\chi_{\{a_3,b_5\}}^{\{a_3,b_5,a_5,b_3\}} \myblocklist{\{a_2,a_3\}, \{a_4,a_5\}} \\
    &= \phantom{H\chi_{\{a_3,b_5\}}^{\{a_3,b_5,a_5,b_3\}}}\myblocklist{\{0,2\}, \{0,1\}}
    \\ &
    = \phantom{H\chi_{\{a_3,b_5\}}^{\{a_3,b_5,a_5,b_3\}}}\myblocklist{\{0,1\}, \{0,2\}} \\
    &= H\chi_{\{a_3,b_5\}}^{\{a_3,b_5,a_5,b_3\}} \myblocklist{\{b_2,b_3\}, \{b_4,b_5\}}
    = H\chi_{\{a_3,b_5\}}^{\{a_3,b_5,a_5,b_3\}} \cdot \xi(b_1)
\end{align*}
At this point the algorithm terminates because $\unnicefrac{X}{Q_2} =
\unnicefrac{X}{P_2}$, while incorrectly not distinguishing $a_1$ and
$b_1$. 

Note that this result remains the same if we chose the subblock
$\{a_5,b_3\}$ in the third iteration or if we chose
$\{a_3,b_5\}$ in the second iteration and $\{a_1,b_1\}$ in the third one.
\end{example}

\noindent Observe that, in general,
$\ker H\fpair{a,b}$ differs from $\ker\fpair{Ha,Hb}$ even if~$H$ is
zippable:
\begin{example}
  For $H=\Potf$ and product projections
  $\pi_1\colon A \times B \to A$ and $\pi_2\colon A \times B \to B$,
  $\fpair{\Potf\pi_1,\Potf\pi_2}$ in general fails to be injective
  although $\Potf\fpair{\pi_1,\pi_2}=\Potf\id_{A \times B}=\id_{\Potf(A
    \times B)}$. Thus
  \[
    \ker \Potf\fpair{\pi_1,\pi_2} \cong \Potf(A \times B) \not\cong
    \ker \fpair{\Potf\pi_1,\Potf\pi_2}.
  \]
\end{example}

Hence, in addition to zippability of $H$, we will need to enforce constraints on the
\op{select} routine to achieve the desired
optimization~\eqref{kernelOptimization}.

The next example illustrates this issue, and a related one: One might
be tempted to implement splitting by a subblock $S$ by using the usual
characteristic function $q_i = \chi_S\colon X \to K_i$.
While this approach is sufficient for systems with real-valued
weights~\cite{ValmariF10}, it may in general let
$\ker (H\fpair{\bar q_i, q_{i+1}}\cdot \xi)$ and
$\ker \fpair{H\bar q_i\cdot \xi, Hq_{i+1}\cdot \xi}$ differ even if
$H$ is zippable, thus rendering the algorithm incorrect:

    \tikzstyle{coalgebraNodes}=[
            every node/.style={
                draw=none,
                inner sep = 1pt,
                label distance=0mm,
            },
            every label/.append style={
                font=\small,
                execute at begin node=\(,
                execute at end node=\),
                inner sep=2pt,
                fill=white,
            },
            x = 14mm,
            y = 14mm,
    ]
    \tikzstyle{partitionPi}=[
        on background layer,
        every node/.style={
          partitionBlock,
          solid,
        },
    ]
    \tikzstyle{partitionQi}=[
        partitionPi,
        every node/.append style={
            inner sep = 3mm,
            dashed,
        },
    ]
    \tikzstyle{doublearrow}=[
      draw=white,
      line width=2pt,
    ]
\begin{figure}
    \newcommand{\exampleCoalgebra}{
    \begin{scope}[coalgebraNodes]
        \node[label={north:t_1}] (triangle1) at (2,1) {$\blacktriangle$};
        \node[label={north:t_2}] (triangle2) at (3,1) {$\blacktriangle$};
        \node[label={north:s_1}] (square1) at (1,1) {$\mdblksquare$};
        \node[label={south:c_1}] (circle1) at (1,0) {$\smblkcircle$};
        \node[label={south:c_2}] (circle2) at (2,0) {$\smblkcircle$};
        \node[label={south:c_3}] (circle3) at (3,0) {$\smblkcircle$};
    \end{scope}
    \foreach \style in {doublearrow,->} {
      \path[\style]
          (triangle1) edge[bend right=10] (circle1)
          (triangle1) edge (circle2)
          (triangle1) edge (circle3)
          (triangle2) edge[bend left=5] (circle1)
          (triangle2) edge (circle3)
          (circle1) edge (square1)
          (circle2) edge (circle3)
          ;
    }
    }
    \begin{subfigure}{.48\textwidth} \centering
    \begin{tikzpicture}
        \exampleCoalgebra
        \begin{scope}[partitionQi ]
            \node[fit = (square1) (circle1) (circle2) (circle3)
                        (triangle2) (triangle1)] {};
        \end{scope}
        \begin{scope}[partitionPi ]
            \node[fit = (square1)] {};
            \node[fit = (circle1) (circle2)] {};
            \node[fit = (circle3)] {};
            \node[fit = (triangle2) (triangle1)] {};
        \end{scope}
    \end{tikzpicture}
    \caption{$Q_0, P_0$ for $\bar q_0 = \mathbin{!}$}
    \label{figPartitionQ0}
    \end{subfigure}%
\hfill\begin{subfigure}{.49\textwidth} \centering
    \begin{tikzpicture}
        \exampleCoalgebra
        \begin{scope}[ partitionQi ]
            \node[fit = (square1)] {};
            \node[fit = (circle1) (circle2)] {};
            \node[fit = (circle3)] {};
            \node[fit = (triangle2) (triangle1)] {};
        \end{scope}
        \begin{scope}[ partitionPi ]
            \node[fit = (square1)] {};
            \node[fit = (circle2)] {};
            \node[fit = (circle1)] {};
            \node[fit = (circle3)] {};
            \node[fit = (triangle2) (triangle1)] {};
        \end{scope}
    \end{tikzpicture}
    \caption{$Q_1, P_1$ for $\bar q_1 =
    \kappa_{P_0}$}\noshowkeys\label{figPartitionQ1}
    \end{subfigure}
    \caption{Partitions of a coalgebra $\xi$ for $H=\{\blacktriangle,
    \mdblksquare, \smblkcircle\} × \Potf(-)$. $\unnicefrac{X}{Q_i}$ is
    indicated by dashed, $\unnicefrac{X}{P_i}$ by solid lines.}
    \label{figSplitStep}
  \end{figure}
  \begin{figure}
    \begin{tikzpicture}[
            elementX/.style={
                execute at begin node=\(,
                execute at end node=\),
            },
            x=7mm,
            y=8mm,
            baseline=-5mm, 
        ]
        \foreach \prefix/\yshift in {top/2,middle/1,bottom/0} {
        \node[elementX] (\prefix circle1) at (0,\yshift) {c_1};
        \node[elementX] (\prefix circle2) at (1,\yshift) {c_2};
        \node[elementX] (\prefix circle3) at (2,\yshift) {c_3};
        \node[elementX] (\prefix square1) at (3,\yshift) {s_1};
        \node[elementX] (\prefix triangle1) at (4,\yshift) {t_1};
        \node[elementX] (\prefix triangle2) at (5,\yshift) {t_2};
        }
        \begin{scope}[
                every node/.append style={partitionBlock,inner sep=0mm},
                ]
            \node[fit = (topsquare1)] {};
            \node[fit = (topcircle1) (topcircle2)] {};
            \node[fit = (topcircle3)] {};
            \node[fit = (toptriangle2) (toptriangle1)] (XQ1) {};
        \end{scope}
        \begin{scope}[
                every node/.append style={partitionBlock,inner sep=0mm},
                ]
            \node[fit =  (middlecircle1)] {};
            \node[fit = (middlesquare1)
                        (middlecircle2) (middlecircle3) (middletriangle2)
                        (middletriangle1)] (chiS) {};
        \end{scope}
        \begin{scope}[
                every node/.append style={partitionBlock,inner sep=0mm},
                ]
            \node[fit =  (bottomcircle1)] {};
            \node[fit = (bottomcircle2) ] {};
            \node[fit = (bottomcircle3) (bottomsquare1) (bottomtriangle2)
                        (bottomtriangle1)]  (chiSC){};
        \end{scope}
        \begin{scope}[every node/.style={
                inner xsep=0cm,
                anchor=west,
                xshift=5mm
            }]
        \node at (chiS.east) {$\unnicefrac{X}{\ker \chi_S}$};
        \node at (chiSC.east) {$\unnicefrac{X}{\ker \chi_S^C}$};
        \node at (XQ1.east) {$\unnicefrac{X}{Q_1} = \unnicefrac{X}{P_0}$};
        \end{scope}
    \end{tikzpicture}
    \caption{Grouping of elements when $S:=\{c_1\}$ is chosen as the
      next subblock and $C := \{c_1,c_2\}$ as the compound block.  }
    \noshowkeys\label{groupingExamples}
\end{figure}
\begin{example}\label{expl:no-respect}
  Consider the coalgebra $\xi\colon X\to HX$ for the zippable functor
  $H=\{\blacktriangle, \mdblksquare, \smblkcircle\} × \Potf(-)$
  illustrated in \autoref{figSplitStep} (essentially a Kripke
  model). The initial partition $\unnicefrac{X}{P_0}$ splits the set
  of all states by shape and by $\Potf!$, i.e.~states with successors
  are distinguished from the ones without
  successors (\autoref{figPartitionQ0}). Now, suppose that
  $\op{select}$ returns $k_1 := \id_{\unnicefrac{X}{P_0}}$, i.e.\
  retains all information (cf.~\autoref{finalchain}), so that
  $Q_1 = P_0$ and $P_1$ puts $c_1$ and $c_2$ into different blocks
  (\autoref{figPartitionQ1}). Since $q_0 = \mathord{!}$, we have $\ker \bar q_1 = \ker
  q_1$ and thus simplify notation by directly defining
  $\bar q_1 := \kappa_{P_0}$. We now analyse the next partition that
  arises when we split w.r.t.~the subblock $S=\{c_1\}$ but not
  w.r.t.~the rest $C\setminus S$ of the compound block
  $C=\{c_1,c_2\}$; in other words, we take
  $k_2 := \chi_{\{\{c_1\}\}}\colon \unnicefrac{X}{P_1}\to 2$, making
  $q_2 = \chi_{\{c_1\}}\colon X\to 2$.  Then,
  $H\fpair{\bar q_1,q_2}\cdot \xi$ splits $t_1$ from $t_2$, because
  $t_1$ has a successor $c_2$ with $\bar q_1(c_2) = \{c_1,c_2\}$ and
  $q_2(c_2) = 0$ whereas $t_2$ has no such successor. However,
  $t_1,t_2$ fail to be split by $\fpair{H\bar q_1, Hq_2}\cdot \xi$
  because their successors do not differ when we look at successor
  blocks in $\unnicefrac{X}{Q_1}$ and $\unnicefrac{X}{\ker \chi_S}$
  separately: both have $\{c_1,c_2\}$ and $\{c_3\}$ as successor
  blocks in $\unnicefrac{X}{Q_1}$ and $\{c_1\}$ and
  $X\setminus\{c_1\}$ as successors in
  $\unnicefrac{X}{\ker\chi_S}$ (cf.~\autoref{groupingExamples}). Formally:
    \begin{align*}
        H\bar q_1\cdot \xi (t_1) &= 
        (\id×\Potf\kappa_{P_0})\cdot \xi (t_1) = 
        (\blacktriangle,\big\{
            \{c_1,c_2\},
            \{c_3\}
        \big\})
        = H\bar q_1\cdot \xi (t_2),
    \\
        H q_2\cdot \xi (t_1) &= 
        (\id× \Potf\chi_{\{c_1\}})\cdot \xi (t_1) = 
        (\blacktriangle,\big\{
            0, 1
        \big\})
        = H q_2\cdot \xi (t_2).
    \end{align*}
    So if we computed $P_2$ iteratively as in
    \eqref{kernelOptimization} for $q_2=\chi_S$, then $t_1$ and $t_2$
    would not be split, and we would reach the termination condition
    $P_2 = Q_2$ before all behaviourally inequivalent states have been
    separated.

    Already Paige and Tarjan \cite[Step~6 of the
    Algorithm]{PaigeTarjan87} note that one additionally needs to
    split by $C\setminus S=\{c_3\}$, which is accomplished by
    splitting by $q_i=\chi_S^C$ (see
    \autoref{exampleSelects}\ref{exampleSelectsChi}). This is
    formally captured by the condition we introduce next and explain
    further in \autoref{compoundBlockEquivalences}.
\end{example}

\begin{definition} \label{defRespectCompounds}
  A \op{select} routine \emph{respects compound blocks} if whenever
  $k= \op{select}(X\overset{y}{\twoheadrightarrow} Y
  \overset{z}{\twoheadrightarrow} Z)$
  then the union $\ker z\cup \ker k$ is a kernel.
\end{definition}

\noindent Since in $\Set$, reflexive and symmetric relations are
closed under unions, the definition boils down to $\ker z\cup\ker k$
being transitive. In $\Set$, we have an intuitive characterization:

\begin{lemma} \label{compoundBlockEquivalences}
    For maps $a\colon Y\to A$, $b\colon Y\to B$, the following are
    equivalent:
    \begin{enumerate}[ref={(\arabic*)}]
    \item $\ker a\cup \ker b\rightrightarrows Y$ is a kernel (i.e.~an
      equivalence relation).
    \item $\ker a\cup \ker b\rightrightarrows Y$ is the kernel of
    the pushout of $a$ and $b$.
    \item For all $x,y,z \in Y$, $a(x) =a(y)$ and $b(y)=b(z)$ implies
    $a(x) = a(y) = a(z)$ or $b(x)=b(y)=b(z)$.
    \item For all $x\in Y$, $[x]_a \subseteq [x]_b$ or  $[x]_b
    \subseteq [x]_a$.
    \label{classInclusion}
    \end{enumerate}
\end{lemma}
\noindent The last item states that when we move from $a$-equivalence classes to
$b$-equivalence classes, the classes either merge or split, but do not merge
with other classes and split at the same time. In \autoref{subfig:trans}, $\ker
a\cup\ker b$ is transitive, and thus also a kernel. For $x\in \{y_1,y_2,y_3\}$,
we have $[x]_b\subseteq [x]_a$ and for $x\in \{y_4,y_5\}$ we have
$[x]_a\subseteq [x]_b$. On the other hand in \autoref{subfig:nontrans}, $\ker
a\cup \ker b$ is not transitive because $(y_1,y_3)\in \ker b$ and $(y_3,y_5) \in
\ker a$, but $(y_1,y_5)\not \in \ker a\cup\ker b$, and indeed condition
\ref{classInclusion} of \autoref{compoundBlockEquivalences} fails because
$[y_3]_a\not\subseteq [y_3]_b$ and $[y_3]_b\not\subseteq [y_3]_a$ (because $[y_3]_a =
\{y_3,y_4,y_5\}$, $[y_3]_b=\{y_1,y_2,y_3\}$).

\begin{figure}
  \newcommand{\exampleDABpartitions}[2][5]{
    \begin{tikzpicture}[
            x=9mm,
            y=8mm,
            baseline=-5mm, 
        ]
        \foreach \prefix/\displayprefix/\yshift in {A/A/1,B/B/0,I/{Y\!/\!\ker\fpair{a,b}}/-1} {
          \foreach \x in {1,...,#1} {
            \node (\prefix \x) at (\x,\yshift) {\(y_\x\)};
          }
          \node[anchor=east,outer sep=0mm] at ([yshift=0pt]\prefix 1.west) {\(\displayprefix = \big\{\)};
          \node[anchor=west,outer sep=0mm] at ([yshift=0pt]\prefix #1.east) {\(\}\)};
        }
        \foreach \block [count=\mycount]  in {#2} {
        \begin{scope}[]
          \node[partitionBlock, inner sep=0pt, fit=\block] (thisnode) {};
          \ifthenelse{\equal{1}{\mycount}}{}{
            \draw[draw=none]
            let \p1=(lastnode.east) in
            let \p2=(thisnode.west) in
            (lastnode.east) to
            node[yshift=-1.5pt] {\ifthenelse{\lengthtest{\y1=\y2}}{,}{}} (thisnode.west);
          }
          \coordinate (lastnode) at (thisnode.east);
        \end{scope}
        }
    \end{tikzpicture}
  }
  \hfill
  \begin{subfigure}{.49\textwidth} \centering
    \exampleDABpartitions[5]{
      (A1), (A2) (A3), (A4), (A5),
      (B1), (B2), (B3), (B4) (B5),
      (I1), (I2), (I3), (I4), (I5)
    }
    \caption{$\ker a\cup \ker b$ is transitive.}
    \label{subfig:trans}
  \end{subfigure}
  \hfill
  \begin{subfigure}{.49\textwidth} \centering
    \exampleDABpartitions[5]{
      (A1) (A2), (A3) (A4) (A5),
      (B1) (B2) (B3), (B4) (B5),
      (I1) (I2), (I3), (I4) (I5)
    }
    \caption{$\ker a\cup \ker b$ is not transitive.}
    \label{subfig:nontrans}
  \end{subfigure}
  \hfill{}
  \caption{Maps $a\colon Y\to A$ and $b\colon Y\to B$, where the elements of~$A$
    and~$B$ are considered as equivalence classes of elements of $Y$, defining
    $a$ and $b$ implicitly. $Y/\ker\fpair{a,b}$ is the block-wise intersection
    of the partitions defined by $a$ and~$b$.}
  \label{fig:kerTransitive}
\end{figure}

In \autoref{expl:no-respect} of a concrete run of our algorithm, one sees in
\autoref{groupingExamples} that $Q_1\cup \ker \chi_S$ fails to be transitive,
while $Q_1\cup \ker \chi_S^C$ is transitive.

\begin{proof}[Proof of \autoref{compoundBlockEquivalences}]
    (4)~$\Rightarrow$ (1) In \Set, kernels are equivalence relations. Obviously, $\ker
    a\cup \ker b$ is both reflexive and symmetric. For transitivity,
    take $(x,y),(y,z) \in \ker a\cup \ker b$. Then $x, z \in [y]_a\cup
    [y]_b$. If $[y]_a\subseteq [y]_b$, then $x,z \in [y]_b$ and $(x,z)
    \in \ker b$; otherwise $(x,z) \in \ker a$.
    
    (1)~$\Rightarrow$ (2) In \Set, monomorphisms are stable under pushouts, so it is
    sufficient to show that $\ker a \cup \ker b$ is the kernel of the
    pushout of the regular epis from the image factorization of $a$
    and $b$, respectively. In other words, w.l.o.g.~we may
    assume that $a$ and $b$ are surjective maps, and we need to check that $\ker
    a\cup \ker b$ is the kernel of $p:= p_A\cdot a = p_B\cdot b$,
    where $p_A$ and $p_B$ are the two injections of the pushout below:
    \[
    \begin{mytikzcd}[baseline=(bot.base), ampersand replacement=\&]
        Y \ar[->>,d,swap,"b"] \ar[->>,r,"a"]
        \& A
        \ar[d,"p_A"]
        \\
        |[alias=bot]|
        B
        \ar[r,"p_B"]
        \& P.
        \pullbackangle{135}
    \end{mytikzcd}
    \]
    Let $\ker a\cup \ker b$ be the kernel of some $f\colon Y\to Y'$. Then,
    $f$ makes the projections of $\ker a$ (resp.~$\ker b$) equal and hence
    the coequalizer $a$ (resp.~$b$) induces a unique $f_A$
    (resp.~$f_B$) such that the triangles in the diagrams below commute:
    \[
    \begin{mytikzcd}
        \ker a
        \ar[r,hook]
        \ar[rr,shiftarr={yshift=5mm},"\pi_1"]
        \ar[rr,shiftarr={yshift=-5mm},"\pi_2"']
        &
        \ker a \cup \ker b
        \ar[r,shift left = 1,"{\pi_1}"]
        \ar[r,shift left = -1,"{\pi_2}"']
        & Y
        \ar[r,"f"]
        \ar[dr,->>,"a"']
        & Y'
        \\
        &&& A
        \ar[u,dashed,"\exists! f_A"']
    \end{mytikzcd}
    \quad
    \begin{mytikzcd}
        \ker b
        \ar[r,shift left = 1,"\pi_1"]
        \ar[r,shift left = -1,"\pi_2"']
        & Y
        \ar[r,"f"]
        \ar[dr,->>,"b"']
        & Y'
        \\
        && B
        \ar[u,dashed,"\exists! f_B"']
    \end{mytikzcd}
    \]
    Since $f_B \cdot b = f = f_A \cdot a$,  $(f_A,f_B)$ is a competing
    cocone for the above pushout. This induces a cocone morphism $m\colon (P,p_A,p_B)\to
    (Y,f_A,f_B)$, and we have 
    \begin{equation}\label{eq:ypp-y}
      m\cdot p = m\cdot p_A\cdot a = f_A\cdot a = f.
    \end{equation}
    We are ready to show that $\ker a\cup \ker b$ is a kernel of $p$. By the
    definition of $p$, the projections of $\ker a \cup\ker b$ are made equal by
    $p$. For the universal property, let $d_1\colon D'\to D$, $d_2\colon D'\to
    D$ such that $p\cdot d_1 = p\cdot d_2$. Then we have
    \[
      f\cdot d_1 \stackrel{\text{\eqref{eq:ypp-y}}}{=} m\cdot p \cdot d_1 =
      m\cdot p \cdot d_2 \stackrel{\text{\eqref{eq:ypp-y}}}{=} f\cdot d_2.
    \]
    This means that $d_1,d_2\colon D'\rightrightarrows D$ is a competing cone
    w.r.t.~the kernel of $f$.
    We thus obtain a unique cone morphism $u\colon D'\to \ker a \cup \ker b$ as
    desired.

    (2) $\Rightarrow$ (3) Take $x,y,z \in Y$ with $a(x)=a(y)$ and
    $b(y) = b(z)$. Then $a(x)$ and $b(z)$ are identified in the
    pushout $P$: 
    \[
      p(x) = p_A \cdot a(x) = p_A \cdot a(y) = p_B\cdot b(y) = p_B
      \cdot b(z) = p(z). 
    \]
    This shows that $(x,z)$ lies in $\ker a\cup \ker b$, hence we have
    that $a(x) = a(z)$ or $b(x) = b(z)$.

(3) $\Rightarrow$ (4) For a given $y\in Y$, there is nothing to show
in the case where $[y]_a \subseteq [y]_b$. Otherwise if
$[y]_a \not\subseteq [y]_b$, then there is some $x\in [y]_a$ that does
not lie in $[y]_b$, i.e.~such that $a(x) = a(y)$ and $b(x) \neq
b(y)$. Now let $z \in [y]_b$, i.e.~$b(y) = b(z)$. Then, by assumption,
$a(x) = a(y) = a(z)$ or $b(x) = b(y) = b(z)$. Since the latter does
not hold, we have $a(y) = a(z)$, i.e.~$z\in [y]_a$.
\end{proof}

\begin{example} \label{examplesRespectCompoundBlocks}
All $\op{select}$ routines in \autoref{exampleSelects} respect compound blocks.
To see this, let $k = \op{select}\big(
X\!\overset{y}{\twoheadrightarrow}\!Y\!
    \overset{z}{\twoheadrightarrow}\!Z
\big)
$.
\begin{enumerate}
\item For $S\in Y$ and $[S]_z\subseteq Y$, $k :=
  \chi_{\{S\}}^{[S]_z}\colon Y \to 3$ respects
  compound blocks using
  Lemma~\ref{compoundBlockEquivalences}\ref{classInclusion}. Indeed,
  one proceeds by case distinction on $B \in Y$:
  \begin{enumerate}
  \item If $B = S$, then $[B]_k = \{S\} \subseteq [S]_z = [B]_z$.
  \item If $B \in [S]_z$ and $B \neq S$, then $k(B) = 1$ so that
    \[
      [B]_k = [S]_z \setminus \{S\} \subseteq [S]_z = [B]_z.
    \]
  \item Finally, if $B \in Y \setminus [S]_z$, then we have
    $z(B) \neq z(S)$ and therefore
    $[B]_z \subseteq Y \setminus [S]_z = [B]_k$, where the latter
    equation follows from $k(B) = 0$.
  \end{enumerate}

  \takeout{%
  For $B \
    \begin{itemize}
        \item For $p \in Y\setminus C$, $z(p) \neq z(S)$ and so $[p]_z
        \subseteq Y\setminus C = [p]_k$.

        \item For $p \in C$, $z(p) = z(S)$ and so $[p]_k \subseteq C =
        [p]_z$.
      \end{itemize}
    }

\item The \op{select} routine returning the identity $\id_Y$ respects compound
blocks, because for any morphism $z\colon Y\to Z$, $\ker \id_Y \cup \ker z =
\ker z$ is a kernel.

\item The constant $k=\mathbin{!}$ respects compound blocks, because for all
  $B\in Y$: $[B]_z \subseteq Y = [B]_!$.
\qedhere
\end{enumerate}
\end{example}
\noindent For every pair $a\colon Y\to A$, $b\colon Y\to B$ of maps,
the kernel of $\fpair{a,b}\colon Y\to A\times B$ is the intersection
$\ker a\cap\ker b$. If the union $\ker a \cup \ker b$ is an
equivalence relation, then every block in the partition
$Y/\mathord{\ker\fpair{a,b}}$ is either is an equivalence class from $Y/\mathord{\ker a}$
or from $X/\mathord{\ker b}$. That this happens can be visually illustrated as
follows (see \autoref{subfig:trans}): Every equivalence class of
$\ker a\cap \ker b$ already appears in $\ker a$ or in $\ker b$ (or
both), i.e.~for all $x\in Y$, we have $[x]_a = [x]_{\fpair{a,b}}$ or
$[x]_b = [x]_{\fpair{a,b}}$. However for \autoref{subfig:nontrans},
$\{y_1,y_2,y_3\} = [y_3]_b\neq [y_3]_{\fpair{a,b}} \neq [y_3]_{a} =
\{y_3,y_4,y_5\}$.
In the following we prove formally that whenever $\ker a\cup\ker b$
is an equivalence relation, then every equivalence class of
$\ker\fpair{a,b}$ comes from one of $\ker a$ (i.e.~is in the set $A'$)
or of $\ker b$ (i.e.~is in the set $B'$).
\begin{lemma}
  \label{ABpartition}
  Let $a\colon Y\to A$,  $b\colon Y\to B$ such that
  $\ker a\cup \ker b$ is a kernel. Then there exist sets $A'$, $B'$,
  and maps $m$, $q$, $f_A$, and $f_B$ such that the following diagrams commute:
  \[
    \begin{mytikzcd}
      Y
      \arrow[shiftarr={yshift=6mm}]{rr}{\fpair{a,b}}
      \arrow{r}{q}
      & A'+B'
      \arrow{r}{m}
      & A\times B
    \end{mytikzcd}
    \quad
    \begin{mytikzcd}
      A' + \bar A'
      \arrow[equals]{d}
      & A' + B'
      \arrow{d}{m}
      \arrow{l}[swap]{A'+f_A}
      \arrow{r}{f_B+B'}
      & \bar B' + B'
      \arrow[equals]{d}
      \\
      A
      & A\times B
      \arrow{l}[swap]{\pi_1}
      \arrow{r}{\pi_2}
      & B
    \end{mytikzcd}
  \]
\end{lemma}
\begin{proof}
  Define the sets
  \begin{align*}
    Y_A &= \{ x \in Y\mid [x]_a\subseteq [x]_b \},
      &
      Y_B &= \{ x \in Y\mid [x]_a\supsetneqq [x]_b \},
      \\
      A' &= \{a(x)\mid x \in Y_A\},
      &
      B' &= \{b(x)\mid x \in Y_B\}.
  \end{align*}
  By \autoref{compoundBlockEquivalences}, $Y = Y_A + Y_B$. Next define
  $q = a' + b'\colon Y \cong Y_A+Y_B \to A' + B'$, where~$a'$ and~$b'$
  are the obvious restrictions of $a$ and $b$, respectively. Now
  put $\bar A' := A\setminus A'$, $\bar B' := B\setminus B'$, and
  define $f_A\colon B'\to \bar A'$
  and $f_B\colon A'\to \bar B'$ by
  \[
    f_A (b(x)) = a(x)
    \qquad
    f_B (a(x)) = b(x).
  \]
  These functions are well-defined by the definition of $Y_A$ and
  $Y_B$. Moreover, the
  codomain of $f_A$ restricts to~$\bar A'$, because $a(x) \in A'$
  implies that $[x]_a \subseteq [x]_b$, contradicting $x\in Y_B$.
  Analogously, the codomain of~$f_B$ restricts to $\bar B'$. Let $m$
  be the unique map such that the right-hand diagram above commutes,
  i.e.~$m$ is induced by the universal property of the product $A
  \times B$. Since  $(\id_{A'}+f_A)\cdot q = a$ and $(f_B +
  \id_{B'})\cdot q = b$, we see that the left-hand diagram above
  commutes. 
\end{proof}
\noindent This factorization is the main ingredient making~$H$
essentially commute with $\fpair{-,-}$:
\begin{proposition} \label{propZippable}
  Let $a\colon Y\to A$, $b\colon Y\to B$
  such that $\ker a\cup \ker b$ is a kernel, and let
  $H\colon \Set \to \Set$ be a zippable functor. Then we have
  \begin{equation}
    \ker \fpair{Ha,Hb} = \ker H\fpair{a,b}.
  \end{equation}
\end{proposition}
\begin{proof}
  Using the additional data provided by \autoref{ABpartition}, we note that
  the following diagram commutes:
  \[
    \begin{mytikzcd}
      HY
      \arrow{d}[swap]{Hq}
      \arrow{rr}{H\fpair{a,b}}
      &[26mm] & H(A\times B)
      \arrow{d}{\fpair{H\pi_1,H\pi_2}}
      \\
      H(A'+B')
      \arrow[>->]{r}[swap]{\fpair{H (A'+f_A), H(f_B+B')}}
      \arrow[>->]{urr}[sloped,above]{Hm}
      & H(A'+\bar A')
      \times H(\bar B'+ B')
      \arrow[equals]{r}
      & HA \times HB
    \end{mytikzcd}
  \]
  By \autoref{zippableExtended}, the composition at the bottom is a
  mono, because $H$ is zippable. Hence, $Hm$ is a mono as well.
  We conclude
  \[
    \ker \fpair{Ha, Hb} = 
    \ker \big(\fpair{H\pi_1,H\pi_2}
    \cdot H\fpair{a,b}\big)
    =
    \ker (H q)
    =
    \ker H\fpair{a,b}, 
  \]
  using \autoref{rem:kernel}\ref{i:mono} in the last two equalities.
\end{proof}
\begin{remark}
  Note that \autoref{propZippable} holds more generally for every
  functor $H\colon \Set \to \D$, where $\D$ is a finitely
  complete category. 
\end{remark}
\noindent For a $\op{select}$ routine respecting compound blocks, we
can now apply \autoref{propZippable} to prove the equivalence of
\eqref{kernelOptimization} and Step~\ref{step:P} of \autoref{catPT}:
\begin{theorem}\label{optimizationCorrectness}
  If $H\colon \Set\to\Set$ is zippable and $\op{select}$ respects
  compound blocks, then the optimization \eqref{kernelOptimization} is
  correct.
\end{theorem}
\begin{proof}
  Correctness of \eqref{kernelOptimization} means that
  \[
    P_{i+1} = P_{i+1}' = \ker(H\fpair{\bar q_i,q_{i+1}}\cdot \xi) = P_i \cap \ker (Hq_{i+1}\cdot \xi).
  \]
  Indeed, suppose that $H$ is zippable, and let $k =
  \op{select}(\kappa_{P_i}, f_i)$, where $f_i\colon
  X/P_i\twoheadrightarrow X/Q_i$ witnesses that $P_i$ is finer than
  $Q_i$ (see \autoref{catPT}). Then we have
  \begin{equation}\label{eq:aux}
    q_{i+1} = k \cdot \kappa_{P_i}
    \qquad
    \text{and}
    \qquad
    \bar q_i = m \cdot \kappa_{Q_i} = m \cdot f_i \cdot \kappa_{P_i},
  \end{equation}
  where $m$ is obtained by the image factorization of $\bar q_i$. By
  \autoref{rem:kernel}\ref{i:mono} we have $\ker f_i = \ker (m \cdot
  f_i)$. Since $\op{select}$ respects compound blocks we know that
  $\ker f_i \cup \ker k$ is a kernel, thus so is $\ker(m \cdot f_i)
  \cup \ker k$. By \autoref{propZippable}, we obtain
  \[
    \ker \fpair{H(m \cdot f_i), Hk} = \ker H \fpair{m \cdot f_i, k}, 
  \]
  which, using \autoref{rem:kernel}\ref{i:squares}, implies
  \begin{equation}\label{eq:aux2}
    \ker (\fpair{H(m \cdot f_i), Hk} \cdot H\kappa_{P_i} \cdot \xi)
    = 
    \ker (H\fpair{m \cdot f_i, k} \cdot H\kappa_{P_i} \cdot \xi). 
  \end{equation}
  Thus we obtain the desired result:
  \begin{align*}
    P_i \cap \ker(Hq_{i+1} \cdot \xi) 
    &= \ker\big(\fpair{H\bar q_i \cdot \xi, Hq_{i+1} \cdot \xi}\big) &
    \text{def.~of $P_i$}
    \\
    &= \ker\big(\fpair{H(m \cdot f_i), Hk} \cdot H\kappa_{P_i} \cdot
    \xi\big)
    & \text{by~\eqref{eq:aux}}
    \\
    &= \ker (H\fpair{m \cdot f_i, k} \cdot H\kappa_{P_i} \cdot \xi)
    & \text{by~\eqref{eq:aux2}}
    \\
    &= \ker\big(H \fpair{\bar q_i, q_{i+1}} \cdot \xi\big)
    & \text{by~\eqref{eq:aux}} 
    \\
    &= \ker\big(H\bar q_{i+1} \cdot \xi\big)
    & \text{def.~of $\bar q_{i+1}$}
    \\
    &= P_{i+1} & \text{def.~of $P_{i+1}$}.
                 \tag*{\qedhere}
  \end{align*}
  \takeout{
  \[
    \begin{array}[b]{cl}
&\op{select}\text{ respects compound blocks}
                 \\
\overset{\text{Def.}}{\Leftrightarrow} &
\text{$\ker f_i\cup \ker k$ is a kernel}
\\
\overset{\ref{propZippable}}{\Rightarrow}
& \ker\fpair{Hf_i, Hk} = \ker H\fpair{f_i, k_{i+1}}
\\
\overset{\text{\ref{rem:kernel}\ref{i:squares}}}{\Rightarrow}
  &
\ker(\fpair{Hf_i, Hk_{i+1}}\cdot H\kappa_{P_i}\cdot\xi) = \ker (H\fpair{f_i,
k_{i+1}}\cdot H\kappa_{P_i}\cdot\mathrlap{\xi )}
\\
\overset{\phantom{\ref{propZippable}}}{\Rightarrow}
  &
      P_i \cap \ker(Hq_{i+1}\cdot \xi) 
    = \ker \big(\fpair{H\bar q_i\cdot \xi, H q_{i+1}\cdot \xi}\big)
    = \ker \big(H\fpair{\bar q_i, q_{i+1}} \cdot \xi\big) \\
    & \phantom{P_i \cap \ker(Hq_{i+1}\cdot \xi) }= \ker \big(H\bar q_{i+1} \cdot \xi\big)
     = P_{i+1}
  \end{array}
  \tag*{\qedhere}
\]
}
\end{proof}
\noindent
Combining \autoref{optimizationCorrectness} and \autoref{thm:correct}
we obtain:
\begin{corollary}\label{optimizationSummary}
  Suppose that $H$ is a zippable endofunctor and that
  $\op{select}$ respects compound blocks and is progressing. Then \autoref{catPT} with
  optimization~\eqref{kernelOptimization} terminates and computes the
  simple quotient of a given finite $H$-coalgebra.
\end{corollary}

\begin{remark}
  Note that all results in this section can be formulated and proved
  in a Boolean topos $\C$ in lieu of $\Set$, e.g.~the category of
  nominal sets and equivariant maps. In particular, the set-theoretic
  statements in \autoref{compoundBlockEquivalences} and
  \autoref{ABpartition} can be formulated in the internal language of
  a Boolean topos, i.e.~the ordinary set theory ZF (with bounded
  quantifiers only) without the axiom of choice, but with the law of
  excluded middle. A detailed definition and discussion of this
  language can be found, e.g.\ in Mac Lane and
  Moerdijk~\cite{MacLane1992}.
\end{remark}

\section{Refinement Interfaces}
\label{sec:interface}
\tikzsetfigurename{partref-efficient-} In \autoref{catPT}, we left
unspecified how the kernels and partitions can be computed
efficiently. In this section we introduce the notion of a
\emph{refinement interface} for the given type functor. This interface
is aimed at efficient computation of the partition $X/P$ in
\autoref{catPT}, both in the initialization and in the
\emph{refinement step}, i.e.~in the optimization
\eqref{kernelOptimization} of Step~\ref{step:P} in \autoref{catPT}. In
\autoref{sec:efficient} we will assume that the type functor $H$ comes
equipped with a refinement interface which is used as a parameter for
a generic initialization procedure and an implementation
of~\eqref{kernelOptimization} such that partition refinement on a
coalgebra $\xi \colon X\to HX$ with $n= |X|$ states and $m$ edges runs
in $\CO((m+n)\cdot \log n)$. In order to phrase and actually achieve
this complexity bound, we need a notion of edges in coalgebras. This
notion will also provide a representation of coalgebras, for purposes
of using them as inputs to the partition refinement algorithm:
\begin{definition}\label{D:enc}
  An \emph{encoding} of a functor $H$ consists of a set $A$ of
  \emph{labels} and a family of maps
  $\flat\colon HX\to \Bag(A\times X)$, one for every set~$X$. The
  \emph{encoding} of an $H$-coalgebra $\xi\colon X\to HX$ is given by
  the map
  $\fpair{H!,\flat}\cdot \xi\colon X\to H1\times \Bag(A\times X)$ and we
  say that the coalgebra has $n = |X|$ states and
  $m = \sum_{x\in X}|\flat(\xi(x))|$ edges.
\end{definition}
\noindent The purpose of these functions is to enable a representation
of an input coalgebra $\xi\colon X\to HX$ as a labelled graph. In
concrete examples, an encoding expresses how one intuitively
visualizes coalgebras (see e.g.~\autoref{fig:exPartRef},
\autoref{exWrongRunForPP}, and \autoref{figSplitStep}). The reader
should note here that standard partition refinement for graphs on the
encoding does \emph{not} correctly minimize the input coalgebra, that
is, an encoding does by no means provide a reduction of the
minimization problem from coalgebras to graphs
(\autoref{fig:exPartRef} illustrates how Markov chain lumping and
bisimilarity in transition systems lead to different
partitions). Technically speaking, behavioural equivalence in
$H$-coalgebras is not the same as behavioural equivalence in
$H1\times\Bag(A\times (-))$-coalgebras. Moreover, it is not even
assumed that the map~$\flat:HX\to \Bag(A\times X)$ in an encoding is
natural in~$X$.
\begin{example} \label{exFunctorEncoding}
  \begin{enumerate}
  \item For $H=\Potf$, the labels are a singleton set $A=1$ and we define
    $\flat\colon \Potf X\to \Bag(1\times X)\cong \Bag(X)$ to be the obvious
    inclusion $\flat(t)(x) = 1$ if $x\in t$ and $\flat(t)(x) = 0$ if $x\not\in t$.
  \item For $HX=\R^{(X)}$, the labels are non-zero real numbers $A= \R_{\neq 0}$, and $\flat\colon
    \R^{(X)} \to \Bag(\R_{\neq 0}\times X)$ is defined as $\flat(t)(r,x) = 1$ if $t(x) =
    r \neq 0$ and $\flat(t)(r,x) = 0$ if $t(x)\neq r$ or $t(x) = 0$. An example
    of this encoding is visualized in \autoref{fig:exPartRef:Markov}.
  \item For a polynomial functor $H_\Sigma$ (see
    \autoref{ex:coalg}\ref{i:poly}), we have an encoding where the
    label set $A= \N$ describes the ordering of the components:
    $\flat\colon H_\Sigma X\to \Bag(\N\times X)$ is defined by
    $\flat(\sigma(x_1,\ldots,x_n)) = \{(1,x_1),\ldots,(n,x_n)\}$ for
    an $n$-ary symbol $\sigma \in \Sigma$. Note that the set $H_\Sigma1$ may
    be identified with the signature $\Sigma$ so
    $H!\cdot \xi$ assigns to each state $x$ of a $H_\Sigma$-coalgebra its
    output symbol, and the bag $\flat\cdot \xi(x)$ contains
    the successor states of $x$ together with their corresponding input $i \in \{1,
    \ldots, n \}$. 
  \end{enumerate}
\end{example}
\noindent Having a coalgebra encoded in terms of labelled edges, the following
computational task is solved in the refinement step: Given a coalgebra $\xi
\colon X\to HX$ as its encoding $X\to H1\times \Bag(A\times X)$ and partitions
$X/P$ and $X/Q$ such that $P$ is stable w.r.t.~$Q$ (see \autoref{D:stable}), how
does splitting a block $C\in X/Q$ into $S\subseteq C$ and $C\setminus S$ affect
the partition $X/P$? The explicit case for $H = \Potf$ is visualized in
\autoref{fig:RefineStep} and is computed as follows:

\begin{example}\label{exPaigeTarjanSplit}
  For the finite powerset functor, fix a coalgebra $\xi\colon X \to \Potf X$. We say
  that there is an edge from $x$ to a subset $S \subseteq X$ if $x$ has a
  successor node in $S$.
  
  In the Paige-Tarjan algorithm~\cite{PaigeTarjan87}, when a block
  $C\in X/Q$ is split into smaller blocks $S\subseteq C$ and
  $C\setminus S$, every block of $B\in X/P$ is split into (at most)
  three blocks as visualized in \autoref{fig:RefineStep}:
  \begin{enumerate}
  \item States $x\in B$ with edges to $S$ but not to $C\setminus S$.
  \item States $x\in B$ with edges to both $S$ and $C\setminus S$.
  \item States $x\in B$ with edges to $C\setminus S$ but not to $S$.
  \end{enumerate}
  It is one of the key points in the analysis of the overall time
  complexity that this split is performed in linear time in the number
  of ingoing edges into~$S$, regardless of the sizes of $C\setminus S$
  and~$B$. Roughly, this works as follows. All edges from a state
  $x\in X$ to the same block $C\in X/Q$ share an integer counter
  variable that stores how many such edges exist. Sharing a variable
  means that every edge is equipped with a pointer to a memory cell
  holding the actual counter. For example, in
  \autoref{fig:RefineStep}, $x_2\to y_2$, $x_2\to y_4$, and
  $x_2\to y_3$ share a variable with the value 3. When splitting $C$
  into $S$ and $C\setminus S$, we count the edges ending in $S$ and
  then compare the result with the counter variable of each of those
  edges, which allows us to implement the above three-way split by
  iterating over the source nodes $x \in B$ with edges to~$S$:
  \begin{figure}
    \input{parttree}
    \caption{The refinement step for $\xi\colon X \to \Potf X$, for $S\subseteq
      C\in X/Q$.}
    \label{fig:RefineStep}
  \end{figure}
  \begin{enumerate}
  \item If the counters match, then all edges from $x$ to $C$ go indeed to $S$,
    and we move all these~$x$ to a new block.
  \item If the values mismatch, then some edges from $x$ to $C$ go to
    $S$ and some to $C\setminus S$. Hence, we move all these $x$ to a
    new block and then update the counter variables. To this end, we
    count the number~$n$ of edges from $x$ to $S$. We re-use the
    existing counter variable for edges $x\to C$ and make it the
    variable for $x\to C\setminus S$ by decrementing it by $n$. For
    example, in \autoref{fig:RefineStep} the counter for the edge
    $x_2\to y_2$ is decremented from $3$ to $2$, and this counter
    variable is shared with $x_2\to y_3$ and $x_2\to y_4$, so the
    count for these edges is correctly decremented as well. We then
    allocate a new counter with value $n$ and change the pointer of
    all edges $x\to S$ to this new counter. In the example in
    \autoref{fig:RefineStep}, the pointer of the edge $x_2\to y_2$ now
    points to this new counter, which holds the value $1$.

  \item All the remaining $x \in B$ have no edge to $S$, and
    nothing needs to be done for them, because the counters for edges
    from $x$ need not change. 
  \end{enumerate}
  We thus have refined the partition $X/P$ such that it becomes
  (again) stable w.r.t.~the new partition $X/Q$ where $C$ is split
  into~$S$ and $C\setminus S$.
\end{example}

\noindent Note that in \autoref{exPaigeTarjanSplit}, the block $B$ is
split into smaller blocks in such a way that elements $x,x'$ end up in
the same block iff $\Potf\chi_S^C(\xi(x)) = \Potf\chi_S^C(\xi(x))$.
This split is done efficiently by maintaining an integer variable for
counting edges.  The present section generalizes this procedure from
$\Potf$ to a \Set-functor~$H$ that implements a so-called
\emph{refinement interface}, in which the integer counter generalizes
to a set $W$ \textqt{weights} and the comparison of counters
generalizes to an $\op{update}$ function:

\begin{definition}\label{refinementInterface}
  A \emph{refinement interface} for a $\Set$-functor~$H$ equipped with
  a functor encoding is formed by a set $W$ of \emph{weights} and
  functions
  \begin{align*}
    \begin{array}{l@{\qquad}c@{\quad}l}
      \op{init}\colon H1×\Bagf A\to W,&\op{update}\colon \Bagf A × W \to W× H3× W
    \end{array}
  \end{align*}
  such that there exists a family of \emph{weight maps} $w\colon \Potf X \to (HX \to W)$ 
  such that for all
  $S\subseteq C\subseteq X$, the diagrams
  \vspace{-2mm}
  \begin{equation}
    \label{eqSplitterLabels}
    \hspace{-4mm}
    \begin{mytikzcd}
      H1 × \Bagf A
      \arrow{r}{\op{init}}
      & W
      \\
      HX
      \arrow{u}[left, pos=0.4]{
        \fpair{
          H!,
          \op{fil}_X\cdot \flat}}
      \arrow{ur}[swap]{w(X)}
    \end{mytikzcd}
    \hspace{1mm}
    \begin{mytikzcd}[column sep=1.4cm]
      \Bagf A × W
      \arrow{r}{\op{update}}
      & W × H3 × W
    \\
    HX
    \arrow[
    ]{ur}[swap]{\fpair{w(S),H\chi_S^C, w(C\setminus S)}}
    \arrow{u}[left]{\fpair{\op{fil}_S\cdot \flat,w(C)}}
  \end{mytikzcd}
  \hspace{-3mm}
\end{equation}
commute, where \( \op{fil}_S\colon \Bagf(A×X) \to \Bagf(A)\)
is the filter function
\(\op{fil}_S (f)(a) = \textstyle\sum_{y \in S} f(a,y)\)
for $S\subseteq X$ (e.g.~$\op{fil}_X = \Bagf \pi_1$). 
\end{definition}
\begin{remark} \label{R:filS}
  \begin{enumerate}
  \item Note that for a coalgebra $\xi\colon X \to HX$, $\op{fil}_S \cdot
    \flat \cdot \xi(x)$ is the bag of labels of all edges from $x$ to $S
    \subseteq X$ in the encoding of $\xi$.
  \item Observe that $\op{fil}_S$ is natural in $A$, i.e.~for every
    map $h\colon A \to A'$ the following square commutes:
    \[
      \begin{tikzcd}
        \Bagf(A \times X)
        \arrow{r}{\op{fil}_S}
        \arrow{d}[swap]{\Bagf(h \times X)}
        &
        \Bagf A \arrow{d}{\Bagf h}
        \\
        \Bagf(A' \times X)
        \arrow{r}{\op{fil}_S}
        &
        \Bagf A'.
      \end{tikzcd}
    \]
    Indeed, we have for all $f\in \Bagf(A\times X)$
    \begin{align*}
      \op{fil}_S \cdot \Bagf(h \times X) (f)
      &= \op{fil}_S \big( (a',x) \mapsto \sum_{(h \times X)(a,y) =
        (a',x)} f(a,y) \big)
      \\
      &= \op{fil}_S \big((a',x) \mapsto \sum_{h(a) =
        a'} f(a,x) \big)
      \\
      &= a'\mapsto \sum_{y\in S} \sum_{h(a) = a'} f(a,y)
      \\
      &= a'\mapsto \sum_{h(a) = a'} \sum_{y\in S} f(a,y)
      \\
      &= \Bagf h \big( a \mapsto \sum_{y \in S} f(a,y)\big)
      \\
      &= \Bagf h \cdot \op{fil}_S(f).
    \end{align*}
  \end{enumerate}
\end{remark}
\noindent
Informally, for a given coalgebra $\xi\colon X\to HX$, the above
axioms for \op{init} and \op{update} can be understood as a kind of
contract that their implementation for a given functor needs to
fulfil: \op{init} receives in its first argument the information which
states of $X$ are (non-)terminating, in its second argument the bag of
labels of all outgoing edges of a state $x \in X$ in the graph
representation of $\xi$, and it returns the total weight of those
edges. The operation \op{update} receives a pair consisting of the bag
of labels of all edges from some state $x \in X$ into the set
$S \subseteq X$ and the weight of all edges from $x$ to
$C\subseteq X$, and from only this information (in particular
$\op{update}$ does not know $S$ and $C$ explicitly) it computes the
triple consisting of the weight of edges from $x$ to $S$, the result
of $H\chi_S^C \cdot \xi(x)$ and the weight of edges from $x$ to
$C\setminus S$ (in \autoref{exPaigeTarjanSplit}, the number of edges
from $x$ to $S$, the value for the three way split, and the number of
edges from $x$ to $C\setminus S$). The significance of the set $H3$ is that when using a
set $S\subseteq C\subseteq X$ as a splitter, we want to split every
block $B$ in such a way that it becomes \emph{compatible} with $S$ and
$C\setminus S$, i.e.~we group the elements $x\in B$ by the value of
$H\chi_{S}^C\cdot\xi (x) \in H3$. The set~$W$ depends on the
functor. But in most cases $W=H2$ and
$w(C) = H \chi_C\colon HX \to H2$ are sufficient.

In implementations, the refinement interface does not need to
provide~$w$ explicitly, because the algorithm will compute the values
of $w$ incrementally using \eqref{eqSplitterLabels}, and $\flat$ need
not be implemented because we assume the input coalgebra to be already
\emph{encoded} via $\flat$.
\begin{example} \label{examplePowerset} The refinement interface for
  the powerset functor $H=\Potf$ needs to count the edges into blocks
  $C$, so that we know in the refinement with $S\subseteq C$ whether
  there are edges to $C\setminus S$, as in
  \autoref{exPaigeTarjanSplit}. For a natural number $n$, we define
  the auxiliary function
  \[
    (-) \geZero 0\colon \N \to 2\qquad\text{by }(n\geZero 0) =
    \min(n,1) =
      \begin{cases}
        1 &  \text{if $n > 0$} \\
        0 &  \text{else.}
      \end{cases}
  \]
  For the refinement interface we use
  the following weights and labels:
  \[
    A = 1,
    \quad
    W = 2\times \N,
    \quad
    \text{$w\colon \Potf X \to (\Potf X \to 2\times \N)$ with
      $w(C)(t) = (|t\setminus C| \geZero 0, |t\cap C|)$.}
  \]
  For a map $\xi\colon X\to \Potf X$, a state $x\in X$, and a block
  $C\subseteq X$, the weight $w(C)(\xi(x))$ is a tuple whose first
  component tells us whether there is any edge from $x$ to
  $X\setminus C$ and whose second component is the number of edges from
  $x$ to $C$. Since we have a singleton label alphabet $A=1$, every
  bag of labels is just a natural number, because
  $\Bagf A = \Bagf 1 \cong \N$. The functions 
  \[
    \op{init}\colon \Potf 1\times \N \to 2\times \N
    \qquad\text{and}\qquad
    \op{update}\colon \N \times (2\times \N) \to (2\times \N)\times \Potf 3 \times (2\times \N)
  \]
  in the refinement interface for $\Potf$ are implemented as follows:
  \[
    \op{init}(z,n) = (0,n) \qquad\text{and}\qquad
    \begin{array}[t]{r@{\,}l@{}l}
    \op{update}(n_S,(r,n_C)) &=
    \big(&(r\vee (n_{C\setminus S}\geZero 0), n_S), \\[5pt]
    && (r, n_{C\setminus S}\geZero 0, n_S\geZero 0), \\[5pt]
    && (r\vee (n_{S}\geZero 0), n_{C\setminus S})\big),
  \end{array}
  \]
  where $n_{C\setminus S} := \max(n_C - n_S, 0)$,
  $\vee\colon 2\times 2\to 2$ is disjunction, and the middle return
  value in $\Potf 3$ is written as a bit vector of length three. The
  axioms in~\eqref{eqSplitterLabels} ensure that
  $n_S, n_C, n_{C\setminus S}$ can be understood as the numbers of
  edges to $S$, $C$, and $C\setminus S$, respectively.

  For the verification of \eqref{eqSplitterLabels}, note that for all $S\subseteq X$,
  the map $\op{fil}_S\cdot \flat\colon \Potf X\to \N$ is given by
  $\op{fil}_S(\flat(t)) = |t\cap S|$. Hence, we see that for all $t\in
  \Potf X$ we have
  \[
    \op{init}(\Potf!(t), \op{fil}_X(\flat(t)))
    = (0, \op{fil}_X(\flat(t)))
    = (|t\setminus X|\geZero 0, |t\cap X|)
    = w(X)(t).
  \]
  In the verification of the axiom of \op{update}, the parameters expand as
  follows:
  \[
    \op{update}(\op{fil}_S(\flat(t)), w(C)(t))
    = \op{update}(|t\cap S|, (|t\setminus C|\geZero 0, |t\cap C|))
    \Rightarrow
    \left\{
    \begin{array}{l@{\,}l}
      n_S &= |t\cap S| \\
      r &= |t\setminus C|\geZero 0 \\
      n_C &= |t\cap C| \\
    \end{array}\right.
  \]
  So $n_{C\setminus S} = \max(n_C-n_S, 0)= \max(|t\cap C| - |t\cap S|, 0) =
  |t\cap (C\setminus S)|$, since $S\subseteq C$. The remaining steps of the
  verification are performed component-wise: for $x =
  \op{update}(\op{fil}_S(\flat(t)), w(C)(t))$ we have
  \begin{align*}
    \qquad
    \pi_1(x)
    &= \big((|t\setminus C|\geZero 0) \vee (|t\cap (C\setminus S)|\geZero 0), |t\cap S|\big)
      \\ &
      = \big(|t\setminus S|\geZero 0, |t\cap S|\big)
      = w(S)(t),
    \\[5pt]
    \pi_2(x) &= (|t\setminus C|\geZero 0,
                          |t\cap (C\setminus S)| \geZero 0,
                          |t\cap S|\geZero 0)
    \\
    &= (\bigvee_{\substack{y\in t\\ \chi_S^C(y) = 0}}\!\!1,
    \bigvee_{\substack{y\in t\\ \chi_S^C(y) = 1}}\!\!1,
    \bigvee_{\substack{y\in t\\ \chi_S^C(y) = 2}}\!\!1)
    \qquad\qquad \text{as an element of $\Potf 3$}
    \\ &
    = \{\chi_S^C(y)\mid y\in t\}
    = \Potf\chi_S^C(t),
    \\[5pt]
    \pi_3(x) &=
    \big((|t\setminus C|\geZero 0) \vee (|t\cap S|\geZero 0), |t\cap (C\setminus S)|\big)
                      \\ &
     = \big(|t\setminus (C\setminus S)|\geZero 0, |t\cap (C\setminus S)|\big)
     = w(C\!\setminus\! S)(t).
  \end{align*}
\end{example}
\begin{example} \label{exampleH2} In the following examples, we always
  take
  \[
    W = H2
    \qquad\text{and}\qquad
    w(C) = H\chi_C\colon HX\to H2.
  \]
  We also use the helper function
  \[
    \op{val} := \fpair{H(=2), \id_{H3}, H(=1)}\colon H3\to H2×H3×H2,
  \]
  where $(=x)\colon 3\to 2$ is the equality check for $x\in \{1,2\}$, and in
  each case define
  \[
    \op{update} = (
    \Bagf A \times H2 \xrightarrow{\op{up}} H3
    \xrightarrow{\op{val}} H2 \times H3 \times H2),
  \]
  for a function $\op{up}\colon \Bagf A× H2\to H3$ that is defined
  individually for every functor.

  For the verification of \eqref{eqSplitterLabels} note that, in
  general, for $S\subseteq C\subseteq X$, we have
  \[
    \op{val}\cdot H\chi_S^C = \fpair{H\chi_{S}, H\chi_S^C,
      H\chi_{C\setminus S}}.
  \]
  Hence, to verify the axiom for
  $\op{update} = \op{val}\cdot \op{up}$ it suffices to verify that
  \begin{equation}\label{eq:toveri}
    \op{up}\cdot\fpair{\op{fil}_S\cdot \flat, H\chi_C } = H\chi_S^C;
  \end{equation}
  in fact, using $w(C) = H \chi_C$ we have:
  \begin{align*}
    \op{update}\cdot \fpair{\op{fil}_S \cdot\flat, w(C)}
    &= \op{val} \cdot \op{up} \cdot \fpair{\op{fil}_S \cdot\flat, H\chi_C}
    \\
    &= \op{val} \cdot H \chi^C_S \\
    &= \fpair{H\chi_S, H\chi^C_S, H\chi_{C\setminus S}} \\
    &= \fpair{w(S), H\chi^C_S, w(C\setminus S)}.
  \end{align*}
  
\begin{enumerate}
\item\label{exampleH2.1} For the monoid-valued functor $H= G^{(-)}\!$
  over an Abelian group $(G,+,0)$, we take labels $A = G$ and define
  $\flat(f) = \{ (f(y),y) \mid y\in X, f(y) \neq 0\}$ (which is
  finite because $f$ is finitely supported). With $W= H2= G×G$, the
  weight $w(C) = H\chi_C\colon HX\to G×G$ assigns to $f \in HX = G^{(X)}$
  the pair of accumulated weights of $X\setminus C$ and $C$ under
  $f$:
  \[
    w(C)(f) = \Big(\sum_{y \in X \setminus C} f(y), \sum_{y \in C} f(y)\Big).
  \]
  The remaining functions are
  \begin{align*}
    \op{init}(h_1, e) = (0, \groupsum e) 
    \quad
    \text{and}
    \quad
    \op{up}(e, (r,c)) =
    (r, c - \groupsum e,  \groupsum e),
  \end{align*}
  where $\groupsum\colon \Bagf G \to G$ is the obvious summation map
  assigning to a bag of elements of $G$ their sum in $G$.

  Then for all $f\in HX = G^{(X)}$ and $S\subseteq C\subseteq X$, we have:
  \begin{align*}
    \op{init}(H!(f), \op{fil}_X\cdot\flat(f)) &=
    \big(0, {\textstyle\sum}\Bagf\pi_1\cdot\flat(f)\big)
    \\ & =
    \big(0, \sum_{\mathclap{\substack{y\in X}}} f(y)\big) =
    G^{(\chi_X)}(f) = w(X)(f), 
    \\
    \op{up}(\op{fil}_S(\flat(f)), H\chi_C(f))
    &=
    \op{up}\big(
        \{f(y)\mid y\in S\},
        \big(\sum_{\mathclap{y\in X\setminus C}} f(y),
         \sum_{\mathclap{y\in C}} f(y)\big)
         \big)
    \\ &
    = \big(\sum_{\mathclap{y\in X\setminus C}} f(y),
           \sum_{\mathclap{y\in C}} f(y)- \sum_{\mathclap{y\in S}} f(y),
           \sum_{\mathclap{y\in S}} f(y)
      \big)
    \\ &
    \overset{\mathclap{S\subseteq C}}{=} ~~\big(\sum_{\mathclap{y\in X\setminus C}} f(y),
           \sum_{\mathclap{y\in C\setminus S}} f(y),
           \sum_{\mathclap{y\in S}} f(y)
      \big)
    = H\chi_S^C(f),
  \end{align*}
  which verifies~\eqref{eq:toveri} and therefore~\eqref{eqSplitterLabels}.
    
\item\label{item:R} As a special case, we obtain a refinement
  interface for the functor $\R^{(-)}$, and from this we can derive one
  for the distribution functor $\Dist$, a subfunctor of $\R_{\geq 0}^{(-)}$, the following
    $\op{init}$ and $\op{up}$ functions:
  \[
    \op{init}(h_1, e) = (0,1) \in \Dist 2 \subset [0,1]^2
    \qquad\text{and}\qquad
    \op{up}(e, (r,c)) = (r, c - \groupsum e, \groupsum e),
  \]
  if the latter lies in $\Dist 3$, and $\op{up}(e,(r,c)) = (0,0,1)$ otherwise.

  The axiom for $\op{init}$ clearly holds since for every
  $f \in \Dist X$, we have
  $\groupsum \Bag\pi_1 \cdot \flat(f) = \sum_{y \in X} f(y) =
  1$. 

  The axiom~\eqref{eq:toveri} for $\op{up}$ is proved as in the
  previous example; in fact, note that for an $f \in \Dist X$ all
  components of the triple 
  $ 
  \big(\sum_{y\in X\setminus C} f(y),
  \sum_{y\in C\setminus S} f(y),
  \sum_{y\in S} f(y)
  \big)
  $
  are in $[0,1]$ and their sum is $\sum_{y \in X} f(y) = 1$. Thus,
  this triple lies in $\Dist 3$ and is equal to $\Dist \chi^C_S(f)$. 

\item Similarly, one obtains a refinement interface for $\Bagf =
  \N^{(-)}$, adjusting the one for $\Z^{(-)}\!$; in fact,
  $\op{init}$ remains unchanged and $\op{up}(e, (r,c)) = (r,
  c-\groupsum e, \groupsum e)$ if the middle component is a natural
  number and $(0,0,0)$ otherwise.

  To verify~\eqref{eqSplitterLabels} for the refinement interface for
  $\N^{(-)}$ we argue similarly as in point~\ref{item:R} above: If
  $f$ lies in $\N^{(X)}\subseteq \Z^{(X)}$, i.e.~has only non-negative components,
  then only non-negative components appear in the data returned by $\op{up}$ and
  $\op{update}$, so $\op{up}$ and $\op{update}$ restrict from $\Z^{(-)}$ to
  $\N^{(-)}$ as desired.
  
\item Given a polynomial functor $H_\Sigma$
  for the signature $\Sigma$,
  recall from \autoref{exFunctorEncoding} that the labels $A = \N$
  records the generators contained in a shallow term together with their indices:
  \[
    \flat(\sigma(y_1,\ldots,y_n)) = \{ (1,y_1),\ldots,(n,y_n)\}.
  \]
  Then the functor interface is given by $w(C) = H\chi_C$ and
  \begin{align*}
    \op{init}(\sigma(0,\ldots,0), f) &= \sigma(1,\ldots,1)
    \\
    \op{up}(I, \sigma(b_1,\ldots,b_n))
    &= \sigma(b_1 + (1\in I),\ldots,b_i + (i\in I),\ldots,b_n + (n\in I)).
  \end{align*}
  Here $b_i + (i\in I)$ means $b_i + 1$ if $i\in I$ and $b_i$ otherwise.

  To verify~\eqref{eqSplitterLabels}, let $S \subseteq C \subseteq X$,
  $t=\sigma(y_1,\ldots,y_n) \in H_\Sigma X$ with $\sigma$ of arity
  $n$, let $b_i = \chi_C(y_i)$, and
  $I= \{ 1 \le i \le n \mid y_i \in S\}$. Then we have:
  \begin{align*}
    \op{init}(H_\Sigma!(t), \Bagf\pi_1\cdot \flat(t))
    &= \op{init}(\sigma(0,\ldots,0), \Bagf\pi_1(\{(1,y_1),\ldots,(n,y_n) \}))
    \\ &
     = \op{init}(\sigma(0,\ldots,0), \{1,\ldots,n \})
     = \sigma(1,\ldots,1)
    \\ &
     = \sigma(\chi_X(y_1),\ldots,\chi_X(y_n))
     = H_\Sigma \chi_X (t).
     \\
     \op{up}(\op{fil}_S\cdot \flat(t), H_\Sigma\chi_C(t)) &=
     \op{up}(\op{fil}_S(\{(1,y_1),\ldots,(n,y_n)\}),
     \sigma(b_1,\ldots,b_n))
     \\ &
     = \op{up}(I, \sigma(b_1,\ldots,b_n))
     \\ &
     = \sigma(b_1 + (1\in I),\ldots,b_i + (i\in I),\ldots, b_n + (n\in I))
     \\ &
     = \sigma(\chi_S^C(y_1),\ldots,\chi_S^C(y_i),\ldots, \chi_S^C(y_n))
     \\ &
     =H_\Sigma\chi_S^C(t).
   \end{align*}
   In the second last step we use that:
   \[
     \begin{array}[b]{lcl}
       y_i \in X\setminus C
       &\Rightarrow& b_i + (i\in I) =
       0 + 0 = 0 = \chi_S^C(y_i),
       \\
       y_i \in C\setminus S
       &\Rightarrow& b_i + (i\in I) =
       1 + 0 = 1 = \chi_S^C(y_i),
       \\
       y_i \in S
       &\Rightarrow& b_i + (i\in I) =
       1 + 1 = 2 = \chi_S^C(y_i).
     \end{array}
     \tag*{\qedhere}
   \]
\end{enumerate}
\end{example}
\noindent
By~\autoref{P:reduction}, the interface for Abeliean-group-valued
functors can also be used for monoid-valued functors~$M^{(-)}$ for
cancellative monoids~$M$ (as these embed into Abelian groups). In
further work~\cite{DMSW19}, we provide a refinement interface for
monoid-valued functors~$M^{(-)}$ over unrestricted monoids~$M$; this
also yields a refinement interface for $\Potf$ as the monoid-valued
functor $(2,\vee,0)^{(-)}$. However, the more general refinement
interface is less efficient than the specific interface for $\Potf$
described in \autoref{examplePowerset}; in particular it does not
yield the linear run-time complexity we seek here and require next in
\autoref{interface-linear}.

The next result shows that for every refinement interface the weight
function $w(C)$ provides at least as much information as $H\chi_C$.
\begin{proposition}\label{prop:atleast}
    For every refinement interface, $H\chi_C = H (=1) \cdot \pi_2\cdot
    \op{update}(\emptyset)\cdot w(C)$.
\end{proposition}

\begin{proof}
    The axiom for $\op{update}$ and definition of $\op{fil}_\emptyset$ makes
    the following diagram commute:
    \[
    \begin{mytikzcd}[sep=1.4cm,baseline=(bot.base)]
    HX
    \arrow[
        rounded corners,
        to path={
            -- ([yshift=-5mm]\tikztostart -| \tikztotarget) \tikztonodes
            -- (\tikztotarget)
        },
    ]{dr}[pos=0.4,sloped,below]{\fpair{w(\emptyset),H\chi_\emptyset^C,
    w(C\setminus \emptyset)}}
    \arrow[
        rounded corners,
        to path={
            (\tikztostart.east)
            -- ([yshift=-4.5mm]\tikztostart -| \tikztotarget) \tikztonodes
            -- (\tikztotarget)
        },
    ]{drr}[pos=0.8,below]{H\chi_\emptyset^C}
    \arrow[
        rounded corners,
        to path={
            (\tikztostart.east)
            -- (\tikztostart -| \tikztotarget) \tikztonodes
            -- (\tikztotarget)
        },
    ]{drrr}[pos=0.8,below]{H\chi_C}
    \arrow[shiftarr={xshift=-20mm}]{d}[left]{\fpair{\emptyset!, w(C)}}
    \arrow{d}[left]{\fpair{\flat\cdot \op{fil}_\emptyset,w(C)}}
    \\
    |[alias=bot]|
    \Bagf A × W
    \arrow{r}{\op{update}}
    & W × H3 × W
    \arrow{r}{\pi_2}
    & H3
    \arrow{r}{H(=1)}
    & H2
    \end{mytikzcd}
    \tag*{\qedhere}
    \]
\end{proof}

\begin{assumption}\label{interface-linear}
  From now on, we assume that $H\colon \Set \to \Set$ is zippable and given together with
  a refinement interface such that $\op{init}$ and $\op{update}$ run
  in linear time, $H3$ is linearly ordered, and its elements can be
  compared in constant time.
\end{assumption}
\begin{remark}\label{rem:ram}
  We implicitly impose some standard assumptions regarding arithmetic
  on our computational model, namely that integers can be stored in
  atomic memory cells and the usual operations on them, e.g.~addition
  and comparison, run in constant time.
\end{remark}
\begin{example}
  \label{exampleTime}
  The refinement interfaces in Examples~\ref{examplePowerset}
  and~\ref{exampleH2} satisfy \autoref{interface-linear}:
  
  \begin{enumerate}
  \item
  For all examples using the $\op{val}$-function from \autoref{exampleH2},
first note that $\op{val}$ runs in linear time (with a constant factor of 3,
because $\op{val}$ essentially returns three copies of its
input). Hence $\op{update}$ runs in linear time if $\op{up}$ does.

\item For monoid-valued functors $G^{(-)}$ over an abelian group~$G$,
  for $\N^{(-)}$ and for $\Dist$, all the operations, including the
  summation $\groupsum e$ (see \autoref{exampleH2}.\ref{exampleH2.1}),
  run in time linear in the size of the input. If the elements
  $g\in G$ have a bounded representation, i.e.\ fit into boundedly
  many memory cells, then so do the elements of $G^{(2)}$; thus
  comparing elements of $g_1,g_2\in G^{(2)}$ can be performed in
  constant time. (By our global assumptions as per \autoref{rem:ram},
  this includes cases where~$G$ consists of integer or rational
  numbers.)
\item Given a polynomial functor $H_\Sigma$, we assume that operation
  symbols $\sigma\in \Sigma$ are encoded as integers.  As per
  \autoref{rem:ram}, we can then assume that operation symbols can be
  compared in constant time. If the signature has \emph{bounded
    arities} (i.e.\ there is a finite bound on the arity of all
  symbols in $\Sigma$), then the maximum arity of operation symbols
  present in a given coalgebra $\xi:X\to H_\Sigma X$ is bounded
  independently of~$\xi$, so the comparison of the generators in two
  shallow $\Sigma$-terms $t_1,t_2\in H_\Sigma 3$ runs in constant time
  as well.
\begin{enumerate}
\item The first parameter of $\op{init}$ is of type $H_\Sigma 1$ and
  can be encoded by an operation symbol~$\sigma$, i.e.~by an
  integer. Let $t\in H_\Sigma 2$ be fixed. Then we explicitly
  implement \vspace{2mm}
\[
\op{init}(\sigma, f) =\begin{cases}
    \sigma(\smash{\overbrace{1,\ldots,1}^{\mathclap{\text{arity
    $\sigma$ many}}}})
    & \text{if arity}(\sigma) = |f|
    \\
    t & \text{otherwise.}
\end{cases}
\]
Both the check and the construction of $\sigma(1,\ldots,1)$ run in
linear time in the size of $f \in \Bagf \N$, or in constant time in the
second case, since $t$ was fixed beforehand.

\item In $\op{up}(I, \sigma(b_1,\ldots,b_n))$, we cannot naively check
  all the $1 \in I,\ldots,n\in I$ queries, since this would lead to
  quadratic run-time. Instead, we precompute all results of possible
  queries:

\begin{algorithmic}[1]
    \State Define an array $\op{elem}$ with indices $1\ldots n$, where
    each cell stores a value in $2$.
    \State Initialize \op{elem} to $0$ everywhere.
    \State \textbf{for} $i\in I$ with $i\le n$ \textbf{do} $\op{elem}[i] \gets 1$.
    \State \textbf{return} $\sigma(b_1 + \op{elem}[1], \ldots, b_n +
    \op{elem}[n])$.
\end{algorithmic}
The running time of every line is bound by $|I|+n$.
\end{enumerate}
\end{enumerate}
\end{example}

\section{Efficient Partition Refinement}
\label{sec:efficient}

We now proceed to present \autoref{catPT} in a concrete form that is
parametric in a refinement interface for the coalgebraic type functor
$H$. We will prove the correctness in \autoref{sec:correct} and then
analyse the efficient run-time in \autoref{sec:time}.

We continue to work under \autoref{interface-linear}.

\subsection{The Concrete Algorithm}

In the actual implementation, we need to address edges
explicitly, and so we define the following maps:
\begin{definition}
  \label{defTypeGraph}
  Given a functor $H$, equipped with an encoding, and a map
  $\xi\colon X\to HX$, the set~$E$ of \emph{edges} is defined by
  \[
    E := \coprod_{x\in X} \coprod_{(a,y) \in A\times X} (\flat\cdot\xi (x))(a,y)
  \]
  where $(\flat\cdot\xi (x))(a,y)\in \N$ is considered as a finite
  ordinal number. The \emph{encoding} of $\xi$ is represented by two
  functions (implemented as arrays):
  \[
    \op{type} = \big(X\xrightarrow{\xi} HX \xrightarrow{H!} H1\big)
    \quad
    \text{and}
    \quad
    \op{graph}\colon E\longrightarrow X\times A\times X
    \ \text{with}\ 
    \inj_x(\inj_{a,y} n) \mapsto (x,a,y).
  \]
\end{definition}

\begin{remark}
  For the complexity result, we assume that the partitions
  $\unnicefrac{X}{P}$ and $\unnicefrac{X}{Q}$ are implemented in such
  a way that we can add and remove elements in constant time, remove
  and create blocks in constant time, and find the surrounding block
  of an element $x\in X$ or $y\in X$ in constant time.
  \begin{enumerate}
  \item One possible implementation is by doubly linked lists of the
    blocks, where each block is in turn encoded as a doubly linked
    list of its elements, and additionally every element $x\in X$,
    $y\in X$ holds a pointer to the corresponding list entry in the
    blocks containing them~\cite{PaigeTarjan87}.

  \item An alternative implementation is the \emph{refinable partition} data structure~\cite{ValmariF10},
  in which the partition is an array of elements and elements in the same block
  appear consecutively in the array. So a block of the partition consists of two
  indices, denoting the interval in the array of all elements.
  \end{enumerate}
  
  \noindent Both implementations require $\CO(n)$ space, where $n$ is
  the number of elements, but the linked list approach has a much
  higher constant factor due to the high number of pointers. Thus, the
  implementation of our algorithm~\cite{DMSW19} uses the refinable
  partition structure.
\end{remark}
\noindent The algorithm maintains the following mutable
data structures:
\begin{itemize}
\item An array $\op{toSub}\colon X\to \Bagf E$, mapping $x\in
  X$ to its outgoing edges ending in the currently processed subblock.
\item A pointer mapping edges to memory addresses: $\op{lastW}\colon E \to \N$.
  By this pointer, we achieve that edges with a common source node $x\in X$ and
  a common target block $C\in X/Q$ point to the same memory cell holding
  $w(C,\xi(x))$, as in the concrete \autoref{exPaigeTarjanSplit} above.
\item The actual store for weights $\op{deref}\colon \N\to W$, into which
  $\op{lastW}$ points.
\item For each block $B \in X/P$ a set of markings $\op{mark}_B \subseteq B×\N$, that
  collects those pairs consisting of a state $x\in B$ and the pointer to
  $w(C,\xi(x))$ in the store $\op{deref}$ of weights for those $x$ that have an outgoing
  edge to the current subblock $S$ during the splitting operation. Initially,
  $\op{mark}_B$ is empty for every newly created block $B$ and after each refinement step,
  $\op{mark}_B$ is emptied again.
\end{itemize}
\begin{notation}
  In the following we write $e = x\xrightarrow{a}y$ in lieu of
  $\op{graph}(e) = (x,a,y)$. We also overload notation and write the
  weight function from the refinement interface of $H$ in its
  uncurried form as $w\colon \Potf X \times HX \to W$.
\end{notation}
\begin{definition}[Invariants]
  Our correctness proof below establishes that the following
  properties hold before and after each call to our splitting routine;
  we  call them \emph{the invariants}:
\begin{enumerate}
\item\label{invariant_toSub}
  The array $\op{toSub}$ is empty, i.e.~for all $x\in X, \op{toSub}(x) = \emptyset$.

\item\label{invariant_lastW}
  The pointers in $\op{lastW}$ are the same for two edges if and only
  if their source nodes and target blocks agree: for every $e_1 =
  x_1\xrightarrow{a_1}y_1$ and $e_2 = x_2\xrightarrow{a_2}y_2$ we have
  \[
    \op{lastW}(e_1)
    = \op{lastW}(e_2)
    \quad \Longleftrightarrow\quad %
    x_1=x_2 \text{ and }[y_1]_{\kappa_Q} = [y_2]_{\kappa_Q}.
  \]

\item\label{invariant_deref} The pointers in $\op{lastW}$ point to the
  correct weights in the store of weights, i.e.~for every
  $e = x\xrightarrow{a} y$ in $E$ and
  $C := [y]_{\kappa_Q} \in \unnicefrac{X}{Q}$, we have
  \[
    w(C, \xi(x)) = \op{deref}\cdot \op{lastW}(e).
  \]

\item\label{invariant_Hchi} For every block $C\in X/Q$, The partition $X/P$ is stable
  w.r.t.~$\chi_C$, i.e.~for every $x_1,x_2\in B\in \unnicefrac{X}{P}$ and
  $C\in \unnicefrac{X}{Q}$, we have
  $(x_1,x_2) \in \ker (H\chi_C\cdot \xi)$, cf.~\autoref{D:stable}.
\end{enumerate}
\end{definition}

In the following code listings, we use square brackets for
array lookups and updates in order to emphasize they run in constant
time. We assume that the functions
$\op{graph}\colon E \to X \times A \times X$ and $\op{type}\colon X \to H1$ are
implemented as arrays.  In the initialization step, we precompute the additional
static array $\op{pred}\colon X\to \Potf E$,
\[
  \op{pred}(y) = \{ e\in E \mid e = x\xrightarrow{a} y \}
\]
which holds for every $y\in X$ the set of incoming edges. 

Sets and bags are implemented as lists. We only insert
elements into sets not yet containing them.

\begin{definition}
  We say that we \emph{group} (or \emph{split}) a finite set $Z$ by a map
  $f\colon Z \to Z'$ to indicate that we compute $[-]_f$, i.e.~we
  partition $Z$ according the values of its elements under~$f$.
\end{definition}
This is implemented by first sorting the elements $z\in Z$ by a binary encoding
of $f(z)$ using any $\CO(|Z|\cdot \log|Z|)$ sorting algorithm, and
then grouping elements with the same $f(z)$ into blocks. In order to
keep the overall complexity for the grouping operations low enough,
one needs to use a possible majority candidate during sorting,
following Valmari and
Franceschinis~\cite{ValmariF10}. 
\begin{figure}[t]
\textsc{Initialization}\hspace*{\fill}
\begin{algorithmic}[1]
    \For{$e\in E$, $e = x\xrightarrow{a} y$}
        \State add $e$ to $\op{toSub}[x]$ and $\op{pred}[y]$\label{ini2}
    \EndFor
    \For{$x\in X$}
        \State $p_X \gets$ new cell in $\op{deref}$ containing
        $\op{init}(\op{type}[x],
        \Bagf(\pi_2\cdot\op{graph})(\op{toSub}[x]))$\label{ini4}
        \State \textbf{for} $e\in \op{toSub}[x]$ \textbf{do}
        \label{algoLineInit} $\op{lastW}[e] = p_X$
        \State $\op{toSub}[x] := \emptyset$
        \label{algoLineInitEmptyset}
    \EndFor
    \State $\unnicefrac{X}{P} \gets $ group $X$ by $\op{type}\colon X \to H1$;
    $\unnicefrac{X}{Q} \gets \{ X\}$.
\end{algorithmic}
\vspace{-1mm}
\caption{The initialization procedure.}
\vspace{-3mm}
\label{figAlgoInit}
\end{figure}
The algorithm computing the initial partition is listed in
\autoref{figAlgoInit}. The first two lines initialize $\op{toSub}(x)$
to be the bag of all outgoing edges of $x$ and compute
$\op{pred}(y)$. Then the loop in lines 3--6 initializes the array
$\op{lastW}$, and finally the two partitions $X/P$ and $X/Q$ are
initialized in line 7.

\begin{lemma}\label{lemmaInit}
The initialization procedure runs in time $\CO(|E|+|X|\cdot\log|X|)$ and
the result satisfies the invariants.
\end{lemma}

\begin{proof}
  The grouping in line~7 takes $\CO(|X|\cdot \log|X|)$ time. The first
  loop takes $\CO(|E|)$ steps, and the second one takes $\CO(|X| + |E|)$ time in total over all $x\in
  X$ since $\op{init}$ is assumed to run in linear time. For the invariants:
    \begin{enumerate}
    \item By line~\ref{algoLineInitEmptyset}.

    \item Let $e_1 = x_1\xrightarrow{a_1}y_1$ and
      $e_2 = x_2\xrightarrow{a_2}y_2$. Since
      $[y_1]_{\kappa_Q} = [y_2]_{\kappa_Q}$ holds for all $y_1$,
      $y_2$, it suffices to show that $\op{lastW}(e_1) =
      \op{lastW}(e_2)$ iff $x_1=x_2$. This holds after the procedure
      because $p_X$ in line~\ref{ini4} is the address of a new memory cell,
      whence $\op{lastW}(e_1)$ and $\op{lastW}(e_2)$ are equal iff
      they are assigned their value in the same \textbf{for} loop in line~\ref{algoLineInit},
      equivalently, if $e_1, e_2 \in \op{toSub}(x)$ for some
      $x$. Equivalently, $x_1 = x = x_2$ because after line~\ref{ini2}, $\op{toSub}(x)$ is the bag of
      all outgoing edges of $x$. 

    \item First, we have for every $x \in X$ that
      \begin{align*}
        \Bagf(\pi_2\cdot\op{graph})(\op{toSub}(x))
        &= \Bagf(\pi_2\cdot\op{graph})(\{e \mid e \in E, e = x
        \xrightarrow{a} y\}
        & \text{(line \ref{ini2})}
        \\
        &= \{a \mid \text{$x \xrightarrow{a} y$ in
          $E$}\}
        & \text{(Def.~\ref{defTypeGraph})}
        \\
        &= \Bagf\pi_1\big(\{(a,y) \mid \text{$x \xrightarrow{a} y$ in
          $E$}\}\big)
        \\
        &= \Bagf\pi_1 \cdot \flat \cdot \xi(x)
        & \text{(Def.~\ref{defTypeGraph})}
      \end{align*}
      where the comprehensions are read as \emph{multiset} comprehensions,
      i.e.~multiple edges with the same label~$a$ generate multiple occurrences
      of~$a$.
      Then we use this in the second step below to see that we have for
      every $e = x \xrightarrow{a}y$ in $E$: 
      \begin{align*}
        \op{deref}\cdot\op{lastW}(e)
        &= \op{init}(\op{type}(x),
        \Bagf(\pi_2\cdot\op{graph})(\op{toSub}(x)))
        & \text{(lines \ref{ini4} and \ref{algoLineInit})}
        \\
        &= \op{init}(\op{type}(x), \Bagf\pi_1\cdot \flat\cdot\xi(x))
        \\
        &=\op{init}(H! \cdot \xi(x), \Bagf\pi_1\cdot \flat\cdot\xi(x))
        & \text{(Def.~of \op{type})}
        \\
        &= w(X, \xi(x))
        & \text{by~\eqref{eqSplitterLabels},}
      \end{align*}
      and we are done since $X/Q = \{X\}$.

    \item Since $\ker(H\chi_X\cdot\xi) = \ker(H!\cdot\xi)$, this is
    just the way $\unnicefrac{X}{P}$ is constructed.\qedhere
    \qedhere
    \end{enumerate}
\end{proof}

\noindent
The algorithm for a single refinement step is listed in \autoref{figSplitAlgo}.
It receives as input the the partition $X/P$, where the kernel $P$ is
stable w.r.t.~$\kappa_Q\colon X \epito X/Q$, and a subblock $S
\subseteq X$ contained in $C \in X/Q$. Its task is to split
the blocks in $X/P$ in such a way that $P$ becomes stable
w.r.t.~$\fpair{\kappa_Q,\chi_S^C}$, i.e.~stable
w.r.t.~$\kappa_Q$ after $C$ has been split into $S$ and $C/S$ in
$X/Q$.
\begin{figure}[t]
\begin{subfigure}[t]{.45\textwidth}
\textsc{Split}$(\unnicefrac{X}{P}, S\subseteq X)$
\begin{algorithmic}[1]
\State{$\op{M} \gets \emptyset \subseteq \unnicefrac{X}{P}×H3$}
\noshowkeys\label{algoFirstLine}
\For{$y\in S, e\in \op{pred}[y]$} \label{algoLineToSubLoop}
    \State $x\xrightarrow{a}y \gets e$
    \State $B \gets$ block with $x\in B\in \unnicefrac{X}{P}$
    \If{$\op{mark}_B$ is empty}
        \State $w_C^x \gets \op{deref}\cdot\op{lastW}[e]$
        \State $v_\emptyset \gets
        \pi_2\cdot\op{update}(\emptyset,w_C^x)$
        \State add $(B,v_\emptyset)$ to $\op{M}$\label{algo8}
    \EndIf
    \If{$\op{toSub}[x] = \emptyset$}
        \State add $(x,\op{lastW}[e])$ to $\op{mark}_B$\label{algo10}
    \EndIf
    \State add $e$ to $\op{toSub}[x]$\label{algo11}
\EndFor
\algstore{splitAlgo}
\end{algorithmic}
\end{subfigure}\hfill%
\begin{subfigure}[t]{.52\textwidth}
\begin{algorithmic}[1]
\algrestore{splitAlgo}
\For{$(B,v_\emptyset) \in \op{M}$}
    \noshowkeys\label{algoSplitLoop}
    \State $B_{\neq\emptyset} \gets \emptyset \subseteq X × H3$
    \For{$(x,p_{C})$ in $\op{mark}_B$}
        \noshowkeys\label{algoLineMarkedX}
        \State $\ell \gets \Bagf(\pi_2\cdot \op{graph})(\op{toSub}[x])$
            \noshowkeys\label{algoLineLabels}
        \State $(w_S^x, v^x,w_{C\setminus S}^x) \gets
            \op{update}(\ell, \op{deref}[p_{C}]\mathrlap{)}$%
\noshowkeys\label{algoLineUpdate}
        \State $\op{deref}[p_{C}] \gets w^x_{C\setminus S}$
        \State $p_S\gets $ new cell containing $w_S^x$
            \noshowkeys\label{algoLineNewCell}
        \State \textbf{for} $e \in \op{toSub}[x]$ \textbf{do} $\op{lastW}[e] \gets p_S$
            \noshowkeys\label{algoLineSetLastW}
        \State $\op{toSub}[x] \gets \emptyset$
        \If{$v^x \neq v_\emptyset$} \noshowkeys\label{algoLineIfVXVempty}
            \State remove $x$ from $B$
            \State insert $(x,v^x)$ into $B_{\neq\emptyset}$
            \noshowkeys\label{algoLineInsertVX}
        \EndIf
    \EndFor
    \State $\op{mark}_B := \emptyset$
    \label{algoLastLine}
    \State \label{algoLineSplit}
        \(
        \begin{array}[t]{@{}l}
        B_1×\{v_1\},\ldots,B_\ell×\{v_\ell\} \gets
            \\
            \quad\text{group }B_{\neq\emptyset}\text{ by }\pi_2\colon X×H3\to H3
        \end{array}
        \)
    \State insert $B_1,\ldots,B_\ell \gets$ into $\unnicefrac{X}{P}$
\EndFor
\end{algorithmic}
\end{subfigure} \\[2mm]
\begin{subfigure}[t]{.45\textwidth}
    \caption{Collecting predecessor blocks}
    \label{algoPredCollecting}
\end{subfigure}\hfill%
\begin{subfigure}[t]{.52\textwidth}
    \caption{Splitting predecessor blocks} \noshowkeys\label{algoPredSplitting}
\end{subfigure}
\caption{Refining $\unnicefrac{X}{P}$ for $S\subseteq C$, $C\in X/Q$,
  i.e.~making it stable w.r.t.~$ \chi_S^C\colon X\to 3$.}
\label{figSplitAlgo}
\end{figure}

In the first part, all blocks $B\in \unnicefrac{X}{P}$
that have an edge into $S$ are collected, together with
$v_\emptyset= H\chi_S^C\cdot\xi(x) \in H3$ for all
$x\in B$ that have no edge into $S$. For each $x\in X$, $\op{toSub}[x]$
collects the edges from $x$ into $S$. The markings $\op{mark}_B$ list
those elements $x\in B$ that have an edge into $S$, together with one
of the pointers in $\op{lastW}$ that point to $w(C,x)$ in the store
$\op{deref}$ of weights. 

In the second part, each block $B$ with an edge into $S$ is split
by $H\chi_S^C\cdot \xi$, cf.~\autoref{exPaigeTarjanSplit}. First, for
every $(x,p_C)\in \op{mark}_B$, we compute $w(S,x)$, $v^x = H\chi_S^C\cdot\xi(x)$,
and $w(C\setminus S,x)$ using $\op{update}$. Then, the weight of all edges $x
\to C\setminus S$ is updated to $w(C\setminus S, x)$ and the weight of all edges
$x\to S$ is stored in a new cell containing $w(S,x)$. For all unmarked $x\in B$,
we know that $H\chi_S^C\cdot\xi(x) = v_\emptyset$; so all $x$ with
$v^x=v_\emptyset$ stay in $B$. All other $x\in B$ are removed from
$B$ and collected in $B_{\neq\emptyset}$ together with their value
$v^x = H\chi_S^C \cdot \xi(x)$, and then we group $B_{\neq\emptyset}$
by these values to obtain the new blocks $B_1, \ldots, B_\ell$ that we
add to $X/P$. 

\medskip Now we are ready to combine $\textsc{Split}$ from
\autoref{figSplitAlgo} with what we saw in Sections~\ref{sec:cat}
and~\ref{sec:opti} and instantiate \autoref{catPT} with the
$\op{select}$ routine from
{\autoref{exampleSelects}}.\ref{exampleSelectsChi}, i.e.~we have
\begin{equation}
  \label{eq:qi++}
  q_{i+1} = \op{select}(\kappa_{P_i},\kappa_{Q_i}) \cdot \kappa_{P_i}
  = \chi_{S_{i}}^{C_{i}},
\end{equation}
where $2\cdot |S_i|\le |C_i|$, $S_i,C_i\subseteq X$ in line \ref{step:S}, and
we replace line \ref{step:P} of the algorithm by
\begin{equation}
  \unnicefrac{X}{P_{i+1}} = \textsc{Split}(\unnicefrac{X}{P_{i}}, S_{i}),
  \tag{\ref{kernelOptimization}'}\label{eq:keropti2}
\end{equation}
where the indices are merely intended to facilitate the
analysis. 

This yields the following more concretion of
\autoref{catPT}:
\begin{algorithm} \label{algPTfinal}
  Given the \emph{encoding} of an $H$-coalgebra $\xi\colon X\to HX$ as
  input, do the following:

  \begin{itemize}
  \item Run \textsc{Initialization} (\autoref{figAlgoInit}).

  \item Iterate the following steps while $X/P$ is properly finer than $X/Q$:
  \begin{enumerate}[align=left]
  \item Pick a subblock $S$ in $X/P$, that has at most half of the size of its
    compound block $C\in X/Q$, i.e.~$S\subseteq C$ and $2\cdot |S| \le |C|$.
  \item Split $C$ into $S$ and $C\setminus S$ in $X/Q$.
  \item Call \textsc{Split}$(X/P, S)$ (\autoref{figSplitAlgo}).
  \end{enumerate}
\end{itemize}
\end{algorithm}

\subsection{Correctness}
\label{sec:correct}

\begin{lemma} \label{p1Properties}
Assume that the invariants hold. Then after part (a) of \autoref{figSplitAlgo}, for the given
$S\subseteq C \in \unnicefrac{X}{Q}$ we have:
\begin{enumerate}
    \item \label{p1PropToSub}
        For all $x\in X$: $\op{toSub}(x) = \{ e\in E\mid e =
        x\xrightarrow{a}y, y\in S\}$
    \item \label{p1PropFilS}
        For all $x\in X$: $\op{fil}_S\cdot \flat\cdot \xi(x) =
        \Bagf(\pi_2\cdot \op{graph})(\op{toSub}(x))(a)$.

    \item \label{p1PropBS}
        $\op{M}\colon \unnicefrac{X}{P} \partialto H3$ is a partial map
        defined by
        \[
          M(B) =  H\chi_\emptyset^C\cdot \xi(x),
          \qquad
          \text{if $x\in B$ and there exists an
            $e = x\xrightarrow{a} y$ with $y\in S$},
        \]
        and $M(B)$ is undefined otherwise.   
    \item \label{p1PropMark} For each $B\in\unnicefrac{X}{P}$, we have a partial map
      $\op{mark}_B\colon B\partialto \N$ defined by
      \[
        \op{mark}_B(x) =  \op{lastW}(e) \qquad
        \text{if there exists some $e = x\xrightarrow{a}y$ with $y\in
          S$,}
      \]
      and $\op{mark}_B$ is undefined otherwise.
      
    \item \label{p1PropDeref}
        If defined on $x$, $\op{deref}\cdot\op{mark}_B(x) = w(C,\xi(x))$.

    \item \label{p1PropUnmarked}
        If $\op{mark}_B$ is undefined on $x$, then
        $\op{fil}_S(\flat\cdot \xi(x)) = \emptyset$ and 
        $H\chi^C_S \cdot \xi(x) = H\chi^C_\emptyset \cdot \xi(x)$.

\end{enumerate}
\end{lemma}
\begin{proof}
\begin{enumerate}
\item By lines~\ref{algoLineToSubLoop} and~\ref{algo11},
    \begin{align*}
        \op{toSub}(x) &= \{
        e\in \op{pred}(y) \mid y\in S, e = x\xrightarrow{a}y \}
        \\ &= \{
        e\in E \mid y\in S, e = x\xrightarrow{a}y
        \}.
    \end{align*}

\item \(
\begin{aligned}[t]
    \op{fil}_S(\flat\cdot \xi(x))(a)
    &= \sum_{y\in S} (\flat\cdot\xi(x))(a,y)
    = \sum_{y\in S} |\{e\in E\mid e = x\xrightarrow{a} y\}|
    \\ &
    = |\{e\in E\mid e = x\xrightarrow{a} y, y\in S\}|
    = \Bagf(\pi_2\cdot \op{graph})(\op{toSub}(x)).
\end{aligned}\)

\item First, $\op{M}$ is a partial map since for every block $B$, a
  pair $(B,v)$ is added to $\op{M}$ at most once in line~\ref{algo8} because 
  when any node $x$ from $B$ occurs in line 3 for the first time we
  have that $\op{mark}_B$, $\op{toSub}(x)$ are both empty, and they
  are both nonempty after line~\ref{algo11}.
  By construction $\op{M}$ is defined precisely for those blocks
  $B$ which have at least one element~$x$ with an edge
  $e = x \xrightarrow{a} y$ to $S$. Let
  $C = [y]_{\kappa_Q} \in \unnicefrac{X}{Q}$. Then using
  Invariant~\ref{invariant_deref} in the second step we see that 
  \begin{align*}
    M(B)
    &= \pi_2\cdot \op{update}(\emptyset, \op{deref}\cdot  \op{lastW}(e))
    \\
    &= \pi_2\cdot \op{update}(\emptyset, w(C,\xi(x)))
    \\
    &= \pi_2\cdot \op{update}(\op{fil}_\emptyset(\flat\cdot\xi(x)),
    w(C,\xi(x))) \\
    &\overset{\mathclap{\eqref{eqSplitterLabels}}}{=}
        H\chi_\emptyset^C\cdot\xi(x)
\end{align*}
for the $e=x\xrightarrow{a} y$, $x\in B$, $y\in S$ that occurs first in the loop. Since
$\ker(H\chi_\emptyset^C\cdot\xi)= \ker(H\chi_C\cdot\xi)$,
Invariant~\ref{invariant_Hchi} proves well-definedness.

\item This is precisely how $\op{mark}_B$ has been constructed; that
  it is a partial map follows since for every $x \in X$ a pair $(x,
  n)$ is added to $\op{mark}_B$ at most once in line~\ref{algo10} if
  $\op{toSub}(x)$ is nonempty, and immediately after that
  $\op{toSub}(x)$ becomes nonempty in line~\ref{algo11}.
  Well-definedness follows from Invariant~\ref{invariant_lastW}. Note
  that for every $B$ on which $\op{M}$ is undefined, the list
  $\op{mark}_B$ is empty.

\item If $\op{mark}_B(x) = p_C$ is defined, then $p_C = \op{lastW}(e)$
for some $e\in \op{toSub}(x)$, and so $\op{deref}(p_C) =
\op{deref}\cdot \op{lastW}(e) = w(C,\xi(x))$ by Invariant~\ref{invariant_deref}.

\item If $x\in B$ is not marked in $B$, then $x$ was never mentioned in line 3. Hence, $\op{toSub}(x) =
\emptyset$, and we have
\[
    \op{fil}_S(\flat\cdot \xi(x))(a)
    = |\{e\in E\mid e = x\xrightarrow{a} y, y\in S\}|
    = |\op{toSub}(x)|
    = 0.
\]
Furthermore we have
\begin{align*}
  H\chi_S^C\cdot \xi(x)
  =&\ %
     \pi_2\cdot \op{update}(\op{fil}_S\cdot \flat\cdot\xi(x), w(C,\xi(x)))
  &\text{by \eqref{eqSplitterLabels}}
  \\=&\ %
       \pi_2\cdot \op{update}(\emptyset,  w(C,\xi(x)))
  &\text{as just shown}
  \\=&\ %
       \pi_2\cdot \op{update}(\op{fil}_\emptyset\cdot \flat\cdot\xi(x), w(C,\xi(x)))
  &\text{by definition}
  \\=&\ %
       H\chi_\emptyset^C\cdot \xi(x)
  &\text{by \eqref{eqSplitterLabels}}
  &
    \tag*{\qedhere}
\end{align*}
\end{enumerate}
\end{proof}

\begin{samepage}
\begin{theorem}[Correctness] \label{thmSplitCorrect}
  If the invariants hold before invoking \textsc{Split}, then
  \nopagebreak
  \begin{enumerate}[beginpenalty=9999999]
  \item[(i)]\label{thmSplitCorrect:1} \textsc{Split} returns the correct partitions, that is,
    \textsc{Split}$(\unnicefrac{X}{P}, \unnicefrac{X}{Q} ,S\subseteq
    C\in \unnicefrac{X}{Q})$ refines $\unnicefrac{X}{P}$ by
    $\smash{H\chi_S^C\cdot \xi\colon X\to H3}$, i.e.~all blocks in
    $X/P$ are split by $H\chi_S^C\cdot \xi$ so that $P$ is replaced by 
    $P \cap \ker (H\chi_S^C\cdot \xi)$,
      and
    \item[(ii)] upon termination of \textsc{Split} the invariants hold again.
  \end{enumerate}
\end{theorem}
\end{samepage}
\begin{proof}
  (i)~We show that every block $B\in X/P$ is split by $H\chi_S^C\cdot
  \xi$, by case distinction whether $B\in \op{M}$:
  \begin{itemize}
  \item If $B$ is not in $\op{M}$, then $\op{mark}_B$ is undefined everywhere,
    so for all $x\in B$, we have by \autoref{p1Properties}\ref{p1PropUnmarked}
    that $H\chi_S^C\cdot \xi(x) = H\chi_\emptyset^C\cdot \xi(x)$, which is the
    same for all $x\in B$ by invariant \ref{invariant_Hchi}. Hence every $B$ not
    in $\op{M}$ stays unchanged when splitting by $H\chi_S^C\cdot\xi$.
  \item If $(B,v_\emptyset) \in \op{M}$ (in line \ref{algoSplitLoop}), then
    we show that
    \[
    B_{\neq\emptyset} = \{
    (x,H\chi_S^C\cdot \xi(x)) \mid x\in B, H\chi_S^C\cdot \xi(x) \neq v_\emptyset
    \}.
    \]
    This is sufficient for correctness, since
    $B_{\neq\emptyset}$ is split w.r.t.~the second component (line
    \ref{algoLineSplit}) creating new blocks, and since all elements $x\in B$
    not mentioned in $B_{\neq\emptyset}$ stay in $B$.

    For this characterization of $B_{\neq\emptyset}$, suppose first that $x\in
    B$ is marked, i.e.~we have $p_C = \op{mark}_B(x)$ and lines
    \ref{algoLineMarkedX}--\ref{algoLastLine} are executed. Then we have
    \begin{align*}
          (w^x_S, v^x, w^x_{C\setminus S}) =\ &  \op{update}(\Bagf(\pi_2\cdot \op{graph})(\op{toSub}[x]),
          \op{deref}[p_C])
          && \text{by line \ref{algoLineUpdate}}
          \\ =\ &
          \op{update}(\op{fil}_S\cdot \flat\cdot \xi(x), w(C,\xi(x)))
          && \text{by \autoref{p1Properties} \ref{p1PropFilS},
          \ref{p1PropDeref}}
          \\ =\ &
          (w(S,\xi(x)), H\chi_S^C\cdot \xi(x), w(C\setminus S, \xi(x)))
          && \text{by \eqref{eqSplitterLabels}}.
    \end{align*}
    Thus, if $H\chi_S^C\cdot \xi(x) = v_\emptyset$, $(x,v^x)$ is not added to
    $B_{\neq\emptyset}$ (line~\ref{algoLineIfVXVempty}), and otherwise $x$ is
    correctly removed from $B$ and added to $B_{\neq\emptyset}$ in
    line~\ref{algoLastLine}.

    Now suppose that $x\in B$ is not marked, then $x$ is not added to
    $B_{\neq\emptyset}$ and so we have to show that $v_\emptyset =
    H\chi_S^C\cdot\xi(x)$. We have $v_\emptyset = H\chi_\emptyset^C\cdot\xi(x')$
    for some $x'\in B$ by \autoref{p1Properties}\ref{p1PropBS} and
    $H\chi_S^C\cdot\xi(x) = H\chi_\emptyset^C\cdot\xi(x)$
    by \autoref{p1Properties}\ref{p1PropUnmarked}. Since $\ker(H\chi^C_\emptyset
    \cdot \xi) = \ker(H\chi_C \cdot \xi)$ and $(x,x')\in \ker(H\chi_C\cdot \xi)$
    by invariant \ref{invariant_Hchi}, we have $H\chi_\emptyset^C\cdot
    \xi(x') = H\chi_\emptyset^C\cdot \xi(x)$. Thus,
    \[
      v_\emptyset
      =
      H\chi_\emptyset^C \cdot \xi(x')
      =
      H\chi_\emptyset^C \cdot \xi(x)
      =
      H\chi_S^C \cdot \xi(x).
    \]
  \end{itemize}

  \medskip\noindent
  (ii)~We denote the former values of
  $P,Q, \op{deref}, \op{lastW}$ using the subscript $\old$.
    \begin{enumerate}
    \item It is easy to see that $\op{toSub}(x)$ becomes nonempty in
      line 11 only for marked $x$, and for those $x$ it is emptied
      again in line 20.

    \item Take $e_1 = x_1\xrightarrow{a_1}y_1$,  $e_2 =
    x_2\xrightarrow{a_2}y_2$.
    \begin{itemize}[align=left]
    \item[$\Rightarrow$] Assume $\op{lastW}(e_1) = \op{lastW}(e_2)$.
    If $\op{lastW}(e_1) = p_S$ is assigned in line
    \ref{algoLineSetLastW} for some marked $x$, then
    $\op{lastW}(e_2)$ must be assigned to $p_S$ in the same
    \textbf{for} loop since $p_S$ is the address of a new memory cell
    (line~\ref{algoLineNewCell}). Hence, $e_1, e_2 \in \op{toSub}(x)$,
    which implies that $x_1=x=x_2$ and
    $y_1,y_2\in S\in \unnicefrac{X}{Q}$.
    Otherwise, $\op{lastW}(e_1) =
    \op{lastW}_\old(e_1)$ and so $\op{lastW}_\old(e_1)=
    \op{lastW}_\old(e_2)$ and the desired property follows from the
    invariant for $\op{lastW}_\old$.

  \item[$\Leftarrow$] If $x_1=x_2$ and
    $y_1,y_2\in D\in \unnicefrac{X}{Q}$, where recall that
    $X/Q = X/Q_\old\setminus\{C\}\cup \{S,C\setminus S\}$, then we
    perform a case distinction on $D$. If $D = S$, then
    $e_1,e_2\in \op{toSub}(x)$
    (\autoref{p1Properties}\ref{p1PropToSub}) and so both entries in
    the array are set to the same value
    $\op{lastW}(e_1) = \op{lastW}(e_2) = p_S$ in
    line~\ref{algoLineSetLastW}. If $D \neq S$, then
    $e_1,e_2\not\in \op{toSub}(x)$ and so $\op{lastW}[e_1]$ and
    $\op{lastW}[e_2]$ stay unchanged. Hence, we have
    \[
      \op{lastW}(e_1) = \op{lastW}_\old(e_1) = \op{lastW}_\old(e_2) =
      \op{lastW}(e_2).
    \]

    \end{itemize}

  \item Let $e = x\xrightarrow{a} y$,
    $D := [y]_{\kappa_Q} \in \unnicefrac{X}{Q}$. We again perform a
    case distinction on $D$:
    \[
    \begin{array}{lll}
        D = S
        &\Rightarrow&
            \op{deref} \cdot \op{lastW}(e) = \op{deref}(p_S)
            = w_S^x = w(S,\xi(x)),
        \\
        D = C\setminus S
        &\Rightarrow&
            \op{deref} \cdot \op{lastW}(e) = \op{deref}(p_C)
            = w_{C\setminus S}^x = w(C\setminus S,\xi(x)),
        \\
        D \in \unnicefrac{X}{Q_\old}\setminus \{C\}
        &\Rightarrow&
            \op{deref} \cdot \op{lastW}(e)
            = \op{deref}_\old \cdot \op{lastW}(e)_\old
            = w(D,\xi(x)).
    \end{array}
    \]
    Note that the first equation in the second case holds due to lines~\ref{algo10} and~\ref{algoLineMarkedX}.
    For the first two cases note that
    $w_S^x = w(S,\xi(x))$ and
    $w_{C\setminus S}^x = w(C\setminus S, \xi(x))$ by line
    \ref{algoLineUpdate}, \autoref{p1Properties}, items~\ref{p1PropFilS}
    and~\ref{p1PropDeref}, and by the axiom for $\op{update}$ in
    \eqref{eqSplitterLabels}.

    \item Take $x_1,x_2\in B'\in \unnicefrac{X}{P}$ and $D\in
    \unnicefrac{X}{Q}$ and let $B := [x_1]_{P_\old} = [x_2]_{P_\old}
    \in \unnicefrac{X}{P_\old}$. By case distinction on $\op{M}(B)$ we first
    show that $H\chi_S^C\cdot \xi(x_1) = H\chi_S^C\cdot \xi(x_2)$.
    \begin{enumerate}
    \item If $\op{M}(B)$ is defined, then
        $H\chi_S^C\cdot \xi(x_1) = H\chi_S^C\cdot \xi(x_2)$ --
        for otherwise $x_1$ and $x_2$ would have been put into different blocks in line
        \ref{algoLineSplit}.

      \item If $\op{M}(B)$ is undefined, then $\op{mark}_B$ is undefined
        everywhere, in particular for $x_1$ and $x_2$. Then, by
        \autoref{p1Properties}\ref{p1PropUnmarked} we have
        $H\chi_S^C\cdot \xi(x_i) = H\chi_\emptyset^C\cdot \xi(x_i)$ for $i = 1, 2$.
        Since $C\in \unnicefrac{X}{Q_\old}$ and
        $\ker(H\chi^C_\emptyset \cdot\xi) = \ker(H\chi_C \cdot \xi)$,
        we have
        $H\chi_\emptyset^C\cdot\xi(x_1) =
        H\chi_\emptyset^C\cdot\xi(x_2)$ by invariant
        \ref{invariant_Hchi}, and so
        $(x_1,x_2) \in \ker (H\chi_S^C\cdot \xi)$.
    \end{enumerate}
    We can conclude the invariant by another case distinction on $D$:
    \[
    \begin{array}[b]{lll}
        D = S
        &\Rightarrow&
            (x_1,x_2) \in
            \ker (H\chi_S^C\cdot \xi)
            \subseteq
            \ker (\overbrace{H(=2)\cdot H\chi_S^C}^{H\chi_S=H\chi_D} \cdot \xi),
        \\
        D = C\setminus S
        &\Rightarrow&
            (x_1,x_2) \in
            \ker (H\chi_S^C\cdot \xi)
            \subseteq
            \smash{\ker (\underbrace{H(=1)\cdot H\chi_S^C}_{H\chi_{C\setminus
            S}=H\chi_D} \cdot \xi)},
        \\
        D \in \unnicefrac{X}{Q_\old}\setminus \{C\}
        &\Rightarrow&
            (x_1,x_2) \in \ker(H\chi_D \cdot\xi),
    \end{array}
    \]
    where we use \autoref{rem:kernel}\ref{i:comp} and where the last statement holds
    by invariant \ref{invariant_Hchi} for $\unnicefrac{X}{Q_\old}$.
    \qedhere
    \end{enumerate}
\end{proof}
\begin{corollary}
  \autoref{algPTfinal} computes the simple quotient of a given finite coalgebra
  $\xi\colon X \to HX$. 
\end{corollary}
\begin{proof}
  Indeed, by \autoref{thmSplitCorrect}(i) and \eqref{eq:qi++} we see
  that \eqref{eq:keropti2} is equivalent to
  \eqref{kernelOptimization}:
  \[
    P_{i+1}
    =
    P_i \cap \ker (H\chi_{S_i}^{C_i} \cdot \xi)
    =
    P_i \cap \ker (Hq_{i+1} \cdot \xi).
  \]
  By \autoref{optimizationSummary}, this is equivalent to the original line
  \ref{step:P} in \autoref{catPT},
  since the employed \textqt{select the smaller half} routine respects
  compound blocks and $H$ is zippable. 
  The desired result thus follows from \autoref{thm:correct}.
\end{proof}

\subsection{Efficiency}
\label{sec:time}

Having established correctness of \textsc{Split}, we next
analyse its time complexity. We first analyse lines
$\ref{algoFirstLine}-\ref{algoLastLine}$, then the complexity of the
grouping operation in line \ref{algoLineSplit}, and finally the
overall complexity of the algorithm, accumulating the time for all
\textsc{Split} invocations.
\begin{lemma}\label{lemmaTimeSplit}
  Lines $\ref{algoFirstLine}-\ref{algoLastLine}$ in \textsc{Split} run
  in time $\smash{\CO(\sum_{y\in S}|\op{pred}(y)|)}$.
\end{lemma}
\begin{proof}
    The loop in line \ref{algoLineToSubLoop} has $\sum_{y\in S}
    |\op{pred}(y)|$ iterations, each consisting of constantly many
    operations taking constant time.
    Since each loop appends one element to some initially empty
    $\op{toSub}(x)$, we have
    \[
        \sum_{y\in S}|\op{pred}(y)| = \sum_{x\in X} |\op{toSub}(x)|.
    \]
    In the body of the loop starting in line \ref{algoLineMarkedX}, the
    only statements not running in constant time are
    $\ell \gets \Bagf(\pi_2\cdot \op{graph})(\op{toSub}(x))$ (line
    \ref{algoLineLabels}), $\op{update}(\ell,\op{deref}(p_C))$ (line
    \ref{algoLineUpdate}), and the loop in line
    \ref{algoLineSetLastW}; each of these require time linear in the
    length of $\op{toSub}(x)$. The loop in line \ref{algoLineMarkedX}
    has at most one iteration per $x\in X$.  Hence, since each $x$ is
    contained in at most one block $B$ from line~\ref{algoSplitLoop},
    the overall complexity of line \ref{algoSplitLoop} to
    \ref{algoLastLine} is at most
    \( \sum_{x\in X} |\op{toSub}(x)| = \sum_{y\in S}|\op{pred}(y)| \),
    as desired.
\end{proof}

\noindent In the grouping operation in line \ref{algoLineSplit}, it is not enough to group
the elements using a sorting algorithm. Instead we need to preprocess the
elements and extract a possible majority candidate.
\begin{definition}
When grouping $Z$ by $f\colon Z\to Z'$ we call an element
$p\in Z'$ a \emph{possible majority candidate} (PMC) if either
\begin{equation}
|\{z\in Z\mid f(z) = p\}|
\ge |\{z\in Z\mid f(z) \neq p\}|
\label{pmcproperty}
\end{equation}
or if no element in $Z'$ fulfilling \eqref{pmcproperty} exists.
\end{definition}
\noindent A PMC can be computed in linear time \cite[Sect.~4.3.3]{Backhouse1986}. When
grouping $Z$ by $f$ using a PMC, one first determines a PMC $p\in
Z'$, and then one only sorts and groups $\{z\mid f(z) \neq
p\}$ by $f$ using an $\CO(n\cdot \log n)$ sorting algorithm.
\begin{remark}\label{R:ith}
  In the following lemmas we again index data by the iterations $1
  \leq i \leq k$ in \autoref{algPTfinal}. That
  means we consider $S_i \subseteq C_i \in X/Q_{i+1}$ such that
  \begin{equation}\label{eq:XQi++}
    X/Q_{i+1} = X/Q_i \setminus \{C_i\} \cup \{S_i, C_i \setminus
    S_i\}. 
  \end{equation}
\end{remark}
\begin{lemma} \label{SplitComplexity} Summing over all \textsc{Split}
  invocations in \autoref{algPTfinal}, the
  total time spent on grouping $B_{\neq\emptyset}$ using a PMC is in
  $\CO(|E|\cdot \log|X|)$.
\end{lemma}
\noindent The proof is a generalization of that for the weighted setting of
Valmari and Franceschinis~\cite[Lemma~5]{ValmariF10}:
\begin{proof}
  We shall prove that for $S_i \subseteq C_i\in \unnicefrac{X}{Q_i}$, $0\le i
  < k$, in \autoref{R:ith}, the overall time spent on grouping the
  $B_{\neq\emptyset}$ in all the runs of $\textsc{Split}$ is in $\CO(|E|\cdot
  \log|X|)$. The partition returned by the algorithm is $X/P_k = X/Q_k$ and is
  obtained after \textsc{Split} has been called $k$-many times.

  In the first part of the proof, we define subsets
  $\Leftblock_B, \Middleblock_B \subseteq B_{\neq\emptyset}$ which are
  generalized version of the left-hand and middle subblocks of a block
  $B$ in \autoref{fig:RefineStep} for a $\Potf$-coalgebra, i.e.~those
  state with successors in $S$ but not in $C\setminus S$ and those with
  successors in both. We then give an equivalent characterization of
  $\Leftblock_B$ and $\Middleblock_B$.

  In the second part, we use this characterization and a PMC to argue
  that sorting each $B_{\neq\emptyset}$ is bounded by
  $2\cdot |\Middleblock_B|\cdot \log(2\cdot |\Middleblock_B|)$. Since
  we assume that comparing two elements of $H3$ runs in constant time,
  the time needed for sorting amounts to the number of comparisons
  needed while sorting, i.e.~$\CO(n\cdot \log n)$ many.

  Finally, in the third part, we use this to obtain the desired overall complexity.
%
\begin{enumerate}
\item For a $(B,v_\emptyset)\in \op{M}$ in the iteration $i$ consider
  $B_{\neq\emptyset}$. We define
    \begin{itemize}
    \item the \emph{left block} $\Leftblock_B^i := \{
        x \in B \mid
        H\chi_{C_i}^{C_i}\cdot \xi(x) = H\chi_{S_i}^{C_i}\cdot \xi(x) \neq
        H\chi_{\emptyset}^{C_i}\cdot \xi(x)
    \}$, and

    \item the \emph{middle block} $\Middleblock_B^i := \{
        x \in B \mid H\chi_{C_i}^{C_i}\cdot \xi(x) \neq
        H\chi_{S_i}^{C_i}\cdot \xi(x)
        \neq H\chi_{\emptyset}^{C_i}\cdot \xi(x)
        \}$.
    \end{itemize}
    Now, first note that
    \begin{equation}\label{eq:42}
      (x,v^x)\in B_{\neq\emptyset}
      \ \text{iff}\
      v^x = H\chi_{S_i}^{C_i}\cdot\xi(x) \neq v_\emptyset,
      \qquad\text{and}\qquad
      v_\emptyset = H\chi_\emptyset^C\cdot \xi(x),
    \end{equation}
    where the latter holds by
    \autoref{p1Properties}\ref{p1PropBS}. This implies that 
    the set of first components of the pairs in $B_{\neq\emptyset}$ is
    $\Leftblock_B^i \cup\Middleblock_B^i$. If $x\in B$ has no edge to
    ${S_i}$, then it is not marked, and so
    $H\chi_{S_i}^{C_i}\cdot\xi(x) = H\chi_\emptyset^{C_i}\cdot \xi(x)$, by
    \autoref{p1Properties}\ref{p1PropUnmarked}; by contraposition,
    every $x\in \Leftblock_B^i\cup\Middleblock_B^i$ has some edge into
    ${S_i}$. We can make a similar observation for $C_i\setminus
    S_i$. If $x\in B$ has no edge to ${C_i}\setminus {S_i}$, then
    $\op{fil}_{S_i}(\flat\cdot\xi(x)) =
    \op{fil}_{C_i}(\flat\cdot\xi(x))$
    by the definition of
    $\op{fil}_{S_i}$, and therefore we have:
    \[
      \begin{array}{r@{}c@{}l@{}}
        H\chi_{C_i}^{C_i}\cdot\xi(x)
        &\stackrel{\eqref{eqSplitterLabels}}{=}& \pi_2\cdot\op{update}(\op{fil}_{C_i}(\flat\cdot\xi(x)),
        w({C_i}, \xi(x)))
        \\
        &=& \pi_2\cdot\op{update}(\op{fil}_{S_i}(\flat\cdot\xi(x)),
        w({C_i}, \xi(x)))
        \stackrel{\eqref{eqSplitterLabels}}{=}
        H\chi_{S_i}^{C_i}\cdot\xi(x).
      \end{array} \hspace{-5mm}
    \]
    By contraposition, all $x\in \Middleblock_B^i$ have an edge to
    ${C_i}\setminus {S_i}$.

    Note that
    $\ker(H\chi_{C_i}^{C_i}\cdot \xi) = \ker(H\chi_{C_i}\cdot
    \xi)$; indeed, to see this use \autoref{rem:kernel}\ref{i:mono},
    that $H$ preserves monomorphisms, and that $\chi_{C_i}^{C_i} = m \cdot
    \chi_{C_i}$, where $m\colon 2 \cong\{0,2\} \hookrightarrow
    \{0,1,2\} = 3$ is the inclusion map. 
    Since $\Leftblock_B^i \subseteq B\in \unnicefrac{X}{P_i}$ and
    $C_i\in\unnicefrac{X}{Q_i}$, we conclude from invariant~\ref{invariant_Hchi}
    that there is an $\ell^i_B \in H3$ such that
    $\ell^i_B = H\chi_{C_i}^{C_i}\cdot \xi(x)$ for all
    $x\in B$. Using~\eqref{eq:42} we therefore obtain 
    \[
      \Leftblock_B^i = \{x\mid (x,v^x)\in B_{\neq\emptyset}, v^x =
      \ell^i_B \}\quad\text{and}\quad
      \Middleblock_B^i = \{x\mid (x,v^x)\in B_{\neq\emptyset},
      v^x \neq \ell^i_B \}.
    \]
    We have also seen that every $x \in \Leftblock_B^i$ has an edge to
    $S_i$ and every $x \in \Middleblock_B^i$ has both an edge to $S_i$
    and $C_i \setminus S_i$.

  \item We prove that sorting
    $B_{\neq\emptyset}$ in the iteration $i$ is bound by
    $2\cdot |\Middleblock_B^i|\cdot \log(2\cdot |\Middleblock_B^i|)$
    by case distinction on the possible majority candidate:
    \begin{itemize}
    \item If $\ell_B^i$ is the possible majority
    candidate, then the sorting of $B_{\neq\emptyset}$ sorts
    precisely $\Middleblock_B^i$ which indeed amounts
    to
    \[
        |\Middleblock_B^i|\cdot \log(|\Middleblock_B^i|)
        \le 2\cdot |\Middleblock_B^i|\cdot \log(2\cdot
        |\Middleblock_B^i|).
    \]
    \item If $\ell_B^i$ is not the possible
    majority candidate, then $|\Leftblock_B^i| \le
    |\Middleblock_B^i|$.
    In this case sorting $B_{\neq\emptyset}$ is bounded by
    \[
        (|\Leftblock_B^i|+|\Middleblock_B^i|)\cdot
        \log(|\Leftblock_B^i|+|\Middleblock_B^i|)
        \le 2\cdot |\Middleblock_B^i|\cdot \log(2\cdot
        |\Middleblock_B^i|).
    \]
  \end{itemize}
    
\item Finally, we show that we obtain the desired overall time
  complexity. Let the number of blocks to which $x$ has an edge be
  denoted by
    \begin{align*}
        \sharp^i_Q(x)
        &=
        |\{
            D\in \unnicefrac{X}{Q_i}\mid e = x\xrightarrow{a}{y}, y\in
            D
        \}|, \quad i\ge 0.
    \end{align*}
    Clearly, this number is bounded by the number of outgoing edges of
    $x$, i.e.~$\sharp^i_Q(x) \le |\flat\cdot\xi(x)|$, and so
    \[
        \sum_{x\in X} \sharp_Q^i(x) \le
        \sum_{x\in X} |\flat\cdot\xi(x)| = |E|.
    \]
    Define
    \begin{align*}
        \sharp_\Middleblock^i(x)
        &= |\{
            0 \le j < i \mid
            x\text{ is in some }\Middleblock_B^j
        \}|, \quad i\ge 0.
    \end{align*}
    If in the iteration $i\ge 0$, $x$ is in a middle block $\Middleblock_B^i$,
    then $\sharp_\Middleblock^{i+1}(x) = \sharp_\Middleblock^{i}(x) + 1$ and
    $\sharp_Q^{i+1}(x) = \sharp_Q^{i}(x) + 1$, by~\eqref{eq:XQi++} and since $x$
    has both an edge to $S_i$ and $C_i\setminus S_i$. Otherwise, if $x$ is not
    in any middle block in iteration $i$, then $\sharp_\Middleblock^{i+1}(x) =
    \sharp_\Middleblock^{i}(x)$ and $\sharp_Q^{i+1}(x) \ge \sharp_Q^{i}(x)$.
    This implies that for all $i\ge 0$, $\sharp_\Middleblock^i(x) \le
    \sharp_Q^i(x)$, and therefore
    \[
        \sum_{x\in X} \sharp_\Middleblock^i(x)
        \le \sum_{x\in X} \sharp_Q^i(x)
        \le |E|.
    \]
    Let $T$ denote the total number of middle blocks
    $\Middleblock_B^i$, $0\le i < k$, such that $B$ is contained in
    $\op{M}$ in iteration $i$,
    and let $\Middleblock_t$, $1\le t \le T$, be the $t^\text{th}$
    middle block. The sum of the sizes of all middle blocks is the
    same as the number of times each $x\in X$ was contained in a
    middle block, i.e.
    \[
        \sum_{t=1}^T |\Middleblock_t|
        = \sum_{x\in X} \sharp_\Middleblock^k(x)
        \le |E|.
    \]
    Using the previous bounds and $|\Middleblock_t| \le |X|$, we now obtain
    \begin{align*}
        &\phantom{\le\ }\sum_{t=1}^T
        2 \cdot |\Middleblock_t|\cdot \log (2\cdot|\Middleblock_t|)
        \le
        \sum_{t=1}^T
        2 \cdot |\Middleblock_t|\cdot \log (2\cdot |X|)
        = 2 \cdot \big(\sum_{t=1}^T |\Middleblock_t|\big) \cdot \log(2\cdot
        |X|)
        \\
        & \le 2\cdot |E|\cdot \log(2\cdot |X|)
        = 2 \cdot |E|\cdot \log(|X|) + 2 \cdot |E|\cdot \log(2)
        \in \CO\big(|E|\cdot \log(|X|)\big).
        \tag*{\qedhere}
    \end{align*}
\end{enumerate}
\end{proof}
\begin{lemma}\label{lemmaTime}
    \begin{enumerate}[topsep=1pt]
    \item For each $y\in X$, $|\{i < k \mid y\in S_i\}| \le \log_2 |X| + 1$.
    \item\label{lemmaTime2} The total run-time of all invocations of 
      $\textsc{Split}(X/P_i, S_i)$, $0\le i < k$, is in $\CO(|E|\cdot \log |X|)$.
    \end{enumerate}
\end{lemma}
\begin{proof}
  \begin{enumerate}
  \item We know from \autoref{PfinerthanQ} that $Q_{j}$ is finer than
    $Q_i$ for every $j > i$. Moreover, by~\eqref{eq:XQi++}, we have
    $S_i\in\unnicefrac{X}{Q_{i+1}}$. For every $i < j$ with $y\in S_i$
    and $y\in S_j$, we know that $C_j \subseteq S_i$ since $C_j$ is
    the block containing $y$ in the refinement $\unnicefrac{X}{Q_j}$
    of $\unnicefrac{X}{Q_{i+1}}$ in which $S_i$ contains $y$. Hence,
    we have $2\cdot|S_j| \le |C_j|\le |S_i|$. Now let
    $i_1 < \ldots < i_n$ be all the elements in
    $\{i < k \mid y \in S_i\}$. Since
    $y\in S_{i_1},\ldots, y\in S_{i_n}$, we have
    $2^{n-1} \cdot |S_{i_n}| \le |S_{i_1}|$.  Thus
  \begin{align*}
    |\{i < k \mid y \in S_i\}| = n &= \log_2(2^{n-1})+1 \\ &
    \le \log_2(2^{n-1}\cdot |S_{i_n}|) + 1
    \le \log_2|S_{i_1}| + 1\le \log_2 |X| + 1,
  \end{align*}
  where the last inequality holds since $S_{i_1} \subseteq X$.

\item By \autoref{lemmaTimeSplit}, the first term below gives the total
  run-time in the $\CO$-calculus, and we continue to reason in that calculus 
  (note that the inner sums on the right-hand side of the first line
  are indexed by $S_i$, whence the $\ni$-symbol):
\begin{align*}
    \sum_{0\le i < k}
    \sum_{y\in S_i} |\op{pred}(y)|
    &=
    \sum_{y\in X}
    \sum_{\substack{0\le i < k \\  S_i \ni y}}
    |\op{pred}(y)|
    =
    \sum_{y\in X}
    \big(
    |\op{pred}(y)|
    \cdot
    \sum_{\substack{0\le i < k \\ S_i \ni y}} 1
    \big)
    \\ &
    = \sum_{y\in X}
    \big(
    |\op{pred}(y)|
    \cdot |\{i < k \mid y \in S_i\}|
    \\ &
    \le
    \sum_{y\in X}
    \big(
    |\op{pred}(y)|
    \cdot
    (\log|X|+1)
    \big)
    \\ &
    =
    \big(
    \sum_{y\in X}
    |\op{pred}(y)|
    \big)
    \cdot
    (\log|X|+1)
    \\ &
    =
    |E|\cdot (\log|X| + 1),
\end{align*}
where the inequality holds by the first part of our lemma.\qedhere
\end{enumerate}
\end{proof}

\noindent
By Lemmas~\ref{lemmaInit} and~\ref{lemmaTime}\ref{lemmaTime2}, we
obtain our main result:

\begin{theorem} \label{thmAlgoComplexity}
  Given a zippable functor $H\colon \Set\to\Set$ with a refinement interface,
  whose $\op{update}$ and $\op{init}$ functions can be computed in linear time,
  \autoref{algPTfinal} computes the quotient modulo
  behavioural equivalence of a given coalgebra $\xi\colon X\to HX$ with
  $n=|X|$ states and $m=\sum_{x\in X}|\flat\cdot\xi(x)|$ edges in time
  $\CO((m + n)\cdot \log n )$.
\end{theorem} 
\noindent If the coalgebra is not too sparse, i.e.~every state has at
least one in- or outgoing edge, $2\cdot m \ge n$, then the complexity
is $\CO(m\cdot \log n)$, the bound typically seen in the literature
for efficient algorithms for bisimilarity minimization of transition
systems $\xi\colon X\to \Potf X$ or weighted systems
$\xi\colon X\to \R^{(X)}$. Unlike those algorithms we do not directly
admit an initial partition as a parameter; but switching from a
functor $G$ to $\unnicefrac{X}{\mathcal{I}}\times G$
(cf.~\autoref{coalgebraInitialPartition}) we can equip the generic
algorithm with this additional parameter while maintaining the same
$\CO(m\cdot \log n)$ complexity:

\begin{remark} \label{remInterfaceInitial}
There are two ways to handle functors of type $H=\unnicefrac{X}{\mathcal{I}}\times
G$.
First, we can modify the functor interface as follows:
\[
  \begin{array}{lll}
    G1 &\mapsto& \unnicefrac{X}{\mathcal{I}}\times G1 \\
    W &\mapsto& \unnicefrac{X}{\mathcal{I}}\times W \\
  \end{array}\qquad\qquad \begin{array}{lll}
    \flat &\mapsto& \unnicefrac{X}{\mathcal{I}}\times GX \xrightarrow{\pi_2} GX
                    \xrightarrow{\flat} \Bagf(A\times X) \\
    \op{init} &\mapsto& \id_\unnicefrac{X}{\mathcal{I}}\times \op{init} \\
  \end{array}
  \]
  $\op{update}$ is replaced by the following
  \[
    \begin{mytikzcd}[column sep=9mm]
         (\Bagf(A) \!\times\! W) \!\times\! \unnicefrac{X}{\mathcal{I}} 
         \arrow{r}{\op{update}\times \unnicefrac{X}{\mathcal{I}}}
         &[1mm] (W\!\times\! G3\!\times\! W)\times \unnicefrac{X}{\mathcal{I}}
         \arrow{d}
           \\
         \Bagf(A) \!\times\! (\unnicefrac{X}{\mathcal{I}}\!\times\! W)
         \arrow{u}{\cong}
         & (\unnicefrac{X}{\mathcal{I}}\!\times\! W)\times
           (\unnicefrac{X}{\mathcal{I}}\!\times\! G3)\times (\unnicefrac{X}{\mathcal{I}}\!\times\! W)
      \end{mytikzcd}
\]
where the first and the last morphism are the obvious ones. A second approach is
to decompose the functor into $\unnicefrac{X}{\mathcal{I}}\times(-)$
and $G$, moving to the multisorted setting, see
\autoref{sec:multisorted} for more details, in particular \autoref{ex:comp}. Both
methods have no effect on the complexity.
\end{remark}
\begin{example}\label{ex:singlesort}
  As instances of our algorithm, we obtain the following standard examples for partition refinement algorithms:
  \begin{enumerate}
  \item For $H=\unnicefrac{X}{\mathcal{I}}×\Potf$, we obtain the
    classical Paige-Tarjan algorithm~\cite{PaigeTarjan87} (with
    initial partition $\unnicefrac{X}{\mathcal{I}}$), with the same
    complexity $\CO((m+n)\cdot \log n)$.

    Recall from \autoref{exCoalgEquivalence} that the equivalence
    computed by the Paige-Tarjan algorithm -- bisimilarity -- is
    precisely behavioural equivalence for the powerset functor
    $\Potf$, also when an initial partition is taken into account
    (\autoref{coalgebraInitialPartition}). $\Potf$ and
    $\unnicefrac{X}{\mathcal{I}}×\Potf$ are zippable functors
    (\autoref{exMonoidZippable}, \autoref{lem:closure}), and we have
    refinement interfaces for them (\autoref{examplePowerset},
    \autoref{remInterfaceInitial}) that realize the desired time
    complexity bound (\autoref{exampleTime}). So the instantiation of
    \autoref{algPTfinal} computes the bisimilarity relation on an
    input coalgebra with~$n$ states and~$m$ edges in time
    $\CO((m+n)\cdot \log n)$ by \autoref{thmAlgoComplexity}.

    The instantiation of \autoref{algPTfinal} for the refinement interface of
    $\Potf$ is nearly identical to the Paige-Tarjan
    algorithm~\cite{PaigeTarjan87}, informally described in
    \autoref{exPaigeTarjanSplit}. For example, both maintain references from the
    edges to the integer counters. However, the \mbox{(three way-)}split of a block
    is simpler in the specific implementation than in \textsc{Split} (Line
    \ref{algoLineSplit}), because $|\Potf 3| = 8$ is finite and in fact in Line
    \ref{algoLineSplit} at most three different values can occur (namely the
    three cases from \autoref{exPaigeTarjanSplit}).

    For an example run, see \autoref{exPartRef} and \autoref{fig:exPartRef:LTS}.

  \item For $HX= \unnicefrac{X}{\mathcal{I}}×\R^{(X)}$, we solve
    Markov chain lumping with an initial partition
    $\unnicefrac{X}{\mathcal{I}}$ in time $\CO((m+n)\cdot \log n)$,
    like the best known algorithm by Valmari and
    Franceschinis~\cite{ValmariF10}:

    Coalgebraic behavioural equivalence for $\R^{(-)}$ captures precisely
    weighted bisimilarity (\autoref{exCoalgEquivalence}), and $\R^{(-)}$ is
    zippable (\autoref{exMonoidZippable}) and has a refinement interface
    (\autoref{exampleH2}) with the required complexity bounds
    (\autoref{exampleTime}). So the instantiation of \autoref{algPTfinal}
    computes weighted bisimilarity for an input coalgebra with $n$ states and
    $m$ edges in $\CO((m+n)\cdot \log n)$ by \autoref{thmAlgoComplexity}.

    The algorithm by Valmari and Franceschinis~\cite{ValmariF10} is essentially
    that of \textsc{Split} (\autoref{figSplitAlgo}), after making
    simplifications using properties of $\R^{(-)}$. For instance, Valmari and
    Franceschinis do not need to keep the accumulated weights from states to
    blocks in memory since the weights can be computed from the labels
    within the first loop of \textsc{Split}.

    For an example run, see \autoref{exPartRef} and \autoref{fig:exPartRef:Markov}.

  \item The functor $H=\Bagf$ has a refinement interface
    (\autoref{exampleH2}) with the desired run-time
    (\autoref{exampleTime}), and thus \autoref{algPTfinal} computes
    behavioural equivalence on a $\Bagf$-coalgebra with~$n$ edges
    and~$m$ states in time $\CO((m+n)\cdot \log n)$.

    Partition refinement for \textqt{undirected} $\Bagf$-coalgebras is known as
    \emph{colour refinement} and called the \emph{1-dimensional
      Weisfeiler-Lehman Algorithm} (WL), where undirected means that
    $x\xrightarrow{n} y$ iff $y\xrightarrow{n} x$. Colour refinement is an
    important subroutine in graph isomorphism checking (and was originally
    conjectured to check graph isomorphism, see
    e.g.~\cite{CaiEA1992,Weisfeiler1976,BerkholzBG17}). Its input is an
    undirected graph $(V,E)$, i.e.\ $E$ is a set of two-element subsets of $V$.
    Then color refinement is just partition refinement on the $\Bag$-coalgebra
    $\xi \colon V\to \Bag V$ defined by $\xi(u)(v) = 1$ if $\{u,v\} \in
    E$ and $\xi(u)(v) = 0$ otherwise (i.e.~one introduces two directed edges
    per undirected edge as usual). Our algorithm runs in $\CO((m+n)\cdot \log n)$,
    matching the run-time of the optimal algorithm by Berkholz, Bonsma, and
    Grohe~\cite{BerkholzBG17}, and improving the run-time of $\CO(m\cdot n)$ of
    a previous algorithm~\cite{ShervashidzeEA2011}.

  \item Hopcroft's classical automata minimization~\cite{Hopcroft71}
    is obtained by $HX = 2×X^A$, with running time $\CO(n\cdot\log n)$
    for the binary input alphabet $A = \{0,1\}$ (Gries~\cite{Gries1973} and
    Knuutila~\cite{Knuutila2001} present this algorithm for arbitrary
    finite input alphabets $A$ which are not fixed but part of
    the input of minimization).
    
    For a fixed finite alphabet $A$, the functor $HX=2\times X^A$ --
    and also any other polynomial functor with bounded arity -- is
    zippable (\autoref{lem:closure}) and has a refinement interface
    (\autoref{exampleH2}) fulfilling the complexity bound
    (\autoref{exampleTime}). So \autoref{algPTfinal} computes
    behavioural equivalence for an input coalgebra with $n$ states and
    $m$ edges in time $\CO((m+n)\cdot \log n)$ by
    \autoref{thmAlgoComplexity}.  The functor encoding for polynomial
    functors (\autoref{exFunctorEncoding}) encodes one occurrence of a
    $k$-ary operation symbol by $k$ edges. Since we assume the
    signature to have bounded arity, the maximal $k$ is independent
    of the coalgebra size and thus a constant factor, and so
    $m \le \max(k)\cdot n \in \CO(n)$.  So the run-time of
    \autoref{algPTfinal} simplifies to $\CO(n\cdot \log n)$.
    
    Note that the mentioned automata minimization algorithms
    \cite{Hopcroft71} are different from \autoref{algPTfinal} because
    they perform a refinement step for each input symbol $a\in A$,
    whereas \textsc{Split} refines a block w.r.t.~all labels, since it
    considers all edges into the subblock $S\subseteq C$ of interest.

    If $A$ is variable and part of the input, we can fit the minimization in the
    present generic framework via the following \autoref{sec:multisorted} on
    composite functors and multisorted sets (see~\autoref{ex:comp} below).

  \end{enumerate}
\end{example}

\section{Modularity via Multisorted Coalgebra}
\label{sec:multisorted}
\tikzsetfigurename{partref-multisorted-}
We next describe how to minimize systems that mix different transition
types. For example, recall from~\autoref{ex:coalg}\ref{i:segala}
that Segala systems mix non-deterministic and probabilistic branching
in a way that makes them coalgebras for the composite functor
$X\mapsto \Potf(\Dist(A \times X))$ (or
$X \mapsto \Potf(A \times \Dist X)$ in the case of simple Segala
systems, respectively).  For our purposes, such functors raise the
problem that zippable functors are not closed under composition. In
the following, we show how to deal with this issue by moving from
composite functors to multisorted coalgebras in the spirit of previous
work on modularity in coalgebraic
logic~\cite{SchroderPattinson11}. Subsequently, the arising
multisorted coalgebras are transformed back to singlesorted coalgebras
by coproduct formation.

\subsection{Explicit intermediate states via multisortedness}
\label{sec:intermediate}
\noindent The transformation from coalgebras for composite functors
into multisorted coalgebras is best understood by example:
\begin{example} \label{exVisualMultisort}
  The functor $TX = \Potf(\Bagf X \times \Dist(A\times X))$ can be visualized as
  
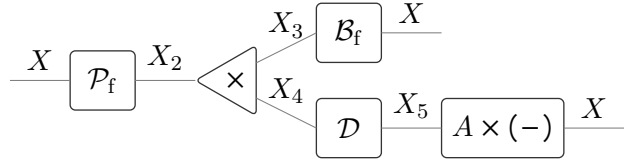
\begin{center}
    \begin{tikzpicture}
      \node[unary] (Potf) {\Potf};
      \node[binary] (times) at (Potf.in) {\times};
      \node[unary] (Bagf) at ([yshift=4mm]times.in1) {\Bagf};
      \node[unary] (Dist) at ([yshift=-4mm]times.in2) {\Dist};
      \node[unary] (Atimes) at (Dist.in) {A\times (-)};
      \path[draw=black!50]
        (times.out) edge node[above] {$X_2$} (Potf.in)
        (Bagf.out) edge node[above] {$X_3$} (times.in1)
        (Dist.out) edge node[above] {$X_4$} (times.in2)
        (Atimes.out) edge node[above] {$X_5$} (Dist.in)
        (Potf.out) edge node[above] {$X$} +(-8mm,0)
        (Bagf.in) edge node[above] {$X$} +(8mm,0)
        (Atimes.in) edge node[above] {$X$} +(8mm,0);
    \end{tikzpicture}
  \end{center}
  where we label the inner connections with fresh names $X_2$, $X_3$,
  $X_4$, $X_5$. From this visualization, we derive a functor
  $\bar T\colon \Set^5 \to \Set^5$:
  \[
    \begin{array}{r@{}l@{\ }l@{\ }l@{\ }l@{\ }l}
    \bar T\,(&X,&X_2,&X_3,&X_4,&X_5\hspace{1.6em})= \\
    (&\Potf X_2, &X_3\times X_4, &\Bagf X, &\Dist X_5, &A\times X).
    \end{array}
  \]
\end{example}
\noindent Formal definitions following~\cite{SchroderPattinson11} are
as follows.
\begin{definition} \label{defFlattening} Given a set $\mathcal{H}$ of
  mono-preserving and finitary functors $H\colon \Set^k\to \Set$ (with
  possibly different arities $k < \omega$), let $T\colon \Set\to \Set$
  be a functor generated by the grammar
  \[
    T ::= (-) \mid H(T,\ldots,T)
  \]
  where $H$ ranges over $\mathcal{H}$ and $(-)$ is the argument, i.e.~$T = (-)$
  is the functor $TX = X$. By inspecting the structure of such a term $T$, we
  can define a functor $\bar T\colon \Set^n \to \Set^n$, where~$n$ is the number
  of non-leaf subterms of~$T$ (i.e.~subterms of~$T$ including $T$ itself but
  not~$(-)$). Let~$f$ be a bijection from non-leaf subterms of $T$ to natural
  numbers $\{1,\ldots,n\}$, with $f(T) = 1$, and write $f(-) = 1$ to simplify
  notation (still, $f^{-1}(1) = T$). The \emph{flattening} of $T$ is $\bar
  T\colon \Set^n\to\Set^n$, given by
  \[
    (\bar T(X_1,\ldots,X_n))_i = H(X_{f(G_1)},\ldots,X_{f(G_k)})
    \quad\text{where }f^{-1}(i) = H(G_1,\ldots,G_k).
  \]
\end{definition}
\noindent Intuitively speaking, we introduce a sort for each wire in
the visualization but identify the outermost wires (labelled~$X$ in
Example~\ref{exVisualMultisort}). Note that we keep track of duplicates, so
e.g.~$TX=\Potf X\times \Potf X$ has three subterms: $(-)\times (-)$, the left
hand and the right hand $\Potf$, so $n=3$ and $\bar T(X_1,X_2,X_3) = (X_2\times
X_3, \Potf X_1, \Potf X_1)$.

In the remainder of this section we write $T$ for functors defined
according to the grammar in \autoref{defFlattening} but continue to
write~$H$ for functors in general (in particular for elements
of~$\mathcal{H}$).
  
\begin{example} \label{exMultisortification}
  In \autoref{exVisualMultisort} the functor
  $T$ is built from the set $\mathcal H$ of functors containing
  \[
    \Potf, \Bagf, \Dist, A\times(-)\colon \Set\to \Set
    \qquad
    \times\colon \Set^2\to \Set,
  \]
  and the term $T= \Potf(\Bagf(-) \times \Dist(A\times (-)))$
  has the
  following non-leaf subterms, implicitly defining the bijection $f$:
  \[
    \begin{array}{lll}
    1.~\Potf(\Bagf(-)\times \Dist(A\times(-))) 
      &
    2.~\Bagf(-)\times \Dist(A\times (-)) 
        &
    3.~\Bagf(-)
          \\
    4.~\Dist(A\times (-))
      &
    5.~A\times (-).
    \end{array}
  \]
  Then, $\bar T\colon \Set^5\to \Set^5$ is defined by
  \begin{equation*}
    \bar T(X_1,X_2,X_3,X_4,X_5)
     = (\Potf X_2,X_3\times X_4,\Bagf X_1,\Dist X_5,A\times X_1).
  \end{equation*}
\end{example}\bigskip

Now a coalgebra for the flattening $\bar T$ of a functor term $T$ is a family of maps
\[
  \xi_{i}\colon X_i \to H(X_{f(G_1)},\ldots,X_{f(G_k)})
  \quad\text{where }f^{-1}(i) = H(G_1,\ldots,G_k).
\]
For example, given a $T$-coalgebra $\xi\colon X\to TX$, the morphism
$(\xi,\id,\ldots,\id)$ in $\Set^n$ is a coalgebra for the flattening
$\bar T\colon \Set^n\to\Set^n$ of~$T$. Note that this defines the
first sort to be~$X$ and implicitly defines the other sorts. This
mapping defines a functor
$\op{Pad}\colon \Coalg(T)\to \Coalg(\bar T)$, which is a fully
faithful right-adjoint \cite{SchroderPattinson11}.
\begin{equation*}
  \begin{mytikzcd}
  \Coalg(T)
  \arrow[bend left=15, shift left = 0]{r}[alias=Pad]{\op{Pad}}
  &
  \Coalg(\bar T)
  \arrow[bend left=15, shift left = 0]{l}[alias=Comp]{\op{Comp}}
  \arrow[draw=none,to=Pad,from=Comp]{}[,description,sloped]{\dashv}
  \end{mytikzcd}
\end{equation*}
The left adjoint $\op{Comp}$ composes the component maps of a
multisorted coalgebra in a suitable way as we now explain. For a given
$\bar T$-coalgebra $(\bar X, \bar \xi)$ with $\bar \xi = (\xi_1,
\ldots, \xi_n)$ we first define for every subterm $T'$ of $T$ a map
\[
  c_{T'}\colon X_{f(T')}\to T'X_1
\]
by induction as follows: for $T' = (-)$ put
$c_{(-)} = \id_{X_1}\colon X_{f(-)} = X_1 \to X_1$, and for $T' =
H(G_1, \ldots, G_k)$, $c_{T'}$ is the following map
\[
  \begin{tikzcd}
    \llap{$X_{f(T')} =\,$}X_{f(H(G_1,\ldots,G_k))}
    \arrow{d}{\xi_{f(H(G_1,\ldots,G_k))}}
    \\
    H(X_{f(G_1)},\ldots,X_{f(G_k)})
    \arrow{d}{H(c_{G_1},\ldots,c_{G_k})}
    \\
    H(G_1X_1,\ldots,G_kX_1)\rlap{$\,= T'X_1$.}
  \end{tikzcd}
\]
Then we obtain the $T$-coalgebra
\[
  \op{Comp}(\bar X, \bar \xi) = \big(X_1 \xrightarrow{c_T} TX_{f(T)} = TX_1\big),
\]
and on morphisms put $\op{Comp}(h_1,\ldots,h_n) = h_1$.
It is not difficult to prove that
\[
  \Comp(\Pad(X,\xi))=(X,\xi),
\]
see \cite[Text after Lemma~4.11]{SchroderPattinson11}.

\begin{example}
  Given a coalgebra
  $\xi\colon X\to \Potf(\Bagf X \times \Dist(A\times X))$ for $T$ as in
  \autoref{exMultisortification}, $\Pad(X,\xi)$ is the following $\bar T$-coalgebra:
  \[
    \begin{array}{l}
    (X,\Bagf X\times \Dist(A\times X),
      \Bagf X, \Dist (A\times X), A\times X)
      \xrightarrow{(\xi,\id,\id,\id,\id)}
      \\
    \qquad
    (\Potf(\Bagf X\times \Dist(A\times X)),
    \Bagf X\times \Dist(A\times X),
    \Bagf X, \Dist (A\times X), A\times X)
    \end{array}
  \]

  Given a $\bar T$-coalgebra
  \[
    \bar \xi\colon \bar X = (X_1,X_2,X_3,X_4,X_5)
    \xrightarrow{(\xi_1,\xi_2,\xi_3,\xi_4,\xi_5)}
    (\Potf X_2,X_3\times X_4,\Bagf X_1,\Dist X_5,A\times X_1) = \bar T
    \bar X
  \]
  we obtain the following $T$-coalgebra $\Comp(\bar X, \bar \xi)$ on $X_1$:
  \[
    X_1 \xrightarrow{\xi_1} \Potf X_2 \xrightarrow{\Potf \xi_2} \Potf
    (X_3 \times X_4) \xrightarrow{\Potf(\xi_3 \times \xi_4)}
    \Potf(\Bagf X_1 \times \Dist X_5) \xrightarrow{\Potf(\id \times \Dist \xi_5)}
    \Potf(\Bagf X_1 \times \Dist (A \times X_1)).
  \]
\end{example}
\noindent The above coalgebra $\Pad(X,\xi)$ is no longer finite; e.g.~$\Bagf(X)$ is
infinite for nonempty~$X$. However, one can find a finite
$\bar T$-coalgebra that conforms to $\xi$ by restricting e.g.~the
sort $\Bagf(X)$ to those elements of $\Bagf(X)$ that actually appear
in $\xi$.

\noindent Note that a functor $H\colon \Set^k\to \Set$ is finitary if
and only if for every finite set $X$ and every map
$g\colon X\to H(Y_1,\ldots,Y_k)$ there exist finite subsets
$m_i\colon Y_i'\subto Y_i$ such that $g$ factorizes through
$H(m_1,\ldots,m_k)$:
\begin{equation}
  \begin{mytikzcd}
    X
    \arrow{r}{g}
    \arrow[dashed]{dr}[swap]{\exists g'}
    & H(Y_1,\ldots,Y_k)
    \\
    & H(Y_1',\ldots,Y_k')
    \arrow{u}[swap]{H(m_1,\ldots,m_k)}
  \end{mytikzcd}
\end{equation}
In situations where $Y_i=G(Z_1,\ldots,Z_\ell)$ for another finitary
functor $G\colon \Set^\ell\to \Set$, this process can be repeated for
each of the $m_i \colon Y_i'\to G(Z_1,\ldots,Z_\ell)$. Formally:
\begin{construction}\label{constr:barT}
  Let $\mathcal H$ be a set of finitary functors, let
  $T\colon \Set \to \Set$ be a functor composed from~$\mathcal H$ as
  in \autoref{defFlattening} with flattening
  $\bar T\colon \Set^n\to \Set^n$ (and a bijection $f$), and let
  $(X,\xi)$ be a finite $T$-coalgebra. We construct a
  $\bar T$-coalgebra $\op{Factor}(X,\xi)$ by repeatedly applying the
  above factorization technique. In this way we obtain finite sets
  $X_1,\ldots,X_n$, for every subterm $T'$ a map
  $m_{T'}\colon X_{f(T')}\to T'X$, and for every non-leaf subterm
  $T' = H(G_1, \ldots, G_k)$ of $T$ a map
  $\bar\xi_{T'}\colon X_{f(T')} \to H(X_{f(G_1)}, \ldots,
  X_{f(G_k)})$. More precisely, we start by putting
  $m_{T} := \xi\colon X_1 := X\to TX$ and then proceed recursively
  down the syntax tree of $T$. For $T' = (-)$ we put
  $m_{(-)} = \id_X \colon X_1 = X \to X$, and for a non-leaf
  subterm $T'=H(G_1,\ldots,G_k)$ we apply the above factorization to
  (the already chosen) $m_{T'}$, i.e.~we choose finite subsets
  $m_{G_i}\colon X_{f(G_i)} \subto G_iX$, $i = 1, \ldots, k$, and a
  map $\bar \xi_{T'}$ such that the triangle below commutes:
  \begin{equation}
    \begin{mytikzcd}
      X
      \arrow{r}{m_{T'}}
      \arrow{dr}[swap]{\bar\xi_{T'}}
      & H(G_1X,\ldots,G_kX)
      \\
      & H(X_{f(G_1)},\ldots,X_{f(G_n)}).
      \arrow{u}[swap]{H(m_{G_1},\ldots,m_{G_n})}
    \end{mytikzcd}
    \label{eq:functorFactor}
  \end{equation}
  For $G_i\neq (-)$, we choose
  $m_{G_i}\colon X_{f(G_i)}\hookrightarrow G_i(X)$ to be a minimal
  finite subset admitting such a factorization $\bar \xi_{T'}$, and
  for $G_i = (-)$ we use
  $m_{(-)} = \id_{X}$ as defined above. Note that we keep track of
  duplicate subterms as in \autoref{defFlattening}.
  The $\ith$ component of the
  $\bar T$-coalgebra $\op{Factor}(X,\xi)$ is now defined using
  $f^{-1}(i) = H(G_1,\ldots,G_k)$ by
  \[
    X_{i} = X_{f(H(G_1,\ldots,G_k))} \xrightarrow{\bar\xi_{H(G_1,\ldots,G_k)}}
    H(X_{f(G_1)},\ldots,X_{f(G_k)}) = (\bar T(X_1,\ldots,X_n))_i.
  \]
  Note that the sets $X_1,\ldots,X_n$ are finite, and for $T'\neq T$,
  $m_{T'}$ is injective.
\end{construction}
\begin{proposition}
  For every $T$-coalgebra $(X,\xi)$ we have
  \[
    (X,\xi) = \op{Comp}(\op{Factor}(X,\xi)).
  \]
  Moreover,
  \[
    \bar m = \big(
    \op{Factor}(X,\xi) \xrightarrow{(\id_X,m_{f^{-1}(2)},m_{f^{-1}(3)},\ldots,m_{f^{-1}(n)})} \op{Pad}(X,\xi)
    \big)
  \]
  is a $\bar T$-coalgebra morphism, in fact a subcoalgebra inclusion.
\end{proposition}
\begin{proof}
  We first verify that $\bar m$ is a $\bar T$-coalgebra morphism. We do this
  component-wise and by case distinction.

  In the first component, we have $f^{-1}(1)= T = H(G_1,\ldots,G_k)$ and 
  \[
    \begin{tikzcd}
      \text{in }\op{Factor}(X,\xi):
      &
      X
      \arrow{d}[swap]{\bar m_1 = \id_X}
      \arrow{r}{\bar\xi_{H(G_1,\ldots,G_k)}}
      \descto{dr}{\eqref{eq:functorFactor}}
      &[8mm] H(X_{f(G_1)},\ldots,X_{f(G_k)})
      \mathrlap{\ = \bar T(X_1,\ldots,X_n)_1}
      \arrow{d}{H(m_{G_1},\ldots,m_{G_k}) = (\bar T(\bar m))_1}
      \\
      \text{in }\op{Pad}(X,\xi):
      &
      X \arrow{r}{\xi = m_{T}}
      & H(G_1(X),\ldots,G_k(X))
      \mathrlap{\ = \bar T(X_1,\ldots,X_n)_1}
    \end{tikzcd}
    \qquad
  \]
  In all other components we have $f^{-1}(i) = T' = H(G_1,\ldots,G_k)$ and
  \[
    \begin{tikzcd}[column sep=2mm]
      \text{in }\op{Factor}(X,\xi):
      &
      X_{f(H(G_1,\ldots,G_k))} 
      \arrow{d}[swap]{\bar m_i = m_{f^{-1}(i)} = m_{T'}}
      \arrow{r}{\bar\xi_{H(G_1,\ldots,G_k)}}
      \descto{dr}{\eqref{eq:functorFactor}}
      &[12mm] H(X_{f(G_1)},\ldots,X_{f(G_k)})
      \arrow{d}{H(m_{G_1},\ldots,m_{G_k}) = (\bar T(\bar m))_i}
      \mathrlap{\ = \bar T(X_1,\ldots,X_n)_i}
      \\
      \text{in }\op{Pad}(X,\xi):
      &
      T'X \arrow{r}{\id}
      & H(G_1(X),\ldots,G_k(X)).
      \mathrlap{\ = \bar T(X_1,\ldots,X_n)_i}
    \end{tikzcd}
    \qquad
  \]
  We know that $\op{Comp}$ is a functor and that
  $\op{Comp}(\op{Pad}(X,\xi)) = (X,\xi)$. Then the definition of
  $\op{Comp}$ on morphisms yields
  $\op{Comp}(\bar m) = \bar m_1 = \id_X$ and therefore the desired identity. 
  \[
    \begin{tikzcd}[column sep = 15mm]
      \op{Comp}(\op{Factor}(X,\xi))
      \arrow[equals]{r}{\op{Comp}(\bar m)}
      & 
      \op{Comp}(\op{Pad}(X,\xi)) = (X,\xi).
    \end{tikzcd}
    \tag*{\qedhere}
  \]
\end{proof}
\begin{remark}
  We do not need $\op{Factor}$ to be functorial. In fact, it is
  functorial if every $H\colon \Set^k\to \Set$ in $\mathcal{H}$
  preserves inverse images. However, some of our functors of interest,
  e.g.~$\R^{(-)}$, do not preserve inverse images.
\end{remark}

\begin{example}
  Consider a finite coalgebra $\xi\colon X \to TX$ for the functor
  $T = \Potf(\Bagf(-) \times \Dist(A\times (-)))$ from
  \autoref{exMultisortification}. The above \autoref{constr:barT}
  yields a (multisorted) $\bar T$-coalgebra $\op{Factor}(X,\xi)$ with
  finite carriers $(X,X_2, X_3, X_4,X_5)$, and structure maps\twnote{I've made
    the variable names consistent with the construction}
  \[
    \begin{array}{r@{}l@{\qquad}r@{}l@{\qquad}r@{}l}
    \bar \xi_{T}&\colon X\to \Potf X_2
    &
    \bar \xi_{\Bagf(-)\times \Dist(A\times (-))}&\colon X_2\to X_3\times X_4
      &
    \bar \xi_{\Bagf (-)}&\colon X_3\to \Bagf X
        \\
    \bar \xi_{\Dist(A\times (-))}&\colon X_4\to \Dist X_5
    &
    \bar \xi_{A\times(-)}&\colon X_5\to A\times X.
    \end{array}
  \]
\end{example}

\noindent For our purposes it is crucial that we may compute the
simple quotient of $\op{Factor}(X,\xi)$ and obtain from its first
component the simple quotient of $(X,\xi)$.  
\begin{proposition}
  If a $\bar T$-coalgebra $(\bar X,\bar \xi)$ is simple, then so is
  the $T$-coalgebra $\op{Comp}(\bar X,\bar\xi)$. 
\end{proposition}
\begin{proof}
  Let $q\colon \op{Comp}(\bar X,\bar \xi)\to (Y,\zeta)$; by
  \autoref{prop:simple}, we need to show that~$q$ is monic. From~$q$,
  we obtain a $\bar T$-coalgebra morphism $q' \colon (X,\xi)\to\Pad(Y,\zeta)$
  by adjoint transposition; that is,
  \begin{equation*}
    q' = \big(
    (\bar X,\bar \xi)
    \xrightarrow{\eta_{(\bar X,\bar \xi)} = (\id_{X_1},m_2,\ldots,m_n)}
    \op{Pad}(\op{Comp}(\bar X,\bar \xi))
    \xrightarrow{\op{Pad}(q) = (q,q_2,\ldots,q_n)}
    \op{Pad}(Y,\zeta)
    \big). 
  \end{equation*}
  Since $(\bar X, \bar \xi)$ is simple, $q'$ is monic in $\Set^n$; in
  particular, the first component~$q$ of~$q'$ is monic, as required.
\end{proof}
\begin{corollary}
  If $q\colon \op{Factor}(X,\xi)\epito (Y,\zeta)$ represents the
  simple quotient of the $\bar T$-coalgebra $\op{Factor}(X,\xi)$, then
  $\op{Comp}(q)\colon (X,\xi)\to \Comp(Y,\zeta)$ represents the simple
  quotient of $(X,\xi)$.
\end{corollary}

\takeout{
A crucial property for us is that a $T$-coalgebra $(X,\xi)$ and its
associated $\bar T$-coalgebra $\op{Factor}(X,\xi)$ have essentially
the same quotients:

\begin{proposition}
  \label{compQuotientCorrespondence}
  For every $\bar T$-coalgebra $(X,\xi)$ and surjective map $q\colon
  X_1\epito Y$, where $X_1$ is the carrier of $\op{Comp}(X,\xi)$, the following
  are equivalent:
  \begin{enumerate}
    \item $q$ is a $T$-coalgebra morphism with domain $\op{Comp}(X,\xi)$.
    \item There exists surjective maps $e_2,\ldots,e_n$ such that
      $(q,e_2,\ldots,e_n)$ is a $\bar T$-coalgebra morphism with
      domain $(X,\xi)$.
  \end{enumerate}
\end{proposition}
\begin{proof}
  If $q\colon \op{Comp}(X,\xi)\to (Y,\zeta)$ is a $T$-coalgebra morphism, then $\op{Pad}(q)$ is a $\bar
  T$-coalgebra morphism, and so is
  \[
    q' = \big(
    (X,\xi)
    \xrightarrow{\eta_{(X,\xi)} = (\id_{X_1},m_2,\ldots,m_n)}
    \op{Pad}(\op{Comp}(X,\xi))
    \xrightarrow{\op{Pad}(q) = (q,q_2,\ldots,q_n)}
    \op{Pad}(Y,\zeta)
    \big).
  \]
  Denote the image factorization of $q'$ by $q' = h\cdot e$. Since $q$ was
  surjective, the $\bar T$-coalgebra morphism $e$ is of the form
  $(q,e_2,\ldots,e_n)$ for some surjective maps as desired.

  Conversely, consider a $\bar T$-coalgebra morphism
  \[
    (X,\xi) \xrightarrow{(q,e_2,\ldots,e_n)} (Y',\zeta)
  \]
  where $Y$ is the first component of $Y'$. We have $q =
  \op{Comp}(q,e_2,\ldots,e_n)\colon \op{Comp}(X,\xi) \to \op{Comp}(Y,\zeta)$ in
  $\op{Coalg}(T)$.
\end{proof}
\begin{corollary}
  \begin{enumerate}
  \item If a $\bar T$-coalgebra $(X,\xi)$ is simple, then so is the
    $T$-coalgebra $\op{Comp}(X,\xi)$.\lsnote{Just prove this directly, forget about
    the previous proposition}
  \item Let $(X,\xi)$ be $T$-coalgebra and let
    $q\colon \op{Factor}(X,\xi)\epito (Y,\zeta)$ represent the simple
    quotient. Then $\op{Comp}(q)\colon (X,\xi)\to \Comp(Y,\zeta)$ represents
    the simple quotient of $(X,\xi)$.
  \end{enumerate}
\end{corollary}
That means we can compute the simple quotient of $\op{Factor}(X,\xi)$
and obtain from its first component the simple quotient of $(X,\xi)$.
}


\noindent In short, the problem of minimizing single-sorted coalgebras
for functors composed in some way from functors $H$ reduces to
minimizing multi-sorted coalgebras for the components~$H$.  In the
next subsection, we will, in turn, reduce the latter problem to
minimizing single-sorted coalgebras, using however coproducts of the
components~$H$ in lieu of $\op{Comp}$. The benefit of this seemingly
roundabout procedure is that refinement interfaces, which fail to
combine along functor composition, do propagate along coproducts of
functors as we show in Subsection~\ref{sec:coprod-interface}.

\subsection{De-sorting multisorted coalgebras}\label{sec:breakdown}
We fix a number~$n$ of sorts, and consider coalgebras over~$\C^n$. As
before, we assume that~$\C$, and hence also~$\C^n$, fulfils
\autoref{ass:C}. We assume moreover that $H\colon \C^n\to \C^n$ is a
mono-preserving functor modelling the transition type of multisorted
coalgebras. We now show that under two additional assumptions on~$\C$,
one can equivalently transform $H$-coalgebras into single-sorted
coalgebras, i.e.~coalgebras on~$\C$, formed by taking the coproduct
of the carriers, a process we refer to as
\emph{de-sorting}. Specifically, we need $\C$ to have finite
coproducts (implying finite cocompleteness in combination with
\autoref{ass:C}) and to be \emph{extensive}~\cite{clw93}. We begin by
taking a closer look at these additional assumptions in the setting of
$\C$ and $\C^n$.

\begin{notation}
We have the usual diagonal functor
\[
  \Delta\colon \C\hookrightarrow \C^n\qquad\Delta(X) = (X,\ldots,X).
\]
This functor has a left adjoint given by taking coproducts
(e.g.~\cite[p.~225]{awodey2010category}), which we denote by
\[
  \coprodfunctor\colon \C^n \to \C\qquad\coprodfunctor(X_1,\ldots,X_n)
  =
  X_1+\ldots+X_n.
\]
The unit $\eta_X\colon X\to \Delta\coprodfunctor X$ of the adjunction
consists of the coproduct injections, and the adjoint transpose of a
$\C^n$-morphism $f\colon X\to \Delta Y$, denoted
$[f]\colon \coprodfunctor X\to Y$, arises by cotupling.
\end{notation}
\begin{defiC}[\cite{clw93}]\label{defExtensive}
  A category $\C$ with finite coproducts is called \emph{extensive} if
  the canonical functor $\C/Y_1 \times \C/Y_2 \to \C/{Y_1+Y_2}$ is an
  equivalence of categories for every pair $Y_1, Y_2$ of objects.
\end{defiC}
\begin{remark}\label{rem:extensive}
  \begin{enumerate}
  \item This compact definition can be equivalently rephrased as
    follows~\cite[Proposition~2.2]{clw93}: $\C$ has pullbacks along
    coproduct injections and a diagram
    \begin{equation}\label{diag:ext}
      \begin{tikzcd}
        X_1' \arrow{d}[swap]{h_1'}
        \arrow{r}{f_1}
        &
        X
        \arrow{d}{h}
        &
        X_2'
        \arrow{l}[swap]{f_2}
        \arrow{d}{h_2'}
        \\
        Y_1
        \arrow{r}{\inl}
        &
        Y_1+Y_2
        &
        Y_2
        \arrow{l}[swap]{\inr}
      \end{tikzcd}
    \end{equation}
    comprises two pullback squares if and only if the top row is a
    coproduct diagram.

  \item\label{rem:extensive:2} It easy to see that coproduct
    injections in an extensive category $\C$ are monomorphic, hence
    $\eta_X\colon X\rightarrowtail \Delta\coprodfunctor X$ is a
    monomorphism in $\C^n$.
  \end{enumerate}
\end{remark}
\begin{example}
  Categories with $\Set$-like coproducts are extensive, in particular $\Set$
  itself, the category of partially ordered sets and monotone maps, and the
  category of nominal sets and equivariant maps. Presheaf categories
  are extensive, and more generally, so is every Grothendieck topos. 
\end{example}
We shall make use of the following equivalent description of extensivity:
\begin{proposition}
  A category $\C$ is extensive if and only if for
  every~$n\ge 0$ we have that a commutative square in $\C^n$ as below
  is a pullback iff $[f]\colon \coprodfunctor X'\to X$ is
  an isomorphism.
  \[
    \begin{mytikzcd}
      X'
      \arrow{r}{f}
      \arrow{d}[swap]{h'}
      & \Delta X
      \arrow{d}{\Delta h}
      \\
      Y
      \arrow{r}{\eta_Y}
      & \Delta \coprodfunctor Y
    \end{mytikzcd}
  \]
\end{proposition}
\begin{proof}
  For $n = 0$ and $n=1$ the above equivalence always holds using that,
  for $n=1$, $\eta_Y = \id_Y$. We consider the case $n = 2$. The above
  diagram then reads:
  \[
    \begin{tikzcd}
      (X_1', X_2')
      \arrow{r}{(f_1,f_2)}
      \arrow{d}[swap]{(h_1',h_2')}
      &
      (X,X)
      \arrow{d}{(h,h)}
      \\
      (Y_1,Y_2) \arrow{r}{(\inl,\inr)}
      &
      (Y_1+Y_2,Y_1+Y_2)
    \end{tikzcd}
  \]
  This is a pullback in $\C^2$ iff each of its components is a
  pullback in $\C$, i.e.~the two squares in \eqref{diag:ext} are
  pullbacks. The adjoint transpose $[f]$ is the morphism $[f_1,
  f_2]\colon X_1'+X_2' \to X$, which is an isormorphism if and only if
  the top row in~\eqref{diag:ext} is a coproduct diagram. Hence, the
  equivalence in the statement of the proposition for $n =2$ is
  equivalent to extensivity. We are done since for $n > 2$
  that equivalence clearly follows from that for $n = 2$. 
\end{proof}
\takeout{
\noindent In this notation, we can define extensivity of $\C$ in the
following way:
\begin{definition} \label{defExtensive} A category $\C$ is
  \emph{extensive} if the following equivalence holds for
  every~$n\ge 0$ and every commuting square in $\C^n$ as below:
  The square is a pullback iff $[f]\colon \coprodfunctor X'\to X$ is
  an isomorphism.
  \[
    \begin{mytikzcd}
      X'
      \arrow{r}{f}
      \arrow{d}[swap]{h'}
      & \Delta X
      \arrow{d}{\Delta h}
      \\
      Y
      \arrow{r}{\eta_Y}
      & \Delta \coprodfunctor Y
    \end{mytikzcd}
  \]
\end{definition}
\begin{remark}\label{rem:extensive}
  The usual definition of extensivity states that the canonical
  functor $\C/A \times \C/B \to \C/{A+B}$ is an equivalence of
  categories for every pair $A, B$ of objects. Our
  \autoref{defExtensive} above is easily seen to be equivalent to the
  characterization of extensivity in terms of pullbacks of coproducts
  given in \cite[Proposition 2.2]{clw93}. In every extensive category,
  coproduct injections are monic, so
  $\eta_X\colon X\rightarrowtail \Delta\coprodfunctor X$ is a mono.
\end{remark}
\begin{example}
  Categories with $\Set$-like coproducts are extensive, in particular $\Set$
  itself, the category of partially ordered sets and monotone maps, and the
  category of nominal sets and equivariant maps.
\end{example}
}
 
\noindent Our goal in this section is to relate $H$-coalgebras (in
$\C^n$) with $\coprodfunctor H\Delta$-coalgebras (in $\C$). This is
via two observations:
\begin{enumerate}
  \item The obvious functor from $H$-coalgebras to
    $H\Delta\coprodfunctor$-coalgebras given by
    \[
      (X \xrightarrow{x} HX) \mapsto (X \xrightarrow{x} HX
      \xrightarrow{H\eta_X} H\Delta\coprodfunctor X)
    \]
    preserves and reflects (simple) quotients
    (\autoref{coalgebraEtaFunctor}).
  \item The categories of $\coprodfunctor H\Delta$-coalgebras and
    $H\Delta\coprodfunctor$-coalgebras are equivalent 
    (\autoref{coalgebraExtensivity}).
\end{enumerate}

Since $H$ preserves monomorphisms we have a natural transformation $H\eta\colon
H \monoto H \Delta \coprodfunctor$ with monic components. The first translation
thus follows from~\autoref{P:reduction}: \twnote{I've reordered the two steps
  because this is the order in which the translation works!}

\begin{corollary}
  \label{coalgebraEtaFunctor}
  If $\C$ is extensive, then a $H$-coalgebra $\xi\colon X\to HX$ and
  its induced $H\Delta\coprodfunctor$-coalgebra
  \[
    X\xrightarrow{\xi} HX \xrightarrow{H\eta_X} H\Delta\coprodfunctor X
  \]
  have the same quotients and, hence, the same simple quotient.
\end{corollary}

\begin{lemma}
  \label{coalgebraExtensivity}
  If $\C$ is extensive, then the lifting $\bar\coprodfunctor$ of the
  coproduct functor $\coprodfunctor$
  \begin{equation*}
    \bar\coprodfunctor\colon \Coalg(H\Delta\coprodfunctor) \to
    \Coalg(\coprodfunctor H\Delta), \qquad \bar\coprodfunctor (
    X\xrightarrow{\xi} H\Delta\coprodfunctor X) = (\coprodfunctor X
    \xrightarrow{\coprodfunctor \xi} \coprodfunctor
    H\Delta\coprodfunctor X)
  \end{equation*}is an equivalence of categories.
\end{lemma}
\begin{proof}
  We have to show that $\bar\coprodfunctor$ is full, faithful and
  isomorphism-dense~\cite{joyofcats}. Faithfulness is immediate from
  the fact that already $\coprodfunctor\colon \C^n\to\C$ is faithful, since
  coproduct injections are monic
  (\autoref{rem:extensive}\ref{rem:extensive:2}). To see
  that $\bar\coprodfunctor$ is full, let
  $h\colon \bar\coprodfunctor(X,\xi) \to \bar\coprodfunctor(Y,\zeta)$
  be a $\coprodfunctor H\Delta$-coalgebra morphism. By naturality
  of $\eta$, we then have a commuting diagram
  \begin{equation*}
    \begin{mytikzcd}
      X
      \arrow{r}{\eta_X}
      \arrow{dd}[swap]{\xi}
      & \Delta\coprodfunctor X
      \arrow{d}{\Delta \scriptcoprodfunctor\xi}
      \arrow{r}{\Delta h}
      &[4mm] \Delta \coprodfunctor Y
      \arrow{d}{\Delta \scriptcoprodfunctor\zeta}
      & Y
      \pullbackangle{-135}
      \arrow{dd}{\zeta}
      \arrow{l}[swap]{\eta_Y}
      \\
      & \Delta \coprodfunctor H\Delta\coprodfunctor X
      \arrow{r}{\Delta\scriptcoprodfunctor H\Delta h}
      & \Delta \coprodfunctor H\Delta \coprodfunctor Y
      \\
      H\Delta \coprodfunctor X
      \arrow{ur}{\eta_{H\Delta\scriptcoprodfunctor X}}
      \arrow{rrr}[swap]{H\Delta h}
      & & {} &
      H\Delta \coprodfunctor Y
      \arrow{ul}[swap]{\eta_{H\Delta\scriptcoprodfunctor Y}}
    \end{mytikzcd}
    \text{in }\C^n.
  \end{equation*}
  We have the indicated pullback by extensivity (since $[\eta_Y]=\id$), and thus
  obtain $h'\colon X \to Y$ such that $\zeta\cdot h'= H\Delta h\cdot \xi$ and
  $\eta_Y\cdot h'=\Delta h\cdot \eta_X$. Using naturality of $\eta$
  the latter equality yields
  \[
    \Delta\coprodfunctor h' \cdot \eta_X = \eta_Y \cdot h' = \Delta h
    \cdot \eta_X,
  \]
  which implies that $\coprodfunctor h' = h$ since $\eta_X$ is
  a universal morphism. Now the first equality above states that
  $h'\colon (X,\xi)\to(Y,\zeta)$ is an
  $H\Delta\coprodfunctor$-coalgebra morphism.

  It remains to show that $\bar\coprodfunctor$ is
  isomorphism-dense. So let $(X,\xi)$ be a
  $\coprodfunctor H\Delta$-coalgebra. Form the pullback
  \begin{equation}\label{eq:pb-density}
    \begin{mytikzcd}
      X'
      \arrow{r}{x}
      \arrow{d}[swap]{\xi'}
      \pullbackangle{-45}
      & \Delta X
      \arrow{d}{\Delta \xi}
      \\
      H\Delta X
      \arrow{r}{\eta_{H\Delta X}}
      & \Delta \coprodfunctor H\Delta X
    \end{mytikzcd}
  \end{equation}   
  in $\C^n$. Since $\C$ is extensive,
  $[x]\colon \coprodfunctor X' \to X$ is an isomorphism; we thus have
  an $H\Delta\coprodfunctor$-coalgebra
  \begin{equation*}
    X'\xrightarrow{\xi'}H\Delta X\xrightarrow{H\Delta [x]^{-1}}
    H\Delta\coprodfunctor X'.
  \end{equation*}
  Applying the adjunction to the square~\eqref{eq:pb-density} shows
  $\xi \cdot [x] =\coprodfunctor \xi'$. It follows that the
  isomorphism~$[x]$ is a coalgebra morphism, and hence an isomorphism
  in $\Coalg(\coprodfunctor H\Delta)$, from
  $\bar\coprodfunctor (X',H\Delta[x]^{-1}\cdot\xi')$ to $(X, \xi)$:
  \begin{equation*}
    \begin{mytikzcd}[baseline=(bot.base), column sep = 17mm]
      \coprodfunctor X'
      \arrow{r}{[x]}
      \arrow{d}[swap]{\coprodfunctor \xi'}
      & X
      \arrow{dd}{\xi}
      \\
      \coprodfunctor H \Delta X
      \arrow{d}[swap]{\coprodfunctor H\Delta [x]^{-1}}
      \arrow{rd}{\id}
      \\
      |[alias=bot]|
      \coprodfunctor H \Delta \coprodfunctor X'
      \arrow{r}[swap]{\coprodfunctor H\Delta [x]}
      & \coprodfunctor H\Delta X 
    \end{mytikzcd}
    \tag*{\qedhere}
  \end{equation*}
\end{proof}

\takeout{
\begin{lemma}
  \label{coalgebraEtaFunctor}
  A $H$-coalgebra $\xi\colon X\to HX$ and its induced
  $H\Delta\coprodfunctor$-coalgebra
  \[
    X\xrightarrow{\xi} HX \xrightarrow{H\eta_X} H\Delta\coprodfunctor X
  \]
  have the same quotients and, hence, the same simple quotient.
\end{lemma}
\begin{proof}
  We prove only the first claim. Let $q\colon X\epito Y$ be a regular
  epimorphism. It suffices to show that $q$ carries an $H$-coalgebra
  morphism with domain $(X,\xi)$ iff it carries an
  $H\Delta\coprodfunctor$-coalgebra morphism with domain
  $(X,H\eta_X\cdot\xi)$. Since $H$ preserves monomorphisms, we see
  that $H\eta\colon H \to H \Delta \coprodfunctor$ induces an
  embedding $\Coalg(H)\to\Coalg(H\Delta\coprodfunctor)$. Hence, `only
  if' is clear, and we prove `if'. Suppose that $q$ is a coalgebra
  morphism $(X, H\eta\cdot \xi) \to (Y,\zeta)$. Then the outside of
  the following diagram commutes:
  \[
    \begin{mytikzcd}
      X
      \arrow{r}{\xi}
      \arrow[->>]{d}[swap]{q}
      & HX
      \arrow{r}{H\eta_X}
      \arrow{d}[swap]{Hq}
      & H\Delta\coprodfunctor X
      \arrow{d}{H\Delta\scriptcoprodfunctor q}
      \\
      Y
      \arrow[dashed]{r}[near end]{\exists! \zeta'}
      \arrow{rr}[swap]{\zeta}
      & |[yshift=5mm]| HY
      \arrow[>->]{r}{H\eta_Y}
      & H\Delta\coprodfunctor Y
    \end{mytikzcd}
  \]
  Since $\eta_X$ is monic and $H$ preserves monomorphisms, $H\eta_Y$ is monic,
  so we obtain $\zeta'$ as in the diagram by the diagonal fill-in property
  (Section~\ref{sec:equiv}), making $q$ an $H$-coalgebra morphism
  $(X,\xi)\to(Y,\zeta')$.
\end{proof}}

So the task of computing the simple quotient of a
multisorted coalgebra is reduced again to the same problem on ordinary
coalgebras in \Set. Thus, it remains to check that the arising
functor $\Delta H\coprodfunctor$ indeed fulfils
\autoref{interface-linear}.

\subsection{Coproducts of refinement
  interfaces}\label{sec:coprod-interface}

We have already seen that zippable functors are closed under
coproducts (\autoref{lem:closure}). We proceed to show that we can
also combine refinement interface along coproducts. Let functors
$H_i\colon \Set\to \Set$, $1 \leq i \leq n$ have refinement interfaces
with labels $A_i$ and weights $W_i$, and associated functions
$\flat_i,\op{init}_i,w_i,\op{update}_i$. We construct a refinement
interface for the coproduct $H=\sortcoprod H_i$, with labels
$A = \sortcoprod A_i$ and weights $W = \sortcoprod W_i$, as follows.
First define the following helper function, which restricts a multiset
of labels to a given sort $i$:
  \[\textstyle
    \op{filter}_i\colon \Bagf(\textstyle\coprod_{j=1}^n A_j)\to\Bagf(A_i),
    \qquad \op{filter}_i(f)(a) = f(\inj_i(a)).
  \]
  (Note that this differs from the filter function $\op{fil}_S$ in
  \eqref{eqSplitterLabels}, which filters for a subset of states
  $S\subseteq X$.) Then we implement the refinement interface for $H$
  component-wise as follows (writing $I=\{1,\dots,n\}$):
  \begingroup\allowdisplaybreaks
  \begin{align*}
    \flat &= \big(\sortcoprod H_iY \xrightarrow{\coprod \flat_i}
    \sortcoprod\Bagf(A_i×Y) \xrightarrow{[\Bagf(\inj_i\times
      Y)]_{i\in I}
            }
            \Bagf(\sortcoprod A_i \times Y)\big)
    \\
    w(S) &= \big(\sortcoprod H_i Y \xrightarrow{\coprod w_i(S)} \sortcoprod W_i\big)
    \qquad\text{for }S\subseteq Y,
    \\
    \op{init} &= \big(
                       \begin{mytikzcd}[column sep=6mm,every cell/.append style={inner sep =0.5pt, outer sep=0.5pt}]
                       \sortcoprod H_i 1×\Bagf \sortcoprod A_i
                       \arrow{r}[yshift=2mm]{\underbrace{(\inj_i(t),a) \mapsto \inj_i(t,a)}}
                       & \coprod_i \big( H_i 1×\Bagf \coprod_j A_j  \big)
                       \arrow{r}[yshift=2mm]{\underbrace{\coprod{\id\times \op{filter}_i}}}
                       & \sortcoprod \left( H_i 1×\Bagf A_i  \right)
                       \arrow{r}[yshift=2mm]{\coprod{\op{init}}_i}
                       & \sortcoprod W_i
                       \end{mytikzcd}\big)
      \\
      \op{update}&= \big(
                       \begin{mytikzcd}[column sep=3mm,every cell/.append
                         style={inner sep =0.5pt, outer sep=0.5pt},
                         baseline=(firstline.base)]
                       |[alias=firstline]|
                       \Bagf \sortcoprod\!A_i\times \sortcoprod  W_i
                       \arrow{r}[yshift=2mm]{\underbrace{(a,\inj_i(t)) \mapsto \inj_i(a,t)}}
                       & \coprod_i \big(\,\Bagf\!\coprod_j A_j \times W_i\big)
                       \arrow{r}[yshift=2mm]{\underbrace{\coprod{(\op{filter}_i\times \id)}}}
                       & \sortcoprod (\Bagf A_i \times W_i)
                       \arrow{r}[yshift=2mm]{\underbrace{\coprod{\op{update}_i}}}
                       & \smash{\sortcoprod}\!  (W_i\times H_i3\times W_i)
                       \arrow{d}{[\inj_i\times\inj_i\times\inj_i ]_{i\in I}
 }                      \\
                       & & & \sortcoprod\!  W_i\times \sortcoprod\! H_i3\times \sortcoprod\! W_i\big)
                       \end{mytikzcd}\!\!
  \end{align*}\endgroup

\begin{proposition}
  \label{coprodInterface}
  The data $W$, $A$, $\flat$, $\op{init}$,~$w$, and $\op{update}$ as
  constructed above form a refinement interface for
  $H=\sortcoprod H_i$, and if the interfaces of the~$H_i$ fulfil
  \autoref{interface-linear}, then so does the one of~$H$.
\end{proposition}
\begin{proof}
  For $i = 1, \ldots, n$, the following diagram commutes:
  \twnote{the implicit indexing is wrong!!!}
  \[
    \begin{mytikzcd}[column sep=17mm]
        \sortcoprod[j] H_j Y
        \arrow{d}[left, pos=0.4]{\begin{array}{r}
            \fpair{
            H!,
            \\[-2mm]
            \Bagf\pi_1\cdot \flat}
            \end{array}}
          \arrow[to path={
            |- ([xshift=10mm,yshift=10mm]HiY.base)
            -- ([xshift=10mm,yshift=10mm]Wi.base) \tikztonodes
            |- (\tikztotarget)
          }, rounded corners]{rrdd}{w(Y)}
        & |[alias=HiY]| H_iY
        \arrow{l}[swap]{\inj_i}
        \arrow{d}[left, pos=0.4]{\begin{array}{r}
            \fpair{
            H_i!,
            \\[-2mm]
            \Bagf\pi_1\cdot \flat_i}
            \end{array}}
        \arrow[shift left=2]{dr}{w_i(Y)}
        \\
        \sortcoprod[j] H_j1 × \Bagf \sortcoprod[j]A_j
        \arrow{d}[swap]{(\inj_i(t),a) \mapsto \inj_i(t,a)}
        \arrow[to path={
          -- ([xshift=-10mm]\tikztostart.west)
          |- ([yshift=-5mm]\tikztotarget.south) \tikztonodes
          -- (\tikztotarget)
        }, rounded corners]{drr}[pos=0.8]{\op{init}}
        & H_i1 × \Bagf A_i
            \arrow{r}{\op{init}_i}
            \arrow{d}{\inj_i}
            \arrow{l}[swap]{\inj_i\times \Bagf \inj_i}
        & |[alias=Wi]| W_i
            \arrow{d}{\inj_i}
        \\
        \coprod_j \big(H_j1 × \Bagf \coprod_kA_k\big)
        \arrow{r}{\coprod(\id\times\op{filter}_j)}
        & \sortcoprod[j](H_j 1\times \Bagf A_j)
        \arrow{r}{\coprod \op{init}_j}
        & \sortcoprod[j] W_j
    \end{mytikzcd}
  \]
  Since the $\inj_i\colon H_i Y \to \sortcoprod H_i Y$ are jointly epic, the
  commutativity shows the axiom for $\op{init}$ in~\eqref{eqSplitterLabels}.
  For $S\subseteq C\subseteq Y$ we have the diagram:
  \[
    \begin{mytikzcd}[column sep=12mm,row sep=10mm]
    \smash{\sortcoprod[j]} H_i Y
    \arrow{d}[pos=0.4]{\fpair{\flat,w(C)}}
    \arrow[to path={
      -- ([yshift=6mm]\tikztostart.north)
      -| ([xshift=5mm]\tikztotarget.east) \tikztonodes
      -- (\tikztotarget)
    }, rounded corners]{dddrr}[pos=0.2]{\fpair{w(S),H\chi_S^C, w(C\setminus S)}}
    &
    H_iY
    \arrow[
        to path={
            -| (\tikztotarget) \tikztonodes
        },
        rounded corners,
    ]{ddr}[pos=0.23,above,yshift=1mm]{\fpair{w_i(S),H_i\chi_S^C, w_i(C\setminus S)}}
    \arrow{d}[pos=0.4,left]{\fpair{\flat_i,w_i(C)}}
    \arrow{l}[swap]{\inj_i}
    \descto[yshift=2mm,xshift=2mm]{ddr}{\eqref{eqSplitterLabels} for $H_i$}
    \\
    \Bagf(\sortcoprod[j] A_j×Y) × \sortcoprod[j] W_j
    \arrow{d}[]{\op{fil}_S × W_i}
    \descto[yshift=2mm]{dr}{Naturality of $\op{fil}_S$\\[-1pt]  (\autoref{R:filS})}
    & \Bagf(A_i×Y) × W_i
    \arrow{l}[swap,yshift=2mm]{\Bagf(\inj_i\times Y)\times \inj_i}
    \arrow{d}[]{\op{fil}_S × W_i}
    \\
    \Bagf \sortcoprod[j]A_j × \sortcoprod[j] W_j
    \arrow{d}[swap]{(a,\inj_i(t)) \mapsto \inj_i(a,t)}
    \arrow[to path={
      -- ([xshift=-11mm]\tikztostart.west)
      |- ([yshift=-5mm]\tikztotarget.south) \tikztonodes
      -- (\tikztotarget)
    }, rounded corners]{drr}[pos=0.8]{\op{update}}
    & \Bagf A_i × W_i
    \arrow{l}[swap,yshift=2mm]{\Bagf\inj_i\times \inj_i}
    \arrow{d}{\inj_i}
    \arrow{r}[yshift=2mm]{\op{update}_i}
    & W_i × H_i3 × W_i
    \arrow{d}{\inj_i}
    \\
    \smash{\coprod_j} (\Bagf\coprod_kA_k \times W_j)
    \arrow{r}{\sortcoprod[j] \op{filter}_j\times W_j}
    & \sortcoprod[j] (\Bagf A_j × W_j)
    \arrow{r}[yshift=2mm]{\sortcoprod[j]\op{update}_j}
    & \sortcoprod[j] (W_j × H_j3 × W_j)
    \end{mytikzcd}
  \]
  Using again that the $\inj_i\colon H_i Y \to \sortcoprod H_i Y$ are
  jointly epic, we see that the claimed refinement interface fulfils
  the axiom for $\op{update}$ in~\eqref{eqSplitterLabels}.  If for
  every $i$, the interface of $H_i$ fulfils the time constraints from
  \autoref{interface-linear}, then so does the interface for $H$: both
  $\op{init}$ and $\op{update}$ preprocess the parameters in linear
  time (via $\op{filter}$), before calling the $\op{init}_i$ and
  $\op{update}_i$ of the corresponding $H_i$. We order $H3$
  lexicographically, i.e.\ $\inj_i(x) < \inj_j(y)$ iff either $i<j$ or
  $i=j$ and $x< y$ in $H_i3$.  Comparison in constant time is then
  clearly inherited.
\end{proof}

\noindent Hence, our coalgebraic partition refinement algorithm is
modular w.r.t.~coproducts.  In combination with the results of
Sections~\ref{sec:intermediate} and~\ref{sec:breakdown}, this gives a
modular efficient minimization algorithm for multisorted coalgebras,
and hence for coalgebras for composite functors.

We proceed to see examples employing the multi-sorted approach,
complementing the examples already given for the single-sorted
approach (\autoref{ex:singlesort}). We build our examples from the
functors
\[
  \Dist,\Potf,A\times (-)\colon \Set\to \Set
  \qquad
  \times,+\colon \Set^2\to \Set,
\]
and in fact many of them appear in work on a coalgebraic hierarchy of
probabilistic system types~\cite{BARTELS200357}. 
\begin{example}\label{ex:comp}
\begin{enumerate}[wide=0pt]

\item \emph{Labelled transition systems} with an infinite set~$A$ of
  labels.  Here we decompose the coalgebraic type functor
  $\Potf (A \times (-))$ into $H_1=\Potf$ and $H_2=A\times(-)$. We
  transform a coalgebra $X\to \Potf(A\times X)$ (with $m$ edges) into
  a multisorted $(H_1,H_2)$-coalgebra; the new sort $Y$ then contains
  one element per edge. By de-sorting, we finally obtain a
  single-sorted coalgebra for $\bar H = H_1 + H_2$ with $n+m$ states
  and $m$ edges, leading to a complexity of
  $\CO((n+m)\cdot \log (n+m))$. If $m \geq n$ we thus obtain a run-time
  in $\CO(m\cdot \log m)$, like in \cite{DovierEA04} but slower
  than Valmari's $\CO(m\cdot \log n)$~\cite{Valmari09}.

  For fixed finite $A$, the running time of our algorithm is in
  $\CO((m+n)\log n)$.  Indeed, by finiteness of $A$, we have
  $\Potf(A\times(-)) \cong \Potf(-)^A$. Then a coalgebra
  $X\xrightarrow{\xi} \Potf(A\times X) \cong \Potf(X)^A$ with $n= |X|$
  states and $m = \sum_{x\in X}|\xi(x)|$ edges is transformed into a
  two-sorted coalgebra
  \[
  \eta_X\colon X\to (A\times X)^A, \;x\mapsto (a\mapsto (a,x));\qquad
  \bar \xi\colon A\times X \to \Potf(X),\;(a,x)\mapsto \xi(x)(a).
  \]
  The arising de-sorted system on $X+A\times X$ has $n + |A|\cdot n$
  states and $|A|\cdot n + m$ edges, so the simple quotient is found
  in $\CO((m+n)\cdot \log n)$ like in the single-sorted approach
  (\autoref{ex:singlesort}).

\end{enumerate}\begin{enumerate}[resume,wide]

\item As mentioned already, Hopcroft's classical automata
  minimization~\cite{Hopcroft71} is obtained by instantiating our
  approach to $HX = 2×X^A$, with running time $\CO(n\cdot\log n)$ for
  fixed alphabet $A$. For non-fixed $A$ the best known complexity is
  in $\CO(|A| \cdot n\cdot\log n)$~\cite{Gries1973,Knuutila2001}.  To
  obtain the alphabet as part of the input to our algorithm, we
  consider DFAs as labelled transition systems encoding the letters of
  the input alphabet as natural numbers. More precisely, given a
  finite input alphabet $A$, we choose
  some injective encoding map $c\colon A \monoto \N$ and we form the natural
  transformation with the components
  \[
    m_X\colon 2 \times X^A \monoto 2 \times \Potf(\N \times X)
    \qquad\text{with}\qquad
    m_X(b,f) = \big(b, \{(c(a), f(a)) \mid a \in A\}\big).
  \]
  Since $m_X$ is clearly monomorphic, we apply
  \autoref{P:reduction} to see that minimization of a DFA $\xi\colon
  X \to 2 \times X^A$ is reduced to minimizing the coalgebra
  \[   
    X \xrightarrow{\xi}2 \times X^A \xrightarrow{m_X} 2 \times \Potf(\N \times X).
  \]
  Further, we decompose the type
  functor into $H_1 = 2 \times \Potf$ and $H_2 = \N \times (-)$.  An
  automaton $\xi$ for a finite input alphabet $A$ 
  is then represented by the two-sorted system
  \[
    \xi\colon X\to 2\times \Potf (A\times X)
    \qquad
    c\times X \colon A\times X\to \N \times X
  \]
  With $|X| =n $, this system has $n + |A|\cdot n$ states and
  $|A|\cdot n + |A|\cdot n$ edges. Thus, our algorithm runs in time
  \[
    \CO((|A| \cdot n)\cdot\log (|A|\cdot n))
    = \CO(|A| \cdot n\cdot\log n + |A|\cdot n\cdot\log|A|),
  \]
  i.e.\ as fast as the above-mentioned best known algorithms except on
  automata with more alphabet letters than states.
  
\item\label{item:alternating} Coalgebras for the functor
  $H X = \Dist X + \Potf(A\times X)$ are \emph{alternating
    systems}~\cite{Hansson:1994:TPF:561335}. The functor $H$ is
  flattened to the multi-sorted functor
  \[
    \bar H(X_1,X_2,X_3,X_4) = (X_2+ X_3, \Dist X_1, \Potf X_4, A\times X_1)
  \]
  on $\Set^4$, which is then de-sorted to obtain the \Set-functor
  \[
    \coprodfunctor \bar H\Delta X = (X+ X) + \Dist X+ \Potf X+ A\times X
  \]
  which has a refinement interface as given by
  \autoref{coprodInterface}. Given an $H$-coalgebra with~$n$ states
  and~$m_d$ edges of type~$\Dist$ and~$m_p$ edges of type
  $\Potf(A\times (-))$, the induced
  $\coprodfunctor \bar H\Delta$-coalgebra has $n+m_p$ states and
  $n+m_d+m_p+m_p$ edges, and is minimized under bisimilarity in time
  \( \CO((n+m_d+m_p)\cdot \log (n+m_p)).  \)

  Other probabilistic system types~\cite{BARTELS200357} are handled
  similarly, where one only needs to take care of estimating the
  number of states in the intermediate sorts as in the treatment
  above. We discuss two further examples explicitly, simple and
  general Segala systems.

\item For a simple Segala system considered as a coalgebra $\xi\colon X\to
  \Potf(A\times \Dist X)$ there are partition refinement algorithms by Baier,
  Engelen, Majster-Cederbaum\ \cite{BaierEM00} and by Groote, Verduzco, and de
  Vink\ \cite{GrooteEA18}. For the complexity analysis, define the number of
  states and edges respectively as
  \[\textstyle
    n = |X|,\qquad
    m_p = \sum_{x\in X} |\xi(x)|.
  \]
  The arising multi-sorted coalgebra consists of maps
  \[
    p\colon X\to \Potf Y
    \qquad
    a\colon Y\to A\times Z
    \qquad
    d\colon Z\to \Dist X.
  \]
  In the coalgebra $\xi$ there is one distribution per
  non-deterministic edge, hence $|Y| = m_p = |Z|$. The
  non-deterministic map~$p$ has $m_p$ edges by construction, and the
  deterministic map~$a$ has $|Y|=m_p$ edges. Let $m_d$ denote the
  number of edges needed to encode $d$; then $m_d \le n\cdot m_p$. We
  thus have $n + 2\cdot m_p$ states and $2\cdot m_p + m_d$ edges, so
  our algorithm runs in time $\CO((n + m_p+m_d)\cdot \log(n + m_p))$.
  For every $z\in Z$, $d(z)$ is a non-empty distribution, and so under
  the assumption that there is at least one non-deterministic edge per
  state $x\in X$, we have $m_d\ge m_p\ge n$, simplifying the
  complexity to $\CO(m_d\cdot \log m_p)$. In independent work, Groote
  et al.~\cite{GrooteEA18} consider $Z$ as part of the input, and
  design and implement an algorithm of time complexity
  $\CO((m_p+m_d) \cdot \log |Z| + m_d\cdot \log n)$, which simplifies
  to the same complexity $\CO(m_d\cdot \log m_p)$ for $|Z| = m_p$ and
  $m_d\ge m_p\ge n$. This is more fine-grained than the complexity
  $\CO((n\cdot m_p)\cdot \log(n+m_p)) = \CO((n\cdot m_p)\cdot
  \log(n\cdot m_p))$
  of~\cite{BaierEM00}, and indeed leads to a faster run-time in the
  (presumably wide-spread) case that probabilistic transitions are
  sparse, i.e.~if~$m_d$ is substantially below $n\cdot m_p$.

\item\label{item:general-segala} For a general Segala system
  $\xi \colon X\to \Potf(\Dist(A\times X))$ one has a similar
  factorization:
  \[
    p\colon X\to \Potf Y
    \qquad
    d\colon Y\to \Dist Z
    \qquad
    z\colon Z\to A\times X
  \]
  So for $n= |X|$ states, $m_p$ non-deterministic edges, and $m_d$ probabilistic
  edges, the multisorted system has $n + |Y| + |Z| = n + m_p + m_d$ states and
  $m_p + m_d + m_d$ edges. Since $d(y)$ is non-empty for all $y\in Y$, $m_d\ge
  m_p$, and so the generic partition refinement has a run-time of
  $\CO(m_d\cdot \log (n + m_d))$. 
\end{enumerate}
\end{example}
\noindent Summing up the last three examples, on simple Segala systems
we obtain faster run-time than the best previous
algorithm~\cite{BaierEM00} (with similar results obtained
independently by Groote et al.~\cite{GrooteEA18}), and we obtain, to
our best knowledge, the first similarly efficient partition refinement
algorithms for alternating systems
(Example~\ref{ex:comp}\ref{item:alternating}) and general Segala
systems~(Example~\ref{ex:comp}\ref{item:general-segala}).
 
\section{Conclusions and Further Work}
\label{sec:conc}
We have presented a generic algorithm that quotients coalgebras by
behavioural equivalence. We have started from a category-theoretic
procedure that works for every mono-preserving functor on a category
with image factorizations, and have then developed an improved
algorithm for \emph{zippable} endofunctors on $\Set$.  Provided the
given type functor can be equipped with an efficient implementation of
a \emph{refinement interface}, we have finally arrived at a concrete
procedure that runs in time $\CO((m+n)\log n)$ where $m$ is the number
of edges and $n$ the number of nodes in a graph-based representation
of the input coalgebra. We have shown that this instantiates to (minor
variants of) several known efficient partition refinement algorithms:
the classical Hopcroft algorithm~\cite{Hopcroft71} for minimization of
DFAs, the Paige-Tarjan algorithm for unlabelled transition
systems~\cite{PaigeTarjan87}, Valmari and Franceschinis's lumping
algorithm for weighted transition systems~\cite{ValmariF10}, and the
1-dimensional Weisfeiler Lehman
Algorithm~\cite{CaiEA1992,Weisfeiler1976,ShervashidzeEA2011}. Moreover,
we have presented a generic method to apply the algorithm to mixed
system types. As an instance, we obtain an algorithm for simple Segala
systems that allows for a more fine-grained analysis of a\-symp\-to\-tic
run-time than previous algorithms~\cite{BaierEM00}, and matches the
run-time of a recent algorithm described independently by Groote
et~al.~\cite{GrooteEA18}.

Further instances can be covered by relaxing the time complexity
assumptions on the refinement interfaces~\cite{DMSW19}, which allows
covering monoid-valued functors~$M^{(-)}$ also for non-cancellative
monoids $M$; by our compositionality methods, we obtain in particular
efficient partition refinement algorithm for $M$-weighted tree
automata, which improves the run-time of a previous algorithm by
Högberg, Maletti, and May~\cite{HoegbergEA07}.

In further work~\cite{DMSW19}, we describe and evaluate a generic
implementation \CoPaR of the generic algorithm. The implementation
supports unrestricted combination of functors with refinement
interfaces (cf.~\autoref{sec:multisorted}) and implements all the
functors from \autoref{sec:efficient}.

It remains open whether our approach can be extended to, e.g.\ the monotone
neighbourhood functor, which is not itself zippable (see \autoref{ex:nbhd}) and
also does not have an obvious factorization into zippable functors. We do expect
that our algorithm applies beyond weighted systems. For example, it should be
relatively straightforward to extend our algorithm to nominal systems,
i.e.~coalgebras for functors on the category of nominal sets and equivariant
maps. Of course, precise complexity bounds will then depend on the
representation of nominal sets.

\bibliographystyle{alpha}
\bibliography{refs}

\end{document}

%% file: parttree.tex
\begin{tikzpicture}
  \begin{scope}[
    every node/.append style={
      inner sep=0pt,
      outer sep=2pt,
      minimum width=1pt,
      minimum height=4pt,
      anchor=center,
    }
    ]
    \begin{scope}[yshift=1.5cm]
      \node (x1) at (-1.2,0) {$x_1$};
      \node (x2) at (0,0) {$x_2$};
      \node (x3) at (1.2,0) {$x_3$};
    \end{scope}
    \node (y1) at (-1.6,0) {$y_1$};
    \node (y2) at (-0.5,0) {$y_2$};
    \node (y3) at (0.5,0) {$y_3$};
    \node (y4) at (1.2,0) {$y_4$};
  \end{scope}
  \node[anchor=west,xshift=2mm] (y5) at (y4.east) {$\ldots$};
  \draw (y1) edge[draw=none] node {$\ldots$} (y2);
  \begin{scope}[
    every node/.append style={
      partitionBlock,
      inner sep=2pt,
    },
    every label/.append style={
      font=\footnotesize,
      anchor=south,
      outer sep=1pt,
      inner sep=1pt,
      minimum height=2mm,
      shape=rectangle,
    },
    ]
    \node[fit=(x1) (x2) (x3)] (x123){};
    \node[fit=(y1) (y5)] (C) {};
    \node[every label] at ([xshift=-1mm]x123.north east) {\normalsize $B$};
    \node[every label] at ([xshift=-1mm]C.north east) {\normalsize $C$};
    \begin{scope}[every node/.append style={minimum width=10mm}, 
      ]
      \newcommand{\nextblockdistance}{8mm} 
      \foreach \name/\nodesource/\anchor/\direction in
      {Peast/x123/east/,
        Qeast/C/east/,
        Pwest/x123/west/-,
        Qwest/C/west/-} {
        \begin{scope}
          \clip ([yshift=1mm]\nodesource.north \anchor)
          rectangle ([yshift=-1mm,xshift=\direction\nextblockdistance]
                     \nodesource.south \anchor);
          \node at ([xshift=\direction\nextblockdistance]\nodesource.\anchor) (\name) {};
        \end{scope}
      }
    \end{scope}
  \end{scope}
  \draw[thick, decoration={brace,mirror},decorate]
  ([yshift=-2mm]y1.south west) -- node[anchor=north,yshift=-2pt]{$S$} ([yshift=-2mm]y2.south east) ;
  \draw[thick, decoration={brace,mirror},decorate]
  ([yshift=-2mm]y3.south west) -- node[anchor=north,yshift=-2pt]{$C\setminus S$} ([yshift=-2mm]y3.south west -| y5.south east) ;
  \foreach \name in {Peast,Pwest,Qeast,Qwest} {
    \node[draw=none,minimum width=1pt,minimum height=1pt,text height=1pt,inner xsep=1pt]
       at (\name.center) {$\ldots$};
  }
  \begin{scope}[text depth=2pt]
  \node[outer sep=5mm,anchor=east] at (Pwest.center -| Qwest.center) (P) {$X/P:$};
  \node[outer sep=5mm,anchor=east] at (Qwest.center) (Q) {$X/Q:$};
  \coordinate (mapXCoordinate) at ([xshift=10mm]Qeast.center);
  \node[outer sep=1pt,anchor=center] at (Peast.center -| mapXCoordinate) (X) {$X$};
  \node[outer sep=1pt,anchor=center] at (Qeast.center -| mapXCoordinate) (PY) {$\Potf X$};
  \end{scope}
  \draw[commutative diagrams/.cd, every arrow, every label] (X) edge node {$\xi$} (PY);
  \begin{scope}[
    bend angle=10,
    space/.style={
      draw=white,
      line width=4pt,
    },
    edge/.style={
      draw=black,
      -{>[length=2mm,width=2mm]},
      preaction={draw,-,line width=1mm,white},
      every node/.append style={
        fill=white,
        shape=rectangle,
        inner sep=1pt,
        anchor=base,
        pos=0.55,
      },
    },
    ]
    \draw[edge,bend right] (x1) to (y1);
    \draw[edge,bend left] (x1) to (y2);
    \draw[edge,bend right] (x2) to (y2);
    \draw[edge,bend left] (x2) to (y3);
    \draw[edge,bend left] (x2) to (y4);
    \draw[edge,bend right] (x3) to (y3);
    \draw[edge] (x3) to (y4);
    \draw[edge,bend left] (x3) to (y5);
  \end{scope}
  \begin{scope}[
    linestyle/.style={
      dashed,
      draw=black!50,
    },
    scissors/.style={
      text=black!50,
      font=\tiny,
      inner sep=0pt,
    },
    gapstyle/.style={
      draw=white,
      line width=1mm,
    }
    ]
    \foreach \leftnode/\rightnode/\northlen/\southlen in
    {x1/x2/4mm/4mm,x2/x3/4mm/4mm,y2/y3/4mm/4mm} {
      \coordinate (northend) at ($ (\leftnode) !.5! (\rightnode) + (0,\northlen)$);
      \coordinate (southend) at ($ (\leftnode) !.5! (\rightnode) - (0,\southlen)$);
      \node[anchor=east,rotate=-90,scissors] at (northend) {\scissors};
      \draw[gapstyle] (northend) -- (southend);
      \draw[linestyle] (northend) -- (southend);
    }
  \end{scope}
\end{tikzpicture}

%% file: coalgpartref-lmcs.bbl
\newcommand{\etalchar}[1]{$^{#1}$}
\begin{thebibliography}{DMSW19}

\bibitem[ABH{\etalchar{+}}12]{DBLP:conf/fossacs/AdamekBHKMS12}
Jir{\'{\i}} Ad{\'{a}}mek, Filippo Bonchi, Mathias H{\"{u}}lsbusch, Barbara
  K{\"{o}}nig, Stefan Milius, and Alexandra Silva.
\newblock A coalgebraic perspective on minimization and determinization.
\newblock volume 7213 of {\em LNCS}, pages 58--73. Springer, 2012.

\bibitem[Ad{\'{a}}05]{Adamek05}
Ji\v{r}\'i Ad{\'{a}}mek.
\newblock Introduction to coalgebra.
\newblock {\em Theory Appl.~Categ.}, 14:157--199, 2005.

\bibitem[AHS90]{joyofcats}
Jiří Adámek, Horst Herrlich, and George Strecker.
\newblock {\em Abstract and Concrete Categories}.
\newblock Wiley Interscience, 1990.

\bibitem[AM89]{aczelmendler:89}
Peter Aczel and Nax Mendler.
\newblock A final coalgebra theorem.
\newblock In {\em Proc.~Category Theory and Computer Science (CTCS)}, volume
  389 of {\em Lecture Notes Comput.~Sci.}, pages 357--365. Springer, 1989.

\bibitem[AR94]{AdamekR94}
Ji\v{r}\'i Ad\'{a}mek and Ji\v{r}\'i Rosick\'y.
\newblock {\em Locally presentable and accessible categories}.
\newblock Cambridge University Press, 1994.

\bibitem[AT90]{AdamekT90}
Ji\v{r}\'i Ad\'{a}mek and V\v{e}ra Trnkov\'a.
\newblock {\em Automata and Algebras in Categories}.
\newblock Kluwer, 1990.

\bibitem[Awo10]{awodey2010category}
Steve Awodey.
\newblock {\em Category Theory}.
\newblock Oxford Logic Guides. OUP Oxford, 2010.

\bibitem[Bac86]{Backhouse1986}
Roland Backhouse.
\newblock {\em Program Construction and Verification}.
\newblock Prentice-Hall, 1986.

\bibitem[BBG17]{BerkholzBG17}
Christoph Berkholz, Paul~S. Bonsma, and Martin Grohe.
\newblock Tight lower and upper bounds for the complexity of canonical colour
  refinement.
\newblock {\em Theory Comput. Syst.}, 60(4):581--614, 2017.

\bibitem[BEM00]{BaierEM00}
Christel Baier, Bettina Engelen, and Mila Majster{-}Cederbaum.
\newblock Deciding bisimilarity and similarity for probabilistic processes.
\newblock {\em J.\ Comput.\ Syst.\ Sci.}, 60:187--231, 2000.

\bibitem[BGK{\etalchar{+}}19]{BunteEA19}
Olav Bunte, Jan~Friso Groote, Jeroen J.~A. Keiren, Maurice Laveaux, Thomas
  Neele, Erik~P. de~Vink, Wieger Wesselink, Anton Wijs, and Tim A.~C. Willemse.
\newblock The mcrl2 toolset for analysing concurrent systems - improvements in
  expressivity and usability.
\newblock In {\em Tools and Algorithms for the Construction and Analysis of
  Systems, {TACAS} 2019, Part {II}}, pages 21--39, 2019.

\bibitem[BO05]{BlomO05}
Stefan Blom and Simona Orzan.
\newblock A distributed algorithm for strong bisimulation reduction of state
  spaces.
\newblock {\em {STTT}}, 7(1):74--86, 2005.

\bibitem[BSdV03]{BARTELS200357}
Falk Bartels, Ana Sokolova, and Erik de~Vink.
\newblock A hierarchy of probabilistic system types.
\newblock In {\em Coagebraic Methods in Computer Science, CMCS 2003}, volume~82
  of {\em ENTCS}, pages 57 -- 75. Elsevier, 2003.

\bibitem[Buc08]{Buchholz08}
Peter Buchholz.
\newblock Bisimulation relations for weighted automata.
\newblock {\em Theoret.~Comput.~Sci.}, 393:109--123, 2008.

\bibitem[CFI92]{CaiEA1992}
Jin-Yi Cai, Martin F\"{u}rer, and Neil Immerman.
\newblock An optimal lower bound on the number of variables for graph
  identification.
\newblock {\em Combinatorica}, 12(4):389--410, dec 1992.

\bibitem[CLW93]{clw93}
Aurelio Carboni, Steve Lack, and Robert F.~C. Walters.
\newblock Introduction to extensive and distributive categories.
\newblock {\em J. Pure Appl. Algebra}, 84:145--158, 1993.

\bibitem[CS02]{CattaniS02}
Stefano Cattani and Roberto Segala.
\newblock Decision algorithms for probabilistic bisimulation.
\newblock In {\em Concurrency Theory, {CONCUR} 2002}, volume 2421 of {\em
  LNCS}, pages 371--385. Springer, 2002.

\bibitem[DEP02]{DesharnaisEA02}
Josee Desharnais, Abbas Edalat, and Prakash Panangaden.
\newblock Bisimulation for labelled markov processes.
\newblock {\em Inf. Comput.}, 179(2):163--193, 2002.

\bibitem[DHS03]{DerisaviEA03}
Salem Derisavi, Holger Hermanns, and William Sanders.
\newblock Optimal state-space lumping in markov chains.
\newblock {\em Inf. Process. Lett.}, 87(6):309--315, 2003.

\bibitem[DMSW17]{dmsw17}
Ulrich Dorsch, Stefan Milius, Lutz Schr\"oder, and Thorsten Wi\ss\/mann.
\newblock Efficient coalgebraic partition refinement.
\newblock In Roland Meyer and Uwe Nestmann, editors, {\em 28th International
  Conference on Concurrency Theory (CONCUR 2017)}, volume~85 of {\em LIPIcs},
  pages 28:1--28:16. Schloss Dagstuhl, 2017.

\bibitem[DMSW19]{DMSW19}
Hans-Peter Deifel, Stefan Milius, Lutz Schröder, and Thorsten Wißmann.
\newblock Generic partition refinement and weighted tree automata.
\newblock In {\em Formal Methods, FM 2019}, LNCS. Springer, 2019.
\newblock To appear. Preprint available on arXiv at
  \url{https://arxiv.org/abs/1811.08850}.

\bibitem[DPP04]{DovierEA04}
Agostino Dovier, Carla Piazza, and Alberto Policriti.
\newblock An efficient algorithm for computing bisimulation equivalence.
\newblock {\em Theor. Comput. Sci.}, 311(1-3):221--256, 2004.

\bibitem[FV02]{FislerV02}
Kathi Fisler and Moshe Vardi.
\newblock Bisimulation minimization and symbolic model checking.
\newblock {\em Formal Methods in System Design}, 21(1):39--78, 2002.

\bibitem[GJKW17]{GrooteEA17}
Jan~Friso Groote, David~N. Jansen, Jeroen J.~A. Keiren, and Anton Wijs.
\newblock An \emph{O}(\emph{m}log\emph{n}) algorithm for computing stuttering
  equivalence and branching bisimulation.
\newblock {\em {ACM} Trans. Comput. Log.}, 18(2):13:1--13:34, 2017.

\bibitem[Gri73]{Gries1973}
David Gries.
\newblock Describing an algorithm by {H}opcroft.
\newblock {\em Acta Informatica}, 2:97--109, 1973.

\bibitem[GS01]{gs01}
Heinz-Peter Gumm and Tobias Schr\"oder.
\newblock Monoid-labelled transition systems.
\newblock In {\em Coalgebraic Methods in Computer Science, CMCS 2001},
  volume~44 of {\em ENTCS}, pages 185--204, 2001.

\bibitem[GVdV18]{GrooteEA18}
Jan~Friso Groote, Jao~Rivera Verduzco, and Erik~P. de~Vink.
\newblock An efficient algorithm to determine probabilistic bisimulation.
\newblock {\em Algorithms}, 11(9):131, 2018.

\bibitem[Han94]{Hansson:1994:TPF:561335}
Hans Hansson.
\newblock {\em Time and Probability in Formal Design of Distributed Systems}.
\newblock Elsevier, 1994.

\bibitem[HMM07]{HoegbergEA07}
Johanna H{\"{o}}gberg, Andreas Maletti, and Jonathan May.
\newblock Bisimulation minimisation for weighted tree automata.
\newblock In {\em Developments in Language Theory, {DLT} 2007}, volume 4588 of
  {\em LNCS}, pages 229--241. Springer, 2007.

\bibitem[HMM09]{HoegbergEA09}
Johanna H\"{o}gberg, Andreas Maletti, and Jonathan May.
\newblock Backward and forward bisimulation minimization of tree automata.
\newblock {\em Theor.\ Comput.\ Sci.}, 410:3539--3552, 2009.

\bibitem[Hop71]{Hopcroft71}
John Hopcroft.
\newblock An $n \log n$ algorithm for minimizing states in a finite automaton.
\newblock In {\em Theory of Machines and Computations}, pages 189--196.
  Academic Press, 1971.

\bibitem[HT92]{HuynhTian92}
Dung Huynh and Lu~Tian.
\newblock On some equivalence relations for probabilistic processes.
\newblock {\em Fund.\ Inform.}, 17:211--234, 1992.

\bibitem[Ihr03]{Gumm03}
Thomas Ihringer.
\newblock {\em Algemeine Algebra. Mit einem Anhang \"{u}ber Universelle
  Coalgebra von H.~P.~Gumm}, volume~10 of {\em Berliner Studienreihe zur
  Mathematik}.
\newblock Heldermann Verlag, 2003.

\bibitem[Jac17]{Jacobs17}
Bart Jacobs.
\newblock {\em Introduction to Coalgebras: Towards Mathematics of States and
  Observations}.
\newblock Cambridge University Press, 2017.

\bibitem[JR97]{JacobsR97}
Bart Jacobs and Jan Rutten.
\newblock A tutorial on (co)algebras and (co)induction.
\newblock {\em Bull.~EATCS}, 62:222--259, 1997.

\bibitem[KK14]{KonigKupper14}
Barbara K{\"{o}}nig and Sebastian K{\"{u}}pper.
\newblock Generic partition refinement algorithms for coalgebras and an
  instantiation to weighted automata.
\newblock In {\em Theoretical Computer Science, IFIP TCS 2014}, volume 8705 of
  {\em LNCS}, pages 311--325. Springer, 2014.

\bibitem[KKZJ07]{KatoenEA07}
Joost{-}Pieter Katoen, Tim Kemna, Ivan Zapreev, and David Jansen.
\newblock Bisimulation minimisation mostly speeds up probabilistic model
  checking.
\newblock In {\em Tools and Algorithms for the Construction and Analysis of
  Systems, {TACAS} 2007}, volume 4424 of {\em LNCS}, pages 87--101. Springer,
  2007.

\bibitem[Kli09]{Klin09}
Bartek Klin.
\newblock Structural operational semantics for weighted transition systems.
\newblock In Jens Palsberg, editor, {\em Semantics and Algebraic Specification:
  Essays Dedicated to Peter D.~Mosses on the Occasion of His 60th Birthday},
  volume 5700 of {\em LNCS}, pages 121--139. Springer, 2009.

\bibitem[Knu01]{Knuutila2001}
Timo Knuutila.
\newblock Re-describing an algorithm by {H}opcroft.
\newblock {\em Theor.\ Comput.\ Sci.}, 250:333 -- 363, 2001.

\bibitem[KS83]{KanellakisSmolka83}
Paris~C. Kanellakis and Scott~A. Smolka.
\newblock Ccs expressions, finite state processes, and three problems of
  equivalence.
\newblock In {\em Proceedings of the Second Annual ACM Symposium on Principles
  of Distributed Computing}, PODC '83, pages 228--240, New York, NY, USA, 1983.
  ACM.

\bibitem[KS90]{KanellakisS90}
Paris~C. Kanellakis and Scott~A. Smolka.
\newblock {CCS} expressions, finite state processes, and three problems of
  equivalence.
\newblock {\em Inf. Comput.}, 86(1):43--68, 1990.

\bibitem[LS91]{LarsenS91}
Kim~Guldstrand Larsen and Arne Skou.
\newblock Bisimulation through probabilistic testing.
\newblock {\em Inf.\ Comput.}, 94:1--28, 1991.

\bibitem[Mil80]{Milner80}
Robin Milner.
\newblock {\em A Calculus of Communicating Systems}, volume~92 of {\em LNCS}.
\newblock Springer, 1980.

\bibitem[MM92]{MacLane1992}
Saunders {Mac Lane} and Ieke Moerdijk.
\newblock {\em Sheaves in Geometry and Logic}.
\newblock Springer New York, 1992.

\bibitem[MPW19]{MPW19}
Stefan Milius, Dirk Pattinson, and Thorsten Wißmann.
\newblock A new foundation for finitary corecursion and iterative algebras,
  2019.
\newblock Submitted, preprint available on
  \url{https://arxiv.org/abs/1802.08070}.

\bibitem[Par81]{Park81}
David Park.
\newblock Concurrency and automata on infinite sequences.
\newblock In {\em Theoretical Computer Science, 5th GI-Conference}, volume 104
  of {\em LNCS}, pages 167--183. Springer, 1981.

\bibitem[PT87]{PaigeTarjan87}
Robert Paige and Robert~E. Tarjan.
\newblock Three partition refinement algorithms.
\newblock {\em SIAM J.~Comput.}, 16(6):973--989, 1987.

\bibitem[RT08]{RanzatoT08}
Francesco Ranzato and Francesco Tapparo.
\newblock Generalizing the {P}aige-{T}arjan algorithm by abstract
  interpretation.
\newblock {\em Inf.\ Comput.}, 206:620--651, 2008.

\bibitem[Rut00]{Rutten00}
Jan Rutten.
\newblock Universal coalgebra: a theory of systems.
\newblock {\em Theor.\ Comput.\ Sci.}, 249:3--80, 2000.

\bibitem[Seg95]{Segala95}
Roberto Segala.
\newblock {\em Modelling and Verification of Randomized Distributed Real-Time
  Systems}.
\newblock PhD thesis, Massachusetts Institute of Technology, 1995.

\bibitem[SP11]{SchroderPattinson11}
Lutz Schr{\"o}der and Dirk Pattinson.
\newblock Modular algorithms for heterogeneous modal logics via multi-sorted
  coalgebra.
\newblock {\em Math.\ Struct.\ Comput.\ Sci.}, 21(2):235--266, 2011.

\bibitem[SSvL{\etalchar{+}}11]{ShervashidzeEA2011}
Nino Shervashidze, Pascal Schweitzer, Erik van Leeuwen, Kurt Mehlhorn, and
  Karsten Borgwardt.
\newblock Weisfeiler-{L}ehman graph kernels.
\newblock {\em J.\ Mach.\ Learn.\ Res.}, 12:2539--2561, 2011.

\bibitem[Val09]{Valmari09}
Antti Valmari.
\newblock Bisimilarity minimization in {$\CO(m \log n)$} time.
\newblock In {\em Applications and Theory of Petri Nets, {PETRI} {NETS} 2009},
  volume 5606 of {\em LNCS}, pages 123--142. Springer, 2009.

\bibitem[vB77]{Benthem77}
Johann van Benthem.
\newblock {\em Modal Correspondence Theory}.
\newblock PhD thesis, Universiteit van Amsterdam, 1977.

\bibitem[vdMZ07]{MeydenZ07}
Ron van~der Meyden and Chenyi Zhang.
\newblock Algorithmic verification of noninterference properties.
\newblock In {\em Views on Designing Complex Architectures, VODCA 2006}, volume
  168 of {\em ENTCS}, pages 61--75. Elsevier, 2007.

\bibitem[VF10]{ValmariF10}
Antti Valmari and Giuliana Franceschinis.
\newblock Simple {$\CO(m\log n)$} time {M}arkov chain lumping.
\newblock In {\em Tools and Algorithms for the Construction and Analysis of
  Systems, TACAS 2010}, volume 6015 of {\em LNCS}, pages 38--52. Springer,
  2010.

\bibitem[Wei76]{Weisfeiler1976}
Boris Weisfeiler.
\newblock {\em On Construction and Identification of Graphs}.
\newblock Springer, 1976.

\bibitem[Wor05]{Worrell05}
James Worrell.
\newblock On the final sequence of a finitary set functor.
\newblock {\em Theor.\ Comput.\ Sci.}, 338:184--199, 2005.

\bibitem[ZHEJ08]{ZhangEA08}
Lijun Zhang, Holger Hermanns, Friedrich Eisenbrand, and David Jansen.
\newblock {F}low {F}aster: Efficient decision algorithms for probabilistic
  simulations.
\newblock {\em Log.\ Meth.\ Comput.\ Sci.}, 4(4), 2008.

\end{thebibliography}
